%% file: main_SPM.tex
\pdfoutput=1
\documentclass[runningheads]{llncs}

\usepackage{amsmath,amssymb}
\usepackage{xspace}
\usepackage{xcolor}
\usepackage{textcomp} %
\usepackage[most]{tcolorbox}
\usepackage{enumitem}
\usepackage{booktabs}

\usepackage{thmtools}
\usepackage{thm-restate}
\usepackage[nospace]{cite}
\usepackage{hyperref}
\usepackage{ulem}
\usepackage[expansion=false]{microtype}

\setlist{nosep, leftmargin=*}
\setlist{nosep, leftmargin=*}

\newcommand{\ia}{\textit{i}}
\newcommand{\ib}{\textit{ii}}
\newcommand{\ic}{\textit{iii}}

\newcommand{\calA}{\mathcal{A}}

\newcommand{\calF}{\mathcal{F}}

\newif\ifsubmitting
\submittingtrue
\ifsubmitting
  \newcommand{\udi}[1]{}
  \newcommand{\ohad}[1]{}
  \newcommand{\idit}[1]{}
\else
  \newcommand{\udi}[1]{\textcolor{olive}{[Udi: #1]}}
  \newcommand{\ohad}[1]{\textcolor{blue}{\textbf{[Ohad:} #1\textbf{]}}}
  \newcommand{\idit}[1]{\textcolor{teal}{[Idit: #1]}}
\fi

\newcommand{\temph}[1]{\textbf{#1}}

\newcommand{\mypara}[1]{\smallskip\noindent\textbf{#1.}}

\newif\iffull
\fulltrue
\newcommand{\full}[2]{\iffull#1\else#2\fi}
\newcommand{\fullversion}{the full version~\cite{eitan2026securing}}

\iffull\else
\makeatletter
\g@addto@macro\normalsize{%
  \abovedisplayskip 1\p@ \@plus 1\p@ \@minus 1\p@
  \belowdisplayskip \abovedisplayskip
  \abovedisplayshortskip 1\p@ \@plus 1\p@
  \belowdisplayshortskip 3\p@ \@plus 1\p@ \@minus 1\p@}
\def\@spthm#1#2#3#4{\topsep 1\p@ \@plus 1\p@ \@minus 1\p@
\refstepcounter{#1}%
\@ifnextchar[{\@spythm{#1}{#2}{#3}{#4}}{\@spxthm{#1}{#2}{#3}{#4}}}
\def\@Thm#1#2#3{\topsep 1\p@ \@plus 1\p@ \@minus 1\p@
\@ifnextchar[{\@Ythm{#1}{#2}{#3}}{\@Xthm{#1}{#2}{#3}}}
\renewcommand\section{\@startsection{section}{1}{\z@}%
                       {-6\p@ \@plus -2\p@ \@minus -2\p@}%
                       {2\p@ \@plus 2\p@ \@minus 2\p@}%
                       {\normalfont\large\bfseries\boldmath
                        \rightskip=\z@ \@plus 8em\pretolerance=10000 }}
\renewcommand\subsection{\@startsection{subsection}{2}{\z@}%
                       {-6\p@ \@plus -2\p@ \@minus -2\p@}%
                       {2\p@ \@plus 2\p@ \@minus 2\p@}%
                       {\normalfont\normalsize\bfseries\boldmath
                        \rightskip=\z@ \@plus 8em\pretolerance=10000 }}
\renewcommand\subsubsection{\@startsection{subsubsection}{3}{\z@}%
                       {-9\p@ \@plus -3\p@ \@minus -3\p@}%
                       {-0.5em \@plus -0.22em \@minus -0.1em}%
                       {\normalfont\normalsize\bfseries\boldmath}}
\renewcommand\paragraph{\@startsection{paragraph}{4}{\z@}%
                       {-8\p@ \@plus -3\p@ \@minus -3\p@}%
                       {-0.5em \@plus -0.22em \@minus -0.1em}%
                       {\normalfont\normalsize\itshape}}
\def\@listI{\leftmargin\leftmargini \parsep \z@ \@plus 1\p@
  \topsep 3\p@ \@plus 1\p@ \@minus 1\p@ \partopsep 1\p@
  \itemsep 1\p@ \@plus .5\p@ \@minus .5\p@}
\let\@listi\@listI \@listi
\makeatother
\fi

\begin{document}

\title{Securing People and their Machines Against Major Faults\full{ (Full Version)}{}}
\titlerunning{Securing People and Machines}

\author{Ohad Eitan\inst{1}\full{}{\orcidID{0009-0007-5651-6080}} \and
Idit Keidar\inst{1}\full{}{\orcidID{0000-0002-6417-1250}} \and \full{}{\\}
Ehud Shapiro\inst{2}\full{}{\orcidID{0009-0002-8266-3125}}
}

\authorrunning{Eitan et al.}
\institute{Technion -- Israel Institute of Technology, Haifa, Israel\full{}{\\
\email{ohad.eitan@campus.technion.ac.il}, \email{idish@technion.ac.il}} \and
London School of Economics and Political Science, UK and Weizmann Institute of Science, Israel\full{}{, 
\email{e.shapiro2@lse.ac.uk}}
}

\maketitle

\idit{Udi, please cut to 15 pages}

\begin{abstract}
We consider grassroots platforms---distributed systems of agents consisting of people identified by self-chosen public keys and their machines (smartphones)---and wish to make them secure against \emph{major faults}: the loss of their private keys and/or their smartphones.  As grassroots platforms have no global resource to rely on for recovery, our peer-based solution is based on: (\ia) \emph{a grassroots social graph} in which agents establish and maintain friendships; 
(\ib)  \emph{identity custodians}, designated by each person, and (\ic) \emph{state custodians}, which are grassroots platform-specific.  Upon a person experiencing identity loss, and given a willing supermajority of their identity custodians, the person's friends  replace the old public key with the new one across the graph and restore friendships, where all friends serve as state custodians for the social graph.

We specify the social graph and its secure version as guarded multiagent atomic transactions, and implement the secure social graph via communicating volitional agents, an asynchronous message-passing model one step closer to implementation. We prove the implementation maps runs with recoverable faults to correct runs of the specification. 

We follow a similar path for grassroots coins and bonds, showing a common core as well as the platform-specific aspects of state recovery: a currency's single-writer log is recovered exactly, the recovered sovereign resuming without double-spending.
\end{abstract}

\input{sections/introduction}

\input{sections/transition-systems}

\input{sections/recovery}

\input{sections/social-graph}

\input{sections/cva}

\input{sections/cva-implementation}

\input{sections/coins-implementation}

\input{sections/related-work}

\section*{Acknowledgements}
We are grateful to Sadegh Keshavarzi, Gregory Chockler, and Alexey Gotsman for fruitful discussions on trusted execution environments.

\bibliographystyle{splncs04}
\bibliography{bib}

\iffull
\appendix
\input{sections/appendix}
\input{sections/transition-systems-appendix}

\input{sections/appendix-gsg-cva}        %
\input{sections/appendix-coins-bonds}   %
\input{sections/appendix-proofs}

\fi

\end{document}

%% file: sections/introduction.tex
\section{Introduction}\label{section:intro}

\mypara{Grassroots platforms} The Internet today is dominated by global platforms: centralised---social networks, Internet commerce, `sharing-economy'---with autocratic control~\cite{zuboff2019age,zuboff2022surveillance},  and decentralised---blockchains and cryptocurrencies~\cite{ethereum:dao,faqir2021comparative,nakamoto2008peer,ethereum,wang2018overview,wang2021ethereum}---with plutocratic control~\cite{vitalikplutocracy}.
Grassroots platforms~\cite{shapiro2023grassrootsBA,shapiro2023gsn,shapiro2024gc,shapiro2025GF} aim to offer an egalitarian alternative.  Grassroots platforms can have multiple instances that emerge and operate independently of any global resource except the network, yet interoperate and coalesce once interconnected into ever-larger instances, possibly a single global one.  Key grassroots platforms include grassroots social networks~\cite{shapiro2023gsn,lewis2026volitional}, grassroots coins~\cite{shapiro2024gc,lewis2023grassroots,shapiro2026bonds}, and grassroots democratic federations~\cite{halpern2024federated,shapiro2025GF,keidar2025constitutional}.

\mypara{Recovery from major faults}
Grassroots platforms consist of people operating their personal machines (smartphones).  A person's identity is a self-chosen cryptographic key, and a platform's state resides only on the personal machines of its participants.  A participant may suffer a \emph{major fault}: the loss of their private key (\emph{identity loss}) and/or their machine state (\emph{state loss}).  As grassroots platforms have no central authority or global infrastructure to recover from, recovery must be peer-based; this is the subject of the paper.

\looseness=-1 Grassroots platforms have been specified via \emph{volitional multiagent atomic transactions}~\cite{lewis2026volitional}, each agent being a person-machine pair and each atomic transaction guarded by the volitions of the participating persons.  We use this foundation to specify recovery from major faults. 

Upon choosing a key, each person designates a set of trusted \emph{identity custodians}, together with a supermajority threshold; the custodians authorise replacement of the person's key should the key be lost or compromised.  Recovery is then guarded by a willing supermajority of a person's identity custodians, upon which the friends of the person replace the old public key with the new one across the social graph, preserving their friendships.  A fault that leaves the key intact---loss of machine state---needs no custodian authorisation: the person proves their identity by signing with the retained key, and the social graph is restored from the friends that still record it, with every friend serving as a \emph{state custodian} of the social graph.  These custodians are the person's real-world friends, known off-chain, so a recovering person can reach a subset of them out of band to initiate recovery.

The social graph is the foundational platform: it carries identities and recovers them.  Every higher platform is built upon it, delegating identity recovery to the social graph and adding only the recovery of its own state, held by platform-specific \emph{state custodians}; we give such a protocol for grassroots coins.  Recovery has a \emph{common core}: it is peer-based, carried by state custodians, delegates identity recovery to the social graph, is grassroots by construction, and is privacy-preserving.  Here we show two instances of this common core: the social graph and grassroots coins.

\mypara{Fault model}  We assume that every agent's machine runs the same attested protocol code inside a trusted execution environment~\cite{sabt2015trusted}, with mutual remote attestation among agents' machines~\cite{coker2011principles}.  An attested machine executes the protocol as written; an attacker controlling the host can still (1) roll the enclave back to an earlier state~\cite{matetic2017rote,keshavarzi2025tee}, (2) manipulate its clock, or (3) run several instances of the same enclave~\cite{briongos2023forking}. Our protocol is robust to these attacks, since (1) we store nothing outside the enclave, so there is no stored state to roll back to; (2) no transaction reads a clock; (3) the protocol is robust to duplication and reordering, so if messages are partitioned between multiple enclaves, the code still behaves correctly \full{(Lemma~\ref{lemma:gsg-cva-stale-inert})}{(\fullversion)}. Since the protocol overcomes state loss, it is also robust to message loss. An agent therefore fails only by losing its private key and/or its machine state, not by deviating from the protocol or misreporting its state. As in~\cite{keshavarzi2025tee},
recovery is peer-based--- it relies on the friends and custodians of a recovering agent reporting their state faithfully. Liveness requires that a supermajority of agents' identity custodians, and currency's state custodians, are functional. 
Every fault is eventually repaired, after a finite but unbounded delay, and a fault occurring during a recovery simply re-triggers it.  The results we state at quiescence are over runs with finitely many faults.
The only adversarial capability we admit is an attacker holding a person's private key, whose effect is confined to the compromised identity and ends once the key is replaced \full{(Appendix~\ref{section:secure-gsg-cva-replace})}{(\fullversion)}. 
Communication among correct agents 
is reliable, whereas agent failure may cause transient message loss to and from the failed agent. There is no bound on message latency, i.e., the system is asynchronous.

\mypara{Secure grassroots coins}
In grassroots coins~\cite{shapiro2024gc,shapiro2026bonds}, each agent is the \emph{sovereign} of their own currency: they mint coins---a coin with a future maturity date is a \emph{bond}---and every transaction in their coins requires their participation.  The Grassroots Flash payment system~\cite{lewis2023grassroots} realises this by having the sovereign approve each payment in their coins, with the sovereign's personal blockchain serving as the authoritative ledger for their currency.  This ledger is append-only and hence monotonic.  If the sovereign loses this ledger, they may inadvertently approve a double-spend---approving payment of a coin that was already paid to someone else.  Designated friends---\emph{state custodians}---participate in every transaction and thereby maintain up-to-date copies of the sovereign's ledger.  Since the ledger is monotonic, a single available state custodian suffices when updates are atomic; the asynchronous implementation instead couples finality to a supermajority of state custodians, recovering the log exactly \full{(Appendices~\ref{section:secure-bonds} and~\ref{section:secure-coins-cva})}{(\fullversion)}.

\mypara{Theft}  An identity thief can transact as the person at all grassroots platforms---befriending, unfriending, paying and minting.  As a practical urgent step, and prior to identity custodians willing key replacement, the victim may personally warn their friends and request them to withhold transactions with the stolen identity; they can abide, for a while, without violating the protocol.  Upon recovery, the person can restore any damage to their social graph. Recovering from final fraudulent transactions in other people's coins is impossible. However, recovering from fraudulent transactions in one's own currency  (i.e., not re-issuing them in the new currency) can be done with the approval of the old currency's custodians.  (Section~\ref{section:coins-impl-overview}).

\mypara{What we prove}  \looseness=-1 We develop each platform at three levels of abstraction: an abstract specification as guarded multiagent atomic transactions; a \emph{secure} specification, again as guarded transactions, adding recovery from major faults; and an implementation as communicating volitional agents (CVA), an asynchronous message-passing model one step closer to implementation.
For the social graph we prove that the secure specification---in which each agent also records its `friends-of-friends', for state recovery---is a fault-resilient implementation of the abstract one (Theorem~\ref{thm:ssg-resilient}), and that its CVA implementation realises the secure specification at quiescence (Theorem~\ref{theorem:secure-gsg-cva-refinement}), and hence the social graph (Corollary~\ref{corollary:secure-gsg-cva-implements-sg}): eventually every quiescent state maps correctly to a specification state.
At quiescence every friendship still recorded by a friend of either party is preserved; a friendship whose only records, at the moment of an identity loss, are the two friends themselves, is lost.
We further establish the soundness of every reachable implementation state, and that every level is grassroots.

We then follow a similar path for grassroots coins and bonds, 
showing a common core as well as the platform-specific aspects of state recovery: a totally-ordered, single-writer transaction log requires \emph{exact} recovery, obtained by coupling finality to a supermajority of state custodians, so the recovered sovereign resumes without double-spending.

\mypara{Paper outline}
Section~\ref{section:dts} recalls guarded multiagent atomic transactions, the formalism in which grassroots platforms and their recovery are specified, \full{with formal details deferred to Appendix~\ref{app:dts}}{with formal details in \fullversion}.
Section~\ref{section:recovery} sets out the recovery framework: identity recovery as the foundational service of the social graph, and platform-specific state recovery layered upon it.
Section~\ref{section:social-graph} specifies the social graph and the secure social graph, and proves that the secure social graph implements the social graph.
Section~\ref{sec:cva} introduces \emph{communicating volitional agents}: the specifications above are atomic multiagent transactions, which an implementation on networked smartphones cannot provide, so CVA refines them to an asynchronous message-passing model one step closer to implementation.
Section~\ref{section:cva-impl-overview} summarises the CVA implementation of the secure social graph, 
\full{with the full development and proofs in Appendix~\ref{section:gsg-cva}}{with the full development and proofs in \fullversion}.
Section~\ref{section:coins-impl-overview} follows a similar path for grassroots coins, the recovery framework's second instance, summarising their specification and CVA implementation.
Section~\ref{section:related-work} discusses related work.
\full{Appendix~\ref{appendix:persons} surveys formal models of persons in concurrent systems.  Appendices~\ref{section:secure-bonds} and~\ref{section:secure-coins-cva} develop 
secure grassroots coins, with their ordinary and secure specification as volitional multiagent atomic transactions and implementation as communicating volitional agents, with recoverable finality and exact recovery.}{A survey of formal models of persons in concurrent systems, the development of secure grassroots coins, and all proofs are in \fullversion.}

%% file: sections/transition-systems.tex
\section{Grassroots Guarded  Multiagent Atomic Transactions}\label{section:dts}

This section recalls guarded multiagent atomic transactions~\cite{lewis2026volitional}, a high-level formalism for specifying  grassroots platforms as systems of \emph{agents}, each consisting of a \emph{person} operating a \emph{machine} (e.g., smartphone).  A guarded atomic transaction specifies how the machine states of the participating agents change, and which of the persons (if any) participating in the transaction must be willing for the transaction to occur.

We assume a potentially infinite set of \emph{agents} $\Pi$, considering only finite subsets $P \subset \Pi$.  
We use $S^P$ to denote the set of all total functions from $P$ to $S$, and for $c \in S^P$ use $c_p$ for the member of $c$ indexed by $p$.

\begin{definition}[Machine State, Configuration, Transaction, Guarded Transaction]\label{definition:mt} 
Given an arbitrary set $S$ of \temph{machine states}, with a designated \temph{initial state} $s0 \in S$, and agents $Q \subset \Pi$, a \temph{machine configuration} over $Q$ is a member of $S^Q$, and a \temph{machine transaction} over \temph{participants} $Q$ is a pair $c\rightarrow c' \in (S^Q)^2$ such that $c\ne c'$. Given such a machine transaction $t$, a 
\temph{guarded transaction} over $t$ is a pair $(t,Q')$ where $Q'\subseteq Q$ are its \temph{guards}.
\end{definition}

Guarded atomic transactions 
can be carried out by their participants regardless of the states of non-participants.  Participants include both active agents (whose state changes) and stationary agents (whose state is a precondition but does not change).  When a transaction is ``guarded by $\{p,q\}$,'' both must be willing; when it is ``guarded by either $p$ or $q$,'' there are two guarded transactions over the same machine transaction, $(t,\{p\})$ and $(t,\{q\})$, so that either person's volition suffices.  

The operational semantics of guarded multiagent atomic transactions~\cite{lewis2026volitional} maps them to volitional multiagent atomic transactions~\cite{shapiro2025atomic}, \full{recalled in Appendix~\ref{app:dts}}{recalled in \fullversion}.  A guarded transaction is \emph{enabled} when each participant is in the precondition machine state and every guard wills it \full{(Definition~\ref{definition:enabled})}{(\fullversion)}; a transaction whose guard is empty is \emph{unguarded}, enabled whenever its machine precondition holds.

A set of multiagent atomic transactions in turn induces a multiagent transition system~\cite{shapiro2021multiagent}, a special case of ``standard'' transition systems. A family of multiagent transition systems, parametrised over the set of participating agents, defines a \emph{protocol}.
We are particularly interested in ~\emph{grassroots protocols}~\cite{shapiro2023grassrootsBA,shapiro2025atomic,lewis2026volitional}, defined informally as follows.  Runs of two disjoint sets can be \emph{interleaved} into a run of their union.  A protocol is \temph{oblivious} if every interleaving of two correct runs of disjoint sets of agents is a correct run of the union, and \temph{interactive} if some correct run of the union is not an interleaving of any two such independent runs, because it includes a step whose participants span both sets; it is \temph{grassroots} if both~\cite{lewis2026volitional}, capturing formally the informal idea that grassroots platforms (specified by grassroots protocols) can have multiple instances that can form and operate independently, yet may coalesce into ever-larger instances, possibly (but not necessarily) into a single global platform.

%% file: sections/recovery.tex
\section{Recovery of Identity and State after Major Faults}\label{section:recovery}

The social graph is the foundational platform: it carries identities and recovers them as well as its own state, and every higher platform is built upon it.  A higher platform does not recover identities itself; it delegates identity recovery to the social graph, and adds only the recovery of its own platform-specific state.

\mypara{Identity recovery} 
Upon choosing a key, each person designates a set of \emph{identity custodians} together with a supermajority threshold, fixed at identity creation.  Upon identity loss, the person chooses a fresh keypair off-chain, obtains a new machine if needed, and convinces a supermajority of their identity custodians to authorise the replacement; the friends of the person then replace the old public key with the new one across the social graph, preserving their friendships. 

A higher platform recovers an identity 
by delegating to the social graph: its agents are social-graph identities, and once the social graph has replaced a lost key, the higher platform's recovery proceeds against the renamed agent.

\mypara{State recovery}  Recovery from state loss is carried by \emph{state custodians}: peers that hold the state from which a recovering agent restores.  State recovery has three components on every platform: (\ia) the state custodians and what each holds; (\ib) the maintenance discipline that keeps the custodians' holdings current; and (\ic) the read-back by which a recovering agent reconstructs its lost state.
What is platform-specific is the filling of the three---who the custodians are, what discipline maintains their copies, and how a recovering agent reads its state back.  State recovery is privacy-preserving, in that the machine of a state custodian stores a friend's state without its person having access to it.

The approaches used in the social graph and grassroots coins are summarised in Table~\ref{table:recovery-instances}.
The secure social graph maintains friends-of-friends lists.
In the message-passing setting, direct friendships are updated via distributed transactions, whereas friend-of-friend lists are periodically disseminated in a lazy manner. 
At the abstract level, a single friend acting as state custodian is enough for recovering the agent's full friend list.  In the message-passing setting, each friend's periodic checkpoint restores the friendship it co-owns, and friend-of-friend lists are recovered in the background using the normal mode of operation. 
Note that even without failures,  lazy updates imply that the secure social graph is only \emph{eventually consistent}. Namely, it
\emph{converges} to the abstract graph (specified using atomic transactions) in quiescence, but may temporarily diverge from it while updates are taking place. These convergence semantics are preserved under failure and recovery. 

\begin{table}[t]
\footnotesize
\renewcommand{\arraystretch}{1.0}
\begin{tabular}{@{}p{0.20\textwidth}p{0.35\textwidth}@{\hskip 10mm}p{0.35\textwidth}@{}}
\toprule
  & \textbf{Social graph} & \textbf{Grassroots coins} \\
\midrule
 \textbf{Recoverable state} & friendships, each co-owned by two friends & the sovereign's transaction log, a single-writer total order \\
  \textbf{State \mbox{custodians}} & every friend & a designated subset of the sovereign's friends \\
 \textbf{Maintenance} & atomic friendship formation, lazy dissemination of friend-of-friend records & atomic log update across a super-majority of custodians \\
  \textbf{Read-back} & each friendship from its friend & from a supermajority of custodians, adopting the longest \\
\bottomrule
\end{tabular}
\caption{\textbf{State recovery:} a shared core filled differently by each platform.}
\label{table:recovery-instances}
\end{table}

%% file: sections/social-graph.tex
\section{The Social Graph}\label{section:social-graph}

The \emph{grassroots social graph} is the backbone of a grassroots platform: an evolving symmetric friendship relation among agents, each an autonomous person-machine pair.  We specify it as guarded multiagent atomic transactions (Section~\ref{section:dts}) at two levels of fault-tolerance.  The \emph{social graph} maintains friendships under fault-free operation.  The \emph{secure social graph} adds recovery from the loss, compromise, or corruption of a person's key or machine.  We then show that the secure social graph implements the social graph: under fault-free operation, and under recovery that does not lose all records of a friendship, the secure social graph exhibits exactly the friendships of the social graph at quiescence, except when a fault erases a friendship's last surviving record.

\subsection{The Social Graph}\label{section:gsg-abstract}

Each agent $p\in\Pi$ is a person-machine pair, identified by a self-chosen public key.  No further attributes are associated with an agent at this level.
The machine state of $p\in P$ is its \emph{friend set} $F_p\subseteq P$, initialised to $\emptyset$.

\begin{definition}[Social Graph Transactions~{\cite{lewis2026volitional}}]\label{definition:gsg-transactions}
The \temph{social graph transactions} are:
\begin{enumerate}
    \item \textbf{Befriend$(p,q)$}: $F'_p := F_p \cup \{q\}$, $F'_q := F_q \cup \{p\}$, provided $q \notin F_p$.  Guarded by $\{p,q\}$.
    \item \textbf{Unfriend$(p,q)$}: $F'_p := F_p \setminus \{q\}$, $F'_q := F_q \setminus \{p\}$, provided $q \in F_p$.  Guarded by either $p$ or $q$.
\end{enumerate}
\end{definition}

Befriending requires both persons to will the transaction; unfriending can be initiated by either. 
\full{In Appendix~\ref{appendix:proofs}}{In \fullversion}, we prove the following:

\begin{restatable}[Friendship Mutuality]{lemma}{ThmGsgMutuality}\label{lemma:gsg-mutuality}
In any run of the social graph, $q\in F_p \iff p\in F_q$ for all $p,q\in P$.
\end{restatable}

\begin{restatable}[The Social Graph is Grassroots]{theorem}{ThmGsgGrassroots}\label{theorem:gsg-grassroots}
The social graph is grassroots.
\end{restatable}

\subsection{The Secure Social Graph}\label{section:secure-gsg-abstract}

The secure social graph recovers from two kinds of major fault.  \emph{Identity loss} is the loss of a person's private key, possibly together with the machine, triggering a Replace transition.
\emph{State loss} is the loss of machine state with the key retained, which is fixed by a Recover transition.
Recovery is friendship-based, drawing only on the state held by the recovering agent's friends, with no global resource.

Two ingredients are added.  First, each agent carries an intrinsic, immutable \emph{identity record}: a finite nonempty set of \emph{identity custodians} and a supermajority threshold, fixed at identity creation.  A supermajority of the identity custodians authorises the replacement of a person's key.  Second, each agent records not only its own friends but, for each friend, that friend's friend set---so that a recovering agent can read its friendships back from any one friend.  Every friend thereby serves as a \emph{state custodian} of the agent's friendships.  Because the transactions are atomic, every friend holds the same, current record; one friend therefore suffices for recovery.

\mypara{Agents}  Each agent $p\in\Pi$ carries an intrinsic immutable \emph{identity record} $\mathit{IR}_p=(K_p,\sigma_p)$, where $K_p\subset\Pi$ is a finite nonempty set of \emph{identity custodians} and $\sigma_p\in(1/2,1]$ a \emph{supermajority threshold}.  A \emph{supermajority} of $K_p$ is a subset $G\subseteq K_p$ with $|G|\ge\lceil\sigma_p\cdot|K_p|\rceil$.  We assume each person befriends their identity custodians and keeps at least a supermajority of them as friends, so a supermajority of $K_p$ lies within $F_p$ and the guard of Replace can be chosen within its participants.
Immutability means that no transaction adds or removes a custodian or changes the threshold; Replace$(p,p')$ substitutes $p'$ for $p$ in the custodian sets naming $p$, as in every other record of $p$.

\mypara{Machine state}  The machine state of $p\in P$ is a pair $(\mathit{IR}_p,N_p)$: the immutable identity record $\mathit{IR}_p$ and the \emph{friend-of-friend map} $N_p:P\rightharpoonup 2^{P}$, a partial map in which $N_p(q)$ is $p$'s record of $q$'s friend set.  The \emph{friend set} $F_p:=\mathrm{dom}(N_p)$ is the set of agents $p$ records as friends; initially $N_p$ is the empty map, so $F_p=\emptyset$.

\begin{definition}[Secure Social Graph Transactions]\label{definition:secure-gsg-transactions}
The \temph{secure social graph transactions} are:
\begin{enumerate}
    \item \textbf{Befriend$(p,q)$}: with participants $\{p,q\}\cup F_p\cup F_q$, provided $q\notin F_p$:
    \begin{itemize}
        \item $F'_p := F_p\cup\{q\}$ and $F'_q := F_q\cup\{p\}$;
        \item $N'_p(q) := F'_q$ and $N'_q(p) := F'_p$;
        \item for each $r\in F_p$: $N'_r(p) := F'_p$; \quad for each $r\in F_q$: $N'_r(q) := F'_q$.
    \end{itemize}
    Guarded by $\{p,q\}$.

    \item \textbf{Unfriend$(p,q)$}: with participants $\{p,q\}\cup F_p\cup F_q$, provided $q\in F_p$:
    \begin{itemize}
        \item $F'_p := F_p\setminus\{q\}$ and $F'_q := F_q\setminus\{p\}$;
        \item for each $r\in F'_p$: $N'_r(p) := F'_p$; \quad for each $r\in F'_q$: $N'_r(q) := F'_q$.
    \end{itemize}
    Guarded by either $p$ or $q$.

    \item \textbf{Recover$(p,q)$}: with participants $\{p,q\}\cup N_q(p)$, unguarded, provided $p\in F_q$:
    \begin{itemize}
        \item $F'_p := N_q(p)$;
        \item $N'_p(r) := F_r$ for each $r\in N_q(p)$.
    \end{itemize}

    \item \textbf{Replace$(p,p')$}: with participants $\{p,p'\}\cup F_p\cup\bigcup_{r\in F_p}F_r$, provided $F_{p'}=\emptyset$ and $p'$ is a fresh agent not yet befriended by anyone:
    \begin{itemize}
        \item $F'_{p'} := F_p$ and $N'_{p'} := N_p$, and $F'_p := \emptyset$, $N'_p := {}$ the empty map;
        \item for each $r\in F_p$: $F'_r := (F_r\setminus\{p\})\cup\{p'\}$, and $N'_r(p') := N_r(p)$;
        \item for each $f\in F_p$ with $p\in K_f$: $K'_f := (K_f\setminus\{p\})\cup\{p'\}$;
        \item for each $r\in F_p$ and each $s\in F_r$ with $p\in F_s$: $N'_s(r) := F'_r$.
    \end{itemize}
    Guarded by a supermajority $G\subseteq K_p$ of $p$'s identity custodians.
\end{enumerate}
\end{definition}

\begin{definition}[State-Loss Fault]\label{def:state-loss-fault}
A \temph{state-loss fault} at $p$ sets $N_p$ to the empty map, so $F_p=\emptyset$, leaving $\mathit{IR}_p$ and the state of every other agent unchanged.
\end{definition}

The set of state-loss faults is disjoint from the secure social graph transactions (Definition~\ref{definition:secure-gsg-transactions}): a fault erases state rather than transforming it under the protocol.

Befriend and Unfriend update the friend-of-friend map atomically: whenever $p$'s friend set changes, every friend's record $N_r(p)$ is updated in the same transaction, so the map stays exact.  Recover rebuilds the friend set from a single friend's record $N_q(p)$, which---being exact and unchanged during the loss---is $p$'s pre-loss friend set.  Replace models identity loss: a supermajority of $p$'s identity custodians authorises substituting a fresh $p'$ for $p$ across the graph, $p'$ adopting $p$'s friendships and every friend rewriting $p$ to $p'$.  Setting $F'_p:=\emptyset$ retires the abandoned identity abstractly; it is not a write to the old machine, whose disk may still hold $p$'s data under the dead key.

\full{In Appendix~\ref{appendix:proofs}}{In \fullversion}, we prove the following:

\begin{restatable}[Friend-of-Friend Exactness]{lemma}{ThmSsgFofExact}\label{lem:ssg-fof-exact}
In any run of the secure social graph in which no state loss occurs, for every $p\in P, q\in F_p$: $N_p(q)=F_q$.
\end{restatable}

\begin{restatable}[The Secure Social Graph is Grassroots]{theorem}{ThmSecureGsgGrassroots}\label{theorem:secure-gsg-grassroots}
The secure social graph is grassroots.
\end{restatable}

\subsection{The Secure Social Graph Implements the Social Graph}\label{section:ssg-implements-sg}

We show that the secure social graph implements the social graph of Section~\ref{section:gsg-abstract} (Definition~\ref{definition:gsg-transactions}), \full{in the sense of Definition~\ref{def:implementation}}{in the sense of implementation of \fullversion}.  The one fault no friendship-based recovery can repair is a state loss in which two friends are the only ones recording their friendship, and both lose their state before either recovers. 

\begin{definition}[Friendship-Preserving Run]\label{def:friendship-preserving}
A friendship is \emph{unrecoverable} in a run of the secure social graph if a state loss erases its last surviving record, or an identity loss replaces one party's key while no third agent records the friendship; otherwise it is \emph{recoverable}.
A run is \emph{friendship-preserving} if every Recover$(p,q)$ restores all the recoverable friendships in 
$p$'s  pre-loss friend set, and every Replace$(p,p')$ reaches every recoverable friend of $p$. 
\end{definition}

\begin{definition}[Projection]\label{def:ssg-projection}
Let $\sigma$ map a secure-social-graph configuration $c$ to the social-graph configuration $\sigma(c)$ with the same agents, the same volitional states, and machine state
\[
  \sigma(c)_p := F_p \cup \{r\in P : p\in F_r\}
\]
for each $p$ (discarding $\mathit{IR}_p$ and $N_p$), with each replaced agent $q'$ relabelled to the original $q$ it replaced.
\end{definition}

The projected friend set of $p$ is the set of agents $p$ records as friends together with those that record $p$ as a friend.  This is symmetric by construction---$q\in\sigma(c)_p \iff q\in F_p \lor p\in F_q \iff p\in\sigma(c)_q$---so every projected configuration satisfies Friendship Mutuality (Lemma~\ref{lemma:gsg-mutuality}) and is a configuration the social graph reaches.  An agent $p$ that has lost its state, with $F_p=\emptyset$, projects to $\{r : p\in F_r\}$, the friendships still recorded at its friends; the projection thus reads a state-lost agent's friendships from the friends that survive, so a state loss and the Recover that follows it leave $\sigma(c)$ unchanged whenever some friend records the edge.  When two friends that are each other's only record both lose their state, neither records the edge, so it is absent on both sides of $\sigma(c)$---symmetric, the consistent state the social graph reaches by Unfriend.

\begin{restatable}[Implementation]{theorem}{ThmSsgImplements}\label{thm:ssg-implements-sg}
With the projection $\sigma$ of Definition~\ref{def:ssg-projection}, the secure social graph is an implementation of the social graph \full{in the sense of Definition~\ref{def:implementation}, correct and complete in the sense of Definition~\ref{def:implementation-properties}}{as defined in \fullversion}.
\end{restatable}

\begin{restatable}[Fault-Resilience]{theorem}{ThmSsgResilient}\label{thm:ssg-resilient}
With the projection $\sigma$ of Definition~\ref{def:ssg-projection} and $F$ the state-loss faults (Definition~\ref{def:state-loss-fault}), the secure social graph is an $F$-resilient implementation of the social graph \full{(Definition~\ref{def:fault-resilient})}{(\fullversion)}, restricted to friendship-preserving runs (Definition~\ref{def:friendship-preserving}).
\end{restatable}

\noindent With $F=\emptyset$ this is the correctness of Theorem~\ref{thm:ssg-implements-sg}; the theorem extends it to runs that perform state-loss faults and recover from them, every recoverable friendship being restored and the unrecoverable residue excluded by friendship-preservation.

%% file: sections/cva.tex
\section{Communicating Volitional Agents}\label{sec:cva}

The framework of Section~\ref{section:dts} abstracts over \emph{how} a guarded multiagent atomic transaction is carried out: a $k$-ary transaction simultaneously updates the machine states of all its participants, with no notion of how participating machines come to act in concert.  An implementation on networked smartphones has no such simultaneity: every exchange between two agents is a separate asynchronous message-passing event, and a $k$-ary transaction must be realised by a protocol composed of such events.

Here we present \temph{communicating volitional agents} (CVA)~\cite{lewis2026volitional}, an intermediary level between the abstract guarded-transaction specification and its implementation on networked smartphones.  A CVA protocol's guarded machine transactions are drawn from four syntactic forms only: binary \temph{discover} transactions by which an agent comes to know of another; binary \temph{communicate} transactions that copy a message from one agent's outbox to another's inbox; 
and unary \temph{platform} transactions.  Thus, a CVA platform need only specify its platform transactions, changing an agent's platform-specific state.  Moreover a CVA protocol is grassroots by construction (Theorem~\ref{thm:cva-grassroots}).

\mypara{Local states} A CVA agent's local state packages five components: a set of \emph{known peers} with whom the agent can communicate, an \emph{outbox} of messages awaiting delivery, an \emph{inbox} of messages received, and a \emph{platform state} in which each platform specifies whatever data it requires (friend sets, groups, bonds, feeds, and so on).
A \emph{message} is a triple of sender, recipient, and cargo; outboxes and inboxes are sets of messages.  The cargo space $C$ and the platform state space $A$ are parameters of the CVA platform.

The formal definitions---CVA local states and configurations, and the CVA protocol, comprising Discover, Communicate, and the platform transactions---are \full{recalled in Appendix~\ref{app:dts}, together with the Known-Peers and Outbox Containment lemmas and the following, proven in Appendix~\ref{appendix:proofs}}{recalled in \fullversion, together with the Known-Peers and Outbox Containment lemmas and the following}:
\begin{restatable}[CVA Grassroots]{theorem}{ThmCvaGrassroots}\label{thm:cva-grassroots}
Every CVA protocol is grassroots.
\end{restatable}

%% file: sections/cva-implementation.tex
\section{CVA Implementation of the Secure Social Graph}\label{section:cva-impl-overview}

We implement the secure social graph as communicating volitional agents (Section~\ref{sec:cva}): the abstract Befriend, Unfriend, StateLoss, Recover, and Replace transactions of Section~\ref{section:secure-gsg-abstract} are realised over CVA's asynchronous message discipline, with no global resource.  This section summarises the construction and states its main results; the full development---platform state, transaction specifications, supporting invariants, and proofs---\full{is given in Appendix~\ref{section:gsg-cva}}{is given in \fullversion}.

\mypara{Construction}  A friendship is encoded by an epoch counter whose parity gives its status: an odd epoch is active, an even epoch inactive.  Befriend is an offer/accept handshake at a fresh odd epoch; Unfriend is unilateral, advancing the epoch to the next even value; and every change an agent makes to one of its friendships is disseminated to its friends as a stream update, so that friend-of-friend views converge. 

On top of this friendship machinery the security layer adds each agent's intrinsic, immutable identity record---its custodian set and supermajority threshold---which is carried on Befriend and stored by each friend alongside the friendship.

\mypara{One recovery mechanism}  Both faults are repaired by a single mechanism, periodic re-broadcast and absorption, with identity loss recovery prepending one step.  Each agent periodically re-broadcasts a snapshot of its friendships, carrying its identity record, to its friends, and a recipient absorbs it monotonically.  A recovered
agent that lost its machine state but kept its key is still recorded by its friends, so their checkpoints keep arriving and it rebuilds passively \full{(Restore, Appendix~\ref{section:secure-gsg-cva-restore})}{(Restore, \fullversion)}---state recovery with no further machinery.  An agent that lost its key recovers as a fresh identity $p'$ that no friend yet records; the one extra step is to \emph{install} $p'$ at the friends, a custodian-authorised rename of $p$ to $p'$ (the Replace cascade below), after which $p'$ is in the position of a state-lost agent and rebuilds by the same absorption.  State loss is thus the special case in which the install is a no-op, the key being retained; identity loss is the case in which the install is custodian-authorised.

\mypara{Replace: installing the recovered identity}  Identity loss is repaired by the Replace cascade \full{(Appendix~\ref{section:secure-gsg-cva-replace})}{(\fullversion)}, the conceptual heart of the construction.  It proceeds in three phases.  Phase~1 is off-protocol: the person behind the lost agent $p$ chooses a fresh agent $p'$ and convinces a supermajority of $p$'s custodians to authorise the replacement.  Phase~2 is on-protocol vouching: each willing custodian sends $p'$ a $\mathsf{vouch}$ carrying its own view of $p$'s friends, bootstrapping $p'$'s knowledge of whom to notify.  Phase~3 is the cascade: once a supermajority has vouched, $p'$ sends $\mathsf{new\_identity}$ to every friend of $p$ it has learned of; each recipient verifies the vouching set against the identity record it stored for $p$ on befriend, renames $p$ to $p'$ preserving the epoch, and returns a $\mathsf{rebind}$ carrying its own view of $p$'s friends---which may reveal friends no voucher named.  The cascade thus reaches every friend of $p$ recorded by a custodian or an already-reached friend.  
Verification is local: each friend checks the vouchers against the immutable identity record it received on befriend, with no global custodian registry.

\mypara{Implementation}  The implementation realises the secure social graph specification at the granularity of completed protocols at quiescence \full{(Definition~\ref{definition:gsg-cva-quiescence})}{(\fullversion)}, and hence realises the social graph.  A run recovering from major faults reaches such a quiescent configuration under eventual delivery, the periodic re-broadcast \full{(Appendix~\ref{section:secure-gsg-cva-rebroadcast})}{(\fullversion)} eventually delivering every checkpoint.  
The abstract friend set $\widetilde F_p$ a configuration is read against \full{(Definition~\ref{definition:secure-gsg-cva-abstract-friend-set})}{(\fullversion)} is the set of friendships of $p$ as recorded in $p$'s state. 

\full{The following are proven in Appendix~\ref{appendix:proofs}}{The following are proven in \fullversion}:

\begin{restatable}[Secure Quiescent Correspondence]{theorem}{ThmSecureCorrespondence}\label{theorem:secure-gsg-cva-refinement}
Each completed transaction protocol of the secure social graph CVA implementation realises, under the mapping $F_p\mapsto\widetilde F_p$, the corresponding transaction of the abstract secure social graph (Section~\ref{section:secure-gsg-abstract}); the correspondence is at the granularity of completed protocols at quiescence, not a step-by-step bisimulation:
\begin{enumerate}
    \item A completed Befriend handshake in CVA realises the abstract Befriend$(p,q)$.
    \item A completed Unfriend in CVA realises the abstract Unfriend$(p,q)$.
    \item A completed Replace cascade in CVA (Vouch, Announce new identity, Integrate new identity, Integrate rebind) realises the abstract Replace$(p,p')$ on every friendship it reaches; on a friendship-preserving run (Definition~\ref{def:friendship-preserving}) it reaches every recoverable friend.  The one unrecoverable friendship---recorded, at the fault, only by the two friends, the identity of one then lost---is not recovered; the specification does not require its recovery (Definition~\ref{def:friendship-preserving}).
    \item A completed Restore in CVA, after a state loss at $p$, realises the abstract Recover$(p,q)$: it sets $\widetilde F_p$ to the friends that still record $p$, which at quiescence is the full set of friendships recorded at $p$ before the state loss.
\end{enumerate}
\end{restatable}

\begin{restatable}[The CVA Implementation Securely Realises the Social Graph]{corollary}{ThmSecureImplementsSG}\label{corollary:secure-gsg-cva-implements-sg}
The secure social graph CVA implementation securely realises the social graph (Section~\ref{section:gsg-abstract}), composing the Secure Quiescent Correspondence with the implementation of the social graph by the secure social graph (Theorem~\ref{thm:ssg-implements-sg}) and transitivity of implementation, on friendship-preserving runs at quiescence.
\end{restatable}

\mypara{Safety} 
Abstract-state correspondence is guaranteed only at quiescence (Theorem~\ref{theorem:secure-gsg-cva-refinement}), so the simulation is a liveness property.  Safety, by contrast, holds at every reachable state, and continues to hold under finitely many failures, which the periodic re-broadcast of checkpoints repairs \full{(Lemma~\ref{lemma:secure-gsg-cva-repair})}{(\fullversion)}: we guarantee \emph{soundness} at all times---a friendship $(p,q)$ is reported only if both parties at some point wanted to befriend each other \full{(Appendix~\ref{section:secure-gsg-cva-properties})}{(\fullversion)}.

\mypara{The unrecoverable fault}  Recovery is exact except on the one information-theoretically unrecoverable condition (Definition~\ref{def:friendship-preserving}\full{; Lemma~\ref{lemma:secure-gsg-cva-irrecoverable}}{}): a major fault erases the last surviving record of a friendship. 
This has two instances.  Under \emph{state loss}, both friends lose their state before recovery: no friend still records the edge, so nothing remains to recover from (the friendship-preserving exclusion of Section~\ref{section:ssg-implements-sg}).  Under \emph{identity loss}, one friend's key is replaced and the edge was disseminated to no other agent, so the Replace cascade never reaches the surviving friend \full{(Appendix~\ref{section:secure-gsg-cva-replace})}{(\fullversion)}.  
The specification does not require the recovery of an unrecoverable friendship (Definition~\ref{def:friendship-preserving}).  The person behind the surviving friend may re-establish it out of band.

%% file: sections/coins-implementation.tex
\section{CVA Implementation of Secure Grassroots Coins}\label{section:coins-impl-overview}

The second instance of the recovery framework (Section~\ref{section:recovery}) is grassroots coins~\cite{shapiro2024gc,lewis2023grassroots,shapiro2026bonds}: each agent is the \emph{sovereign} of its own currency, mints its own coins---a coin with a future maturity date is a \emph{bond}---and authorises every transaction in them.  Where the social graph's recoverable state is friendships, each co-owned by two friends and recovered lazily, a currency's recoverable state is the sovereign's \emph{transaction log}: a totally ordered record under a single authoritative writer, in which every recorded payment matters, demanding \emph{exact} recovery.  We summarise the specification and its CVA implementation; the full development, with proofs, \full{is in Appendix~\ref{section:secure-bonds} (specification) and Appendix~\ref{section:secure-coins-cva} (implementation)}{is in \fullversion}.

\mypara{Specification}  Each sovereign $p$ maintains an append-only \emph{transaction log} $L_p$ recording every transaction in $p$-coins---Mint, Pay, Redeem, Swap---following Grassroots Flash~\cite{lewis2023grassroots}, where the sovereign approves each payment and its personal blockchain is the authoritative ledger.  To recover from a fault, $p$ designates a set of its friends as \emph{state custodians} at the formation of its currency, fixed and immutable thereafter, each holding a copy of $L_p$.  Because the transactions are atomic, every custodian's copy is updated in the same step as the sovereign's, so all agree and a single custodian suffices to recover the log, from which every $p$-coin holding is derived \full{(Appendix~\ref{section:secure-bonds})}{(\fullversion)}.

\mypara{From atomicity to a supermajority}  An asynchronous implementation has no such atomicity: a custodian's copy may lag the sovereign's.  Recovering from a single lagging custodian could reinstate a log that omits an already-approved payment, and the recovered sovereign might then approve the same coin a second time---an inadvertent double-spend.  We instead restore exactness by coupling finality to recoverability through a supermajority: a payment is \emph{final}, and a payee may rely on it, only once a supermajority of the sovereign's state custodians hold its block. 

To facilitate recovery, the sovereign proceeds through monotonically increasing \emph{epochs}, with each log block carrying the sovereign's epoch at its creation. In addition, each custodian keeps the latest epoch number proposed by the sovereign, initially $1$. A recovering sovereign has no memory of the epochs its currency has used and learns it from its custodians. Upon recovery, it first queries them for the epoch they have recorded and for their copies of the log, and waits for a supermajority of replies. It then sets its new epoch to exceed the largest epoch it heard, proposes the new epoch number to the custodians, and waits for a supermajority of acks, indicating that a supermajority of custodians have recorded the new epoch. Next, the sovereign
selects the maximum epoch that appears in blocks in the log replicas it received, and adopts the log copy that has the most blocks in this epoch.  (Note that the maximal epoch in the log might be lower than the maximum proposed epoch number). Finally, the sovereign sends the adopted log to all custodians and resumes operation. 

Custodians, in turn, align their logs with the one adopted by the sovereign, which may involve discarding blocks from earlier epochs. Note that such abandoned blocks never reached a supermajority (and hence were not recovered in an earlier recovery), they are not final. 
Since a supermajority records the new epoch before it is used to create blocks and any two supermajorities intersect, every future recovery is guaranteed discover this epoch number and initiate a higher one. Thus, epochs in the log increase monotonically, and within the latest epoch, all copies  are prefixes of one log. 

Any two supermajorities of the same fixed set intersect, so every final payment is held by some custodian the recovery reaches and is never lost.  This is the supermajority-intersection structure of All-to-All Flash~\cite{lewispye2023flash}, here in a simpler crash-only setting with a single log writer, so the quorums are plain intersecting read and write quorums and no consensus is invoked.

\mypara{Construction}  The protocol realises each spec transaction as a sequence of unary CVA transactions: a holder requests a payment, the sovereign approves it---checking the coins are unspent in its log---appends a block and disseminates it to its custodians, each custodian appends the block in order and acknowledges, and the sovereign binds finality once a supermajority have acknowledged, notifying the payee, who may then rely on the payment.  A holding is debited from the payer at request and credited to the payee at finality; the authoritative record is the sovereigns' logs, of which each agent's local holdings are a view.

\mypara{Recovery}  State custodians, like the identity record, are intrinsic: the set $S_p$ and the threshold survive a fault, so a recovering sovereign knows from whom to collect before it has recovered any log. State loss is handled by adopting the longest ledger in the latest epoch, as described above. Under identity loss, the sovereign is first replaced by a fresh agent $p'$ through the social graph (Section~\ref{section:cva-impl-overview}), which inherits the intrinsic custodian set and then collects from a supermajority exactly as under state loss; the social graph recovers the identity, the currency layer adds only the recovery of its own log.  A holder that loses its local account recovers it from each friend sovereign, which answers from its log; no custodian is involved.  Recovery presumes that a supermajority of a sovereign's state custodians do not concurrently lose state within a recovery window; with attestation and no persistent storage, this bounds the failures tolerated while a recovery completes.

\mypara{Recovery restores the ledger, not the value}  Identity recovery restores a sovereign's \emph{log}, not the \emph{value} of its currency, which rests on the integrity and solvency of its sovereign.  A key lost without abuse leaves the currency intact: the recovered sovereign resumes it unchanged.  A key held by an attacker who paid or minted fraudulently leaves valid, final transactions in the log---a faithfully recorded debt---to the extent of which the sovereign is insolvent.  Finality is not revoked: the recovered log stands, and every payment a payee relied upon remains final.  However, the sovereign may mitigate the loss in its own currency: any mitigation a supermajority of its state custodians approve is possible, and it is up to the sovereign to convince them.  For example, the sovereign may reissue the currency from a final prefix of the log not containing fraudulent transactions.

\mypara{The irrecoverable loss is in other currencies}  Reissue is a remedy only over a sovereign's own currency.  A key compromise does its principal damage elsewhere: the attacker transfers the victim's holdings in other sovereigns' coins to third parties, each payment final at its issuing sovereign, against whom the victim is an ordinary holder.  These transfers cannot be undone---recovery restores the person's identity and the ledgers they issue, not wealth alienated in currencies they do not control.

\mypara{Implementation and recovery results}  The implementation realises the secure coins specification at quiescence, and recovers the log exactly across a fault.

\begin{restatable}[Correspondence at Quiescence]{theorem}{ThmCorrespondenceQ}\label{theorem:secure-coins-cva-correspondence}
At a quiescent configuration \full{(Definition~\ref{definition:gsg-cva-quiescence})}{(\fullversion)} of any run of the CVA protocol in which every custodian has registered the current epoch, the abstract state read off the configuration is a reachable state of the secure coins specification \full{(Appendix~\ref{section:secure-bonds})}{(\fullversion)}: every custodian copy equals its sovereign's log, every log block is final, and each agent's holdings equal the holdings derived from the logs.
\end{restatable}

The correspondence is at quiescence, not at every state, and necessarily so: between a payment's request and its finalisation a coin is debited from the payer and not yet credited to any payee, a configuration the atomic specification never exhibits.

\begin{restatable}[Conservation of Money]{lemma}{ThmCoinsConservation}\label{lemma:secure-coins-cva-conservation}
At every quiescent configuration of a run with no state-loss fault, the $s$-coins across all holdings are exactly those recorded as minted in $L_s$.
\end{restatable}

\begin{restatable}[Exact Recovery]{theorem}{ThmExactRecovery}\label{theorem:secure-coins-cva-exact-recovery}
If a sovereign recovers---under state loss, or as $p'$ under identity loss---by adopting the longest of the log copies in the maximal epoch collected from a supermajority $G\subseteq S_p$, then the recovered log contains every block of $L_p$ that was final at any time before recovery.
\end{restatable}

\begin{restatable}[No Inadvertent Double-Spend]{corollary}{ThmNoDoubleSpend}\label{corollary:secure-coins-cva-no-double-spend}
After recovery, the sovereign never approves a payment of a coin already spent in a final payment.
\end{restatable}

\mypara{Liveness}  In a correct run, a requested payment whose sovereign and a supermajority of its custodians stay live becomes final.

\begin{restatable}[Payment Liveness]{theorem}{ThmPaymentLiveness}\label{theorem:secure-coins-cva-payment-liveness}
In a correct run in which the sovereign $s$ does not suffer a state loss after receiving a valid pay request and a supermajority of $S_s$ registered in $s$'s current epoch does not suffer a state loss, the requested payment eventually becomes final and its payee receives notification of finality.
\end{restatable}

The platform is grassroots, like every CVA protocol (Theorem~\ref{thm:cva-grassroots}); finality soundness, log validity, and the prefix consistency of custodian copies, together with the proofs of the results above, \full{are in Appendix~\ref{section:secure-coins-cva-properties}}{are in \fullversion}.

%% file: sections/related-work.tex
\section{Related Work}\label{section:related-work}

\mypara{Volitional multiagent atomic transactions}
\looseness=-1 The formal framework for grassroots platforms used in this paper---agents consisting of people and machines, with machine transactions guarded by people's volitions---was introduced in~\cite{lewis2026volitional}.  That work decomposes each agent's state into a machine state and a volitional state, and introduces equivalence classes of machine transactions and their consumption upon satisfaction; a full survey of the prior formal tradition of modelling people in concurrent systems, and a discussion of how volitional multiagent atomic transactions depart from it, appears there \full{(see also Appendix~\ref{appendix:persons})}{(see also \fullversion)}.  The contribution of the present work is to extend that framework to cover major faults: identity loss---the loss of a person's private key and/or machine---and state loss---the loss of machine state with the key retained.

\mypara{Key recovery and social recovery of identity}
\looseness=-1 Recovering a cryptographic identity after the loss or compromise of a private key is a long-standing problem.  Classical solutions rely on threshold cryptography---Shamir's secret sharing~\cite{shamir1979share} distributes a secret among custodians such that a threshold subset can reconstruct it.  A distinct approach, which underlies the present work, does not share the secret at all: instead, the person chooses a fresh keypair off-chain and a designated set of friends collectively authorise the substitution of the new public key for the old one across the person's social context.  Shahaf et al.~\cite{shahaf2020genuine} employ such a mechanism for identifier recovery via mutual sureties.  Similar designs appear in social-recovery wallets in the cryptocurrency space, where designated guardians authorise a wallet's key change~\cite{buterin2021social}. For example, Maram, Kelkar and Eyal~\cite{maram2024multicredential} and Mouallem and Eyal~\cite{mouallem2024asynchronous} decide which of two parties presenting credentials is their true owner, assuming the credentials still exist.  In 
contrast, we consider major faults where the key and device are gone and nothing can be presented, so recovery must be performed by the person's custodians rather than by the person.  Here, we give a formal specification of friendship-based major-fault recovery as a guarded multiagent atomic transaction over the social graph, and extend the same approach to the recovery of higher-level state (the sovereign's transaction log in grassroots coins) via state custodians.

\mypara{Comparison}
\looseness=-1 Deployed systems in the same design space identify a person with a keypair and offer no protocol for recovery from its loss: in Secure Scuttlebutt~\cite{tarr2019secure} and Nostr~\cite{nostr2024nip01} a lost key is a lost identity, and its social connections can be re-established only out of band.  Threshold schemes~\cite{shamir1979share} reconstruct the lost secret from shares; social-recovery wallets~\cite{buterin2021social} replace the key but record the change on a blockchain, a global resource.  Here the person assumes a fresh key, authorised by custodians, and the rename is carried by the social graph itself, with no global resource.

%% file: sections/appendix.tex
\section{Formal Models of Persons in Concurrent Systems}\label{appendix:persons}

The formal methods tradition has a long lineage of modeling human agents as sources of nondeterminism alongside deterministic machines, but---to the best of our knowledge---without decomposing an agent into person and machine as components of its state.

\mypara{Turing's choice machines}
Before defining what we now call Turing machines (automatic machines, or a-machines), Turing~\cite{turing1936computable} introduced \emph{choice machines} (c-machines), ``whose motion is only partially determined by the configuration''---at designated states, the machine ``cannot go on until some arbitrary choice has been made by an external operator.''  The external operator is a person who freely chooses between alternatives; the sequence of choices determines which computation unfolds.  Turing immediately set c-machines aside, showing that any c-machine computation can be enumerated by an a-machine.  His 1939 oracle machines~\cite{turing1939systems} extend this further: the oracle ``cannot be a machine'' and provides answers the computation cannot derive internally.  Both formalisms model the person as \emph{external} to the machine, providing input at designated points.

\mypara{Process algebras}
Hoare's CSP~\cite{hoare1985communicating} provides the cleanest process-algebraic encoding of human choice.  External choice ($\Box$) offers the environment---potentially a person---a selection among initial events, while internal choice ($\sqcap$) is resolved by the system.  However, the person remains \emph{outside} the system boundary: CSP models what the person \emph{does}, not what the person \emph{is willing to do}.
Milner's CCS~\cite{milner1980calculus} uses a single summation operator without formally separating internal from external nondeterminism at the syntactic level.  In his later work on bigraphs~\cite{milner2009space}, Milner makes the scope explicit: agents ``can be artificial, as in computing systems\ldots\ or they can be natural, e.g.\ communicating humans.''  Both CSP and CCS model agents uniformly---there is no formal distinction between a person and a machine within the same agent.

\mypara{I/O automata and reactive systems}
Lynch and Tuttle's I/O automata~\cite{lynch1989introduction} partition actions into input (environment-controlled), output (automaton-controlled), and internal actions.  The key property is \emph{input-enabling}: an automaton cannot block input actions, so the environment---potentially a human operator---can act at any moment.  The Hybrid I/O Automata extension~\cite{lynch2003hybrid} explicitly states that HIOAs are ``intended to model all components of hybrid systems, including\ldots\ humans.''  The person, however, is part of the environment, not a component of the automaton's state.
Harel and Pnueli's reactive systems paradigm~\cite{harel1985development} draws the foundational dichotomy between transformational systems (batch, terminating) and reactive systems (ongoing interaction with environment).  The system is deterministic; all nondeterminism is attributed to the environment.  This paradigm was explicitly motivated by human-machine interaction, yet the formalism treats the human as environment rather than as a component of the system.

\mypara{Angelic and demonic nondeterminism}
The distinction between angelic and demonic nondeterminism, developed by Back and von Wright~\cite{back1998refinement} in the refinement calculus, provides a semantic treatment relevant to the person/machine boundary.  Demonic nondeterminism models adversarial environments (the worst-case choice is made); angelic nondeterminism models cooperative choices (the best-case choice is made).  This duality is the closest precursor to the volition/obligation distinction in volitional multiagent atomic transactions~\cite{lewis2026volitional}: a volitional transaction guarded by both parties (both must be willing) versus one guarded by either party (either can force it).  Two distinctions separate the frameworks.  First, the refinement calculus operates within a sequential program framework, not a multiagent transition system.  Second, angelic nondeterminism is a point-of-choice semantics: a choice is resolved locally at each transition, with no residue carried forward.  Volitions, in contrast, are persistent, inspectable state that accumulates across transitions; a guard condition reads an agent's record of willing over the history of the run, not a single local resolution.  This shift---from choice as a point-semantic primitive to choice as state---is what lets volitions be shared, compared, and reasoned about within the transition system, rather than external to it.

\mypara{Game structures and alternating-time temporal logic}
Module checking~\cite{kupferman2001module} models open systems with the environment fully adversarial.  Alternating-time temporal logic (ATL)~\cite{alur2002alternating} interprets formulas over concurrent game structures where multiple agents simultaneously choose actions.  Game semantics~\cite{abramsky2000full} models computation as dialogue between Proponent (program, following a deterministic strategy) and Opponent (environment, making free moves).  These frameworks treat agents as symmetric players but do not decompose a single agent into person and machine components.

\mypara{Ceremony analysis and human-interactive verification}
Ellison's ceremony analysis~\cite{ellison2007ceremony} extends security protocol analysis to include human participants as protocol nodes.  Bolton's Enhanced Operator Function Model (EOFM)~\cite{bolton2013formally} translates hierarchical human task models into state machines for model checking, with the human as the sole source of nondeterminism.  Both treat human nondeterminism as a source of error to be verified against, rather than as a source of legitimate volition to be formally recorded.

\mypara{Normative multiagent systems and electronic institutions}
Electronic institutions~\cite{esteva2001formal} model multiagent interaction as dialogical frameworks where human and software agents are treated uniformly as role-playing entities.  Normative multiagent systems~\cite{artikis2009specifying} use Event Calculus to specify societies where agents ``may fail to, or even choose not to, conform to the specifications.''  These approaches model norms that constrain agents, but do not decompose an agent's state into machine and volitional components, nor do they formalise the distinction between transactions requiring all parties to be willing and those that are obligatory once initiated.

%% file: sections/transition-systems-appendix.tex
\section{Guarded Multiagent Atomic Transactions}\label{app:dts}

This section recalls essential definitions of the theory of volitional multiagent atomic transactions~\cite{lewis2026volitional}, which provides a formal foundation for specifying grassroots platforms as systems of people operating machines.  The full mathematical development---transition systems, runs, liveness, closure, grassroots protocols, and proofs---appears in~\cite{lewis2026volitional}; here we recall only what is needed for the specifications that follow.

Each agent consists of a \emph{person} and a \emph{machine} (e.g., smartphone) operated by the person.  A \emph{machine transaction} specifies what machines do; a \emph{guard} specifies which persons must be willing for the transaction to occur.

We assume a potentially infinite set of \emph{agents} $\Pi$, considering only finite subsets $P \subset \Pi$.  We use $\subset$ for the strict subset relation and $\subseteq$ when equality is also possible.  We use $S^P$ to denote the set of all total functions from $P$ to $S$, and $c_p$ for the member of $c \in S^P$ indexed by $p$.

Machine states, configurations, machine transactions, and guarded transactions are as in Definition~\ref{definition:mt} (Section~\ref{section:dts}), which also recalls the reading of guards and of the unguarded case.

\mypara{Notation for transaction specifications}  Protocol specifications in the sections that follow describe a machine transaction $c\to c'$ (Definition~\ref{definition:mt}) by naming its participants $Q$, giving a precondition on $c$ introduced by ``provided,'' and specifying the updates that produce $c'$ from $c$.  Where the machine state has a short name we use primed notation: ``$F'_p := F_p\cup\{q\}$'' means $c'_p = c_p\cup\{q\}$, reading $F_p$ as $c_p$ and $F'_p$ as $c'_p$. Where the machine state is structured, as in CVA (Section~\ref{sec:cva}), we use a procedural update: ``add $\mathsf{message}(\ldots)$ to $o_p$,'' ``set $a_p.\mathrm{FMap}[q]:=\ldots$,'' each understood as modifying only the component named.  In both styles, components outside the participants $Q$ and components of participants' states not mentioned are unchanged.

Distinct machine transactions can represent ``the same action'' in different configurations, captured by an equivalence relation on machine transactions.
\begin{definition}[Transaction Equivalence]\label{definition:equivalence}
Given a set of machine transactions $R$, a \temph{transaction equivalence} is an equivalence relation $\sim$ on $R$ such that $t \sim t'$ implies $t$ and $t'$ have the same participants.  We write $[t]$ for the equivalence class of $t$ under $\sim$.
\end{definition}

\begin{definition}[Agent State and Configuration]\label{definition:agent-state}
Given agents $P$, states $S$ with initial state $s0$, a set of machine transactions $T$ each over its own participants $Q\subseteq P$ and $S$, and equivalence $\sim$ on $T$, an \temph{agent state} is a pair $(V,m)\in \calA = (2^{T/\sim}  \times S)$ where  $V$ is its \temph{volitional state} and $m \in S$ its  \temph{machine state}.  The \temph{initial agent state} is $(\emptyset,s0)$.  An \temph{agent configuration} $c$ over $P$, $S$, $T$, and $\sim$ is a member $c\in \calA^P$, in which case we write $c^v_p$  for the volitional state and $c^m_p$ for the machine state of agent $p$ in $c$.
\end{definition}

\begin{definition}[Enabled]\label{definition:enabled}
A guarded transaction $(t,Q')$ with $t=d\rightarrow d'$ over participants $Q$ is \temph{enabled} in an agent configuration $c$ if $c^m_p=d_p$ for every participant $p\in Q$ (the machine precondition) and $[t]\in c^v_q$ for every guard $q\in Q'$ (the volitional condition).  An equivalence class is \temph{enabled} when some representative is.  A transaction with empty guard ($Q'=\emptyset$) is enabled whenever its machine precondition holds.
\end{definition}

\begin{definition}[Volitional Multiagent Atomic Transaction]\label{definition:vmat}
Given agents $P$, states $S$, machine transactions $T$ over $P$ and $S$, and equivalence $\sim$ on $T$:
\begin{enumerate}
    \item A \temph{change-volition transaction of agent $p\in P$} is a pair $c\rightarrow c'$ of agent configurations over $\{p\}$, $S$, $T$, and $\sim$ such that $c^v_p \ne c'{^v_p} \subseteq T/\sim$ and $c^m_p = c'{^m_p}$.
    \item A \temph{volitional machine transaction} induced by a guarded machine transaction $(t,Q')$, for some $t= (d\rightarrow d')\in T$ over $Q\subseteq P$ and $Q'\subseteq Q$, is a pair $c\rightarrow c'$ where $c\ne c'$ are agent configurations over $P$, $S$, $T$, and $\sim$ such that $(t,Q')$ is enabled in $c$ (Definition~\ref{definition:enabled}); $c'{^m_p} = d'_p$ for every $p\in Q$; $c^m_p = c'{^m_p}$ for every $p\in P\setminus Q$; and $c'{^v_p} = c^v_p \setminus \{[t]\}$ for every $p\in P$.
    \item A \temph{volitional multiagent atomic transaction} is a change-volition transaction or a volitional machine transaction.
\end{enumerate}
\end{definition}
When a volitional machine transaction induced by $(t,Q')$ is taken, the class $[t]$ is removed from every agent's volitional state.  Volitions are thus discharged upon satisfaction---a person wills a class of transactions, and once any transaction in the class is taken, the will is fulfilled and the class is removed.  A person may independently change their volitional state via change-volition transactions, which may add or remove classes; beyond these, the framework removes a class from $c^v_p$ only via its satisfaction.

\mypara{Protocols, correctness, and the grassroots property}
A \temph{protocol} $\calF$ is a family of multiagent transition systems, one for each finite $P \subset \Pi$, over a common local-states function assigning to each $P$ a set of local states containing the initial state $s0$~\cite{lewis2026volitional}.  A run of $\calF(P)$ is \temph{safe} if every two consecutive configurations form a transition of $\calF(P)$, and \temph{live} if no equivalence class is enabled (Definition~\ref{definition:enabled}) throughout some suffix with no member of the class taken in that suffix; it is \temph{correct} if it is safe and live.  In the general theory~\cite{lewis2026volitional} the relation on transitions is a \emph{partial} equivalence, leaving change-volition transactions (Definition~\ref{definition:vmat}) outside any class and hence under no liveness obligation; here every machine-transaction class is total and the obligation falls on these classes alone.  Unguarded transactions carry the liveness obligation with no volitional prerequisite: once the machine precondition holds, some class representative must eventually be taken.

Runs of two disjoint groups can be \emph{interleaved} into a run of their union.  $\calF$ is \temph{oblivious} if every interleaving of two correct runs is a correct run of the union, and \temph{interactive} if some correct run of the union is not an interleaving of two independent runs, because it includes a step whose participants span both groups; it is \temph{grassroots} if both~\cite{lewis2026volitional}.  A set of guarded transactions induces a \temph{transactions-based protocol}, and we use two obliviousness results from~\cite{lewis2026volitional}.  General criterion: such a protocol is oblivious if no cross-group transaction---one whose participants span both groups---is ever \emph{enabled} in an interleaving of two correct runs.  \temph{Guarded Obliviousness}: a guarded cross-group transaction is never so enabled, since the guarding person, acting in one group, never wills it; guarding every cross-group transaction therefore suffices.

\mypara{Implementations}
We recall when one transition system implements another~\cite{lewis2026volitional}.  The notions below refer to runs and to stutter-removal, not to the equivalence $\sim$ directly (which enters only through ``correct run''), so we write a transition system as the triple $(C,c_0,T)$.

\begin{definition}[Implementation]\label{def:implementation}
Given transition systems $TS=(C,c_0,T)$, the \temph{specification}, and $TS'=(C',c_0',T')$, an \temph{implementation of $TS$ by $TS'$} is a function $\sigma:C'\to C$ with $\sigma(c_0')=c_0$.  For a computation $r'=c_1'\to c_2'\to\cdots$ of $TS'$, $\sigma(r')$ is the computation $\sigma(c_1')\to\sigma(c_2')\to\cdots$ with every \temph{stutter transition} $c\to c$ removed.  The mapping need not preserve transitions: an implementation transition may map to a stutter, and several to a single specification transition.
\end{definition}

\begin{definition}[Correct and Complete Implementation]\label{def:implementation-properties}
An implementation $(TS',\sigma)$ of $TS$ is \temph{correct} if $\sigma$ maps every correct run of $TS'$ to a correct run of $TS$, and \temph{complete} if every correct run of $TS$ is $\sigma(r')$ for some correct run $r'$ of $TS'$.
\end{definition}

\begin{lemma}[Transitivity of Implementation]\label{lemma:implementation-transitivity}
If $(TS',\sigma)$ is an implementation of $TS$ and $(TS'',\tau)$ an implementation of $TS'$, then $(TS'',\sigma\circ\tau)$ is an implementation of $TS$; it is correct if both are correct, and complete if both are complete.
\end{lemma}
\begin{proof}
$\sigma\circ\tau(c_0'')=\sigma(c_0')=c_0$, so it is an implementation.  If both are correct, $\tau$ maps a correct run of $TS''$ to a correct run of $TS'$, which $\sigma$ maps to a correct run of $TS$, and stutter-removal composes, so $\sigma\circ\tau$ is correct.  If both are complete, a correct run of $TS$ is $\sigma(r')$ for a correct run $r'$ of $TS'$, itself $\tau(r'')$ for a correct run $r''$ of $TS''$, so $(\sigma\circ\tau)(r'')$ is that run of $TS$; thus $\sigma\circ\tau$ is complete.
\end{proof}

\begin{definition}[Safety Fault, Fault-Resilient Implementation]\label{def:fault-resilient}
Let $(TS',\sigma)$ be an implementation of $TS$, with $TS'=(C',c_0',T')$.  A \temph{safety fault} is a set $F\subseteq C'^2\setminus T'$ of \temph{faulty transitions}, and a computation \temph{performs} $F$ if it includes a transition in $F$.  The implementation is \temph{$F$-resilient} if $\sigma$ maps every live run $r'\subseteq T'\cup F$ of $TS'$ to a correct run of $TS$.
\end{definition}
Faulty transitions belong to no transaction class, so they impose no liveness requirement; with $F=\emptyset$ the live runs $r'\subseteq T'$ are exactly the correct runs, so $F$-resilience extends correctness (Definition~\ref{def:implementation-properties}) to runs that perform and recover from faults in $F$.

\mypara{Communicating volitional agents: formal definitions}
We recall the formal definitions of communicating volitional agents (Section~\ref{sec:cva})~\cite{lewis2026volitional}.

\begin{definition}[CVA Local States and Configurations]\label{def:cva-state}
Given a set $C$ of \temph{cargoes} and a \temph{platform state space} $A$ with initial state $a_0\in A$, the set of \temph{messages} over $P\subset \Pi$ is
\[
M(P) \;:=\; \{\mathsf{message}(s,r,c) \mid s\ne r\in P,\; c\in C\},
\]
with $\mathsf{message}(s,r,c)$ a message of \temph{sender} $s$, \temph{recipient} $r$, and \temph{cargo} $c$.  A \temph{CVA local state} over $P$ is a tuple
\[
(\mathit{known},\, o,\, i,\, a) \;\in\; S(P) \;:=\; 2^{P} \times 2^{M(P)} \times 2^{M(P)} \times A ,
\]
where $\mathit{known}\subseteq P$ is a set of \temph{known peers}, $o\subseteq M(P)$ is an \temph{outbox}, $i\subseteq M(P)$ is an \temph{inbox}, and $a\in A$ is a \temph{platform state};
the \temph{initial local state} is  $(\emptyset,\emptyset,\emptyset,a_0)$.  A \temph{CVA configuration} over $P$ is a member $c\in S(P)^P$; for $p\in P$ we write $c_p = (\mathit{known}_p,\, o_p,\, i_p,\, a_p)$ for the local state of $p$ in $c$.  The \temph{initial configuration} over $P$, denoted $c0(P)$, assigns the initial local state to every $p\in P$.
\end{definition}

The local-states function $S$ satisfies $P\subseteq P' \implies S(P)\subseteq S(P')$, as required of a transactions-based protocol.  A CVA protocol consists of two ``built in'' binary transactions---\emph{discover} and \emph{communicate}---and platform-specific unary transactions:

\begin{definition}[CVA Protocol]\label{def:cva-transactions}
A \emph{CVA protocol} consists of every binary transaction $c\rightarrow c'$ over $\{p,q\}$ for every $p\ne q\in \Pi$ such that $c'=c$ except that:
\begin{enumerate}
    \item \textbf{Discover}: guarded by $\{p\}$, $\mathit{known}'_p := \mathit{known}_p \cup \{q\}$.

    \item \textbf{Communicate}: unguarded, $i'_q := i_q \cup M$, 
    $o'_p := o_p \setminus  M$, provided $M=\{\mathsf{message}(p,q,c)\} \subseteq o_p$.
\end{enumerate}
In addition it has for each $p\in P$ guarded unary platform transactions:
\begin{enumerate}\setcounter{enumi}{2}
    \item \textbf{Platform transactions}: unary transactions $c\rightarrow c'$ with participant $\{p\}$ and guard $\emptyset$ or $\{p\}$, with a precondition over the local state of $p$---it may inspect $\mathit{known}_p$, $i_p$, and $a_p$---each modifying $a_p$ to $a'_p$  and/or adding one or more $\mathsf{message}(s,r,c)$ to $o_p$ resulting in $o'_p$, provided $s=p$, $r\in \mathit{known}_p$ and $c\in C$.  A platform transaction does not remove messages from $i_p$.
\end{enumerate}
\end{definition}

Inboxes and outboxes are sets (Definition~\ref{def:cva-state}).  A reactive platform transaction reads its inbox in its precondition but does not consume from it; each such transaction is idempotent by its own precondition, which a monotone guard (an epoch, a frontier, or a domain-membership test) renders unsatisfiable once the transaction has fired, so the same retained message cannot enable it again.  Inboxes are therefore monotone non-decreasing, and bounding them is an implementation concern below the CVA abstraction.

The following are proven in Appendix~\ref{appendix:proofs}:
\begin{restatable}[Known-Peers Containment]{lemma}{ThmKnownContainment}\label{lem:known-containment}
In any run of $\calF(P)$, for every configuration $c$ along the run and every $p\in P$: $\mathit{known}_p \subseteq P$.
\end{restatable}
\begin{restatable}[Outbox Containment]{lemma}{ThmOutboxContainment}\label{lem:outbox-containment}
In any run of $\calF(P)$, for every configuration $c$ along the run, every $p\in P$, and every $\mathsf{message}(s,r,\cdot)\in o_p$: $r\in P$.
\end{restatable}

%% file: sections/appendix-gsg-cva.tex
\input{sections/gsg-cva}

\input{sections/gsg-cva-properties}

\input{sections/secure-gsg-cva}

\input{sections/secure-gsg-cva-restore}

\input{sections/secure-gsg-cva-replace}

\input{sections/secure-gsg-cva-properties}

%% file: sections/gsg-cva.tex
\section{CVA Implementation of the Secure Social Graph (Full Development)}\label{section:gsg-cva}

We present the grassroots social graph as a CVA platform (Section~\ref{sec:cva}).  The Befriend and Unfriend transactions of the abstract spec (Section~\ref{section:gsg-abstract}) are realised as unary platform transactions.  The realisation uses an epoch encoding, an offer/accept handshake, and stream dissemination, established directly over CVA's message discipline.

Communication uses a single primitive: a CVA message $(\mathit{sender},\mathit{recipient},\mathit{cargo})$.  A broadcast to $p$'s friends is realised as one such message per recipient in $p$'s outbox.  Every recipient is a friend in $\mathrm{FMap}_p$ (added on every befriend and never removed), so the recipient set is exactly $\mathrm{dom}(\mathrm{FMap}_p)$ and needs no separate state.

\subsection{Platform State}\label{section:gsg-cva-state}

The \emph{epoch} of a friendship is an integer in $\mathbb{N}$ that encodes both a counter and a status by parity: an \emph{odd} epoch denotes an \emph{active} friendship, an \emph{even} epoch (including $0$) denotes an \emph{inactive} one.  We treat $0$ as the distinguished value $\bot$ denoting ``no connection has ever been made.''  The \emph{Friendship Order} on $\mathbb{N}$ is the natural ordering, corresponding to the friendship lifecycle
\[
  \bot = 0 \;<\; 1_{\textit{active}} \;<\; 2_{\textit{inactive}} \;<\; 3_{\textit{active}} \;<\; 4_{\textit{inactive}} \;<\; \cdots
\]

The platform state of $p\in P$ is a pair
\[
a_p \;=\; (\mathrm{FMap}_p,\; \mathrm{FoFMap}_p),
\]
where
\begin{itemize}
    \item $\mathrm{FMap}_p : \Pi \rightharpoonup \mathbb{N}\times D$ is a partial map, each entry $\mathrm{FMap}_p[q]=(e,d)$ recording the epoch $e$ of the friendship between $p$ and $q$ together with an opaque \emph{application-data} field $d\in D$ ($D$ a platform parameter with default $\bot\in D$; Section~\ref{section:gsg-cva-data});
    \item $\mathrm{FoFMap}_p : \Pi\rightharpoonup (\Pi\rightharpoonup \mathbb{N})$ is a partial nested map, each entry $\mathrm{FoFMap}_p[q][f]=e$ recording $p$'s observed epoch of the friendship between $q$ and $f$.
\end{itemize}
The initial platform state is $(\emptyset,\emptyset)$.

We write $\mathrm{epoch}_p(q) := \mathrm{FMap}_p[q].\mathit{epoch}$ when $q\in\mathrm{dom}(\mathrm{FMap}_p)$ and $\mathrm{epoch}_p(q) := 0$ otherwise.  Similarly $\mathrm{epoch}_p(q,f) := \mathrm{FoFMap}_p[q][f]$ when defined and $0$ otherwise.  The \emph{recipient set} of $p$ is $\mathrm{Rec}_p := \mathrm{dom}(\mathrm{FMap}_p)$.

\subsection{Messages}\label{section:gsg-cva-cargo}

The cargo space $C$ includes the following tagged messages, for $x\in\mathbb{N}$ and $f\in\Pi$:
\begin{itemize}
    \item $\mathsf{friend\_request}(x)$: friendship offer at epoch $x$ (odd).
    \item $\mathsf{accept}(x)$: acceptance of a pending $\mathsf{friend\_request}(x)$.
    \item $\mathsf{unfriend}(x)$: unilateral unfriend at epoch $x$ (even).
    \item $\mathsf{stream\_update}(f,x)$: stream update about third party $f$ at epoch $x$.
\end{itemize}

\subsection{Befriend}\label{section:gsg-cva-befriend}

Befriending between $p$ and $q$ proceeds in three steps: $p$ sends a $\mathsf{friend\_request}$ at a fresh odd epoch, $q$ responds with $\mathsf{accept}$, and $p$ integrates the acceptance.  A simultaneous offer from $q$ at the same epoch is resolved by name order on $\Pi$: the agent with the larger name accepts the incoming request, the agent with the smaller name ignores it.

Fix a strict total order $<$ on $\Pi$.

\begin{definition}[Befriend Platform Transactions]\label{definition:gsg-cva-befriend}
\leavevmode

\mypara{Offer friendship}  A unary transaction at $p$, guarded by $\{p\}$, for $q\in\mathit{known}_p$ such that $\mathrm{epoch}_p(q)$ is even:
\begin{itemize}
    \item Let $x := \mathrm{epoch}_p(q) + 1$.
    \item Add $\mathsf{message}(p,q,\mathsf{friend\_request}(x))$ to $o_p$.
\end{itemize}

\mypara{Accept offer}  A unary transaction at $q$, guarded by $\{q\}$, provided $\mathsf{message}(p,q,\mathsf{friend\_request}(x))\in i_q$, $\mathrm{epoch}_q(p)$ is even and $\mathrm{epoch}_q(p)<x$, and no $\mathsf{message}(q,p,\mathsf{friend\_request}(x))\in o_q$:
\begin{itemize}
    \item Set $\mathrm{FMap}_q[p] := x$.
    \item Add $\mathsf{message}(q,p,\mathsf{accept}(x))$ to $o_q$.
    \item For each $r\in\mathrm{Rec}_q$ (after the update), add $\mathsf{message}(q,r,\mathsf{stream\_update}(p,x))$ to $o_q$.
    \item For each $f\in\mathrm{Rec}_q$ with $f\ne p$ (after the update), add $\mathsf{message}(q,p,\mathsf{stream\_update}(f,\mathrm{epoch}_q(f)))$ to $o_q$ (the on-join snapshot).
\end{itemize}

\mypara{Resolve simultaneous offer}  A unary transaction at $q$, unguarded, provided $\mathsf{message}(p,q,\mathsf{friend\_request}(x))\in i_q$, $\mathsf{message}(q,p,\mathsf{friend\_request}(x))\in o_q$, $\mathrm{epoch}_q(p)$ is even and $\mathrm{epoch}_q(p)<x$, and $p<q$:
\begin{itemize}
    \item Set $\mathrm{FMap}_q[p] := x$.
    \item Add $\mathsf{message}(q,p,\mathsf{accept}(x))$ to $o_q$.
    \item For each $r\in\mathrm{Rec}_q$ (after the update), add $\mathsf{message}(q,r,\mathsf{stream\_update}(p,x))$ to $o_q$.
    \item For each $f\in\mathrm{Rec}_q$ with $f\ne p$ (after the update), add $\mathsf{message}(q,p,\mathsf{stream\_update}(f,\mathrm{epoch}_q(f)))$ to $o_q$ (the on-join snapshot).
\end{itemize}

\mypara{Integrate accept}  A unary transaction at $p$, unguarded, provided $\mathsf{message}(q,p,\mathsf{accept}(x))\in i_p$, $\mathrm{epoch}_p(q)$ is even and $\mathrm{epoch}_p(q)<x$:
\begin{itemize}
    \item Set $\mathrm{FMap}_p[q] := x$.
    \item For each $r\in\mathrm{Rec}_p$ (after the update), add $\mathsf{message}(p,r,\mathsf{stream\_update}(q,x))$ to $o_p$.
    \item For each $f\in\mathrm{Rec}_p$ with $f\ne q$ (after the update), add $\mathsf{message}(p,q,\mathsf{stream\_update}(f,\mathrm{epoch}_p(f)))$ to $o_p$ (the on-join snapshot).
\end{itemize}
\end{definition}

Offer friendship and Accept offer are volitional, guarded by the willing agent's person.  Resolve simultaneous offer fires only at the agent with the larger name; the symmetric scenario with $q<p$ is handled by $p$'s own instance of this transaction firing on $q$'s incoming request.  Integrate accept is reactive at the offerer: once the offer has been willed, consuming the acceptance requires no further volition.  In each of the three activating transactions, the activating agent both broadcasts the new edge to its existing friends (the $\mathsf{stream\_update}(\cdot,x)$ to each $r\in\mathrm{Rec}$) and sends the new friend a snapshot of its existing friendships (the $\mathsf{stream\_update}(f,\mathrm{epoch}(f))$ to the new friend); the snapshot is discussed in Section~\ref{section:gsg-cva-stream}.

\subsection{Unfriend}\label{section:gsg-cva-unfriend}

Unfriending is unilateral: either party may initiate.  The initiator increments its FMap entry from an odd epoch to the next even epoch, sends $\mathsf{unfriend}$ to the counterpart, and emits a stream update.  The counterpart integrates on receipt, discarding messages from non-friends.

\begin{definition}[Unfriend Platform Transactions]\label{definition:gsg-cva-unfriend}
\leavevmode

\mypara{End friendship}  A unary transaction at $p$, guarded by $\{p\}$, for $q\in\mathrm{Rec}_p$ with $\mathrm{epoch}_p(q)$ odd:
\begin{itemize}
    \item $\mathrm{FMap}_p[q] := \mathrm{epoch}_p(q) + 1$.  Let $x := \mathrm{FMap}_p[q]$ (even).
    \item Add $\mathsf{message}(p,q,\mathsf{unfriend}(x))$ to $o_p$.
    \item For each $r\in\mathrm{Rec}_p$, add $\mathsf{message}(p,r,\mathsf{stream\_update}(q,x))$ to $o_p$.
\end{itemize}

\mypara{Integrate unfriend}  A unary transaction at $q$, unguarded, provided $\mathsf{message}(p,q,\mathsf{unfriend}(x))\in i_q$, $p\in\mathrm{dom}(\mathrm{FMap}_q)$, and $\mathrm{epoch}_q(p)<x$:
\begin{itemize}
    \item Set $\mathrm{FMap}_q[p] := x$.
    \item For each $r\in\mathrm{Rec}_q$, add $\mathsf{message}(q,r,\mathsf{stream\_update}(p,x))$ to $o_q$.
\end{itemize}
\end{definition}

The precondition $p\in\mathrm{dom}(\mathrm{FMap}_q)$ on Integrate unfriend ensures that a spurious $\mathsf{unfriend}$ from a non-friend has no effect.  
The precondition $\mathrm{epoch}_q(p)<x$ rejects stale and duplicate unfriend messages.

\subsection{Stream Dissemination}\label{section:gsg-cva-stream}

Every Befriend and Unfriend transaction that updates $\mathrm{FMap}_p$ emits a $\mathsf{stream\_update}$ to each $r\in\mathrm{Rec}_p$, as specified in the effects above.  A recipient integrates the update subject to the Friendship Order.

In addition, when a friendship with a new friend is activated --- by Accept offer, Resolve simultaneous offer, or Integrate accept --- the activating agent sends that new friend a one-shot \emph{snapshot} of its current friendships, one $\mathsf{stream\_update}(f,\mathrm{epoch}(f))$ per existing friend $f$, as specified in the effects above.  The snapshot catches the new friend up on friendships formed before it joined, so that an observer view never depends on the order in which friendships were established; without it, a friend that joined after a third-party edge was formed would never learn that edge, as the edge emits no further updates unless it changes again.

\begin{definition}[Integrate stream update]\label{definition:gsg-cva-stream-integrate}
A unary transaction at $r$, unguarded, provided $\mathsf{message}(q,r,\mathsf{stream\_update}(f,x))\in i_r$, $q\in\mathrm{dom}(\mathrm{FMap}_r)$, and $\mathrm{epoch}_r(q,f)<x$:
\begin{itemize}
    \item $\mathrm{FoFMap}_r[q][f] := x$.
    \item $\mathrm{FoFMap}_r[f][q] := x$.
\end{itemize}
\end{definition}

The second assignment preserves FoFMap Symmetry (Invariant~\ref{invariant:gsg-fofmap-symmetry}): $r$'s observed epoch of $(q,f)$ and of $(f,q)$ are equal at every reachable state.  The Friendship Order guard ensures monotone integration: no stream update can drive an observed epoch backward.  Per-edge updates and snapshot entries are integrated by the same transaction, so the snapshot adds no new integration rule.

\subsection{Application Data}\label{section:gsg-cva-data}

The application-data field is an optional extension; the social graph operates unchanged without it, and a platform needing no per-friend data omits $D$.  Each direct-friend record carries an opaque \emph{application-data} field $d\in D$, where $D$ is a platform parameter with default $\bot\in D$.  The social graph transports and stores $d$ but never inspects it, disseminates it, or lets it influence epoch or status.  The Befriend transactions of Section~\ref{section:gsg-cva-befriend}, written ``Set $\mathrm{FMap}_p[q]:=x$'', set the epoch component to $x$ and the data component to $\bot$; the Unfriend transactions of Section~\ref{section:gsg-cva-unfriend} and every other transaction modify only the epoch and preserve the data component.  Two purely-local operations expose $d$ to the layered application:

\begin{definition}[Application-Data Operations]\label{definition:gsg-cva-data}
\leavevmode

\mypara{Get data}  A read-only query at $p$, for $q\in\mathrm{dom}(\mathrm{FMap}_p)$: return $\mathrm{FMap}_p[q].\mathit{data}$ to the application; no state change, no message.

\mypara{Set data}  A unary transaction at $p$, guarded by $\{p\}$, for $q\in\mathrm{dom}(\mathrm{FMap}_p)$ and a value $d\in D$ supplied by the application: set $\mathrm{FMap}_p[q].\mathit{data}:=d$.
\end{definition}

Neither generates a protocol message nor affects epoch or status.  The semantics of $D$ are defined by the layered application and are outside the scope of this paper.

%% file: sections/gsg-cva-properties.tex
\subsection{Properties}\label{section:gsg-cva-properties}

We state the CVA-level invariants of the social graph.  The invariants below are stated over CVA's transactions and message discipline; their proofs are in Appendix~\ref{appendix:proofs}.  These are the properties of the fault-free base graph; the security layer (Appendix~\ref{section:secure-gsg-cva-properties}) inherits the invariants and adds the recovery results, re-establishing only what faults affect. 

\begin{definition}[Quiescence]\label{definition:gsg-cva-quiescence}
A configuration is \emph{quiescent} if (1) no volitional transaction is enabled and (2) there are no new (yet unprocessed) messages in inboxes and outboxes. 
\end{definition}

\noindent Note that periodic rebroadcasts of checkpoints continue to occur in quiescent states, but they repeatedly send the same message while new messages are sent.

Importantly, at a quiescent configuration, every message that has been sent has been delivered and processed at least once 
and no agent holds a willed transaction class that has not been consumed.  The abstract social graph (Section~\ref{section:gsg-abstract}), whose transactions are instantaneous, is quiescent exactly when no agent holds an unconsumed volition; its quiescent friend-set assignments are, by Friendship Mutuality (Lemma~\ref{lemma:gsg-mutuality}), precisely the mutual ones, and every mutual assignment over $P$ is reached by a run that befriends the corresponding pairs.

\subsubsection{Invariants}\label{section:gsg-cva-invariants}

\begin{restatable}[Knowledge Monotonicity]{lemma}{ThmKnowledgeMono}\label{lemma:gsg-knowledge-mono}
In any run of the social graph CVA implementation, once $q\in\mathrm{dom}(\mathrm{FMap}_p)$ or $f\in\mathrm{dom}(\mathrm{FoFMap}_p[q])$ becomes true, it remains true at all later states.
\end{restatable}

\begin{restatable}[Friendship Monotonicity]{lemma}{ThmFriendshipMono}\label{lemma:gsg-friendship-mono}
In any run of the social graph CVA implementation, $\mathrm{epoch}_p(q)$ and $\mathrm{epoch}_p(q,f)$ are monotonically non-decreasing in the natural order on $\mathbb{N}$.
\end{restatable}

\begin{restatable}[Message Bounds]{lemma}{ThmMessageBounds}\label{lemma:gsg-message-bounds}
For every message that has been sent but not yet fully processed by its recipient:
\begin{enumerate}
    \item if $\mathsf{friend\_request}(x)$ is in transit from $B$ to $A$, then $\mathrm{epoch}_B(A)\ge x-1$;
    \item if $\mathsf{accept}(x)$ is in transit from $B$ to $A$, then $\mathrm{epoch}_B(A)\ge x$;
    \item if $\mathsf{unfriend}(x)$ is in transit from $B$ to $A$, then $\mathrm{epoch}_B(A)\ge x$;
    \item if $\mathsf{stream\_update}(f,x)$ is in transit from $B$, then $\mathrm{epoch}_B(f)\ge x$.
\end{enumerate}
\end{restatable}

\begin{restatable}[Observer Bound]{lemma}{ThmObserverBound}\label{lemma:gsg-observer-bound}
For any observer $A$ and any pair $(p,q)$, at all times $\mathrm{epoch}_A(p,q)\le\max(\mathrm{epoch}_p(q),\,\mathrm{epoch}_q(p))$.
\end{restatable}

\begin{restatable}[FoFMap Symmetry]{invariant}{ThmFofMapSymmetry}\label{invariant:gsg-fofmap-symmetry}
For every $p,q,f\in\Pi$ at every reachable state of the social graph CVA implementation, $\mathrm{epoch}_p(q,f) = \mathrm{epoch}_p(f,q)$.
\end{restatable}

The following lemma establishes robustness to message reordering and duplication. 

\begin{restatable}[Monotone Integration]{lemma}{ThmStaleInert}\label{lemma:gsg-cva-stale-inert}
Let $m$ be a message from $p$ to $q$ concerning the friendship between $p$ and $r$ (possibly $r=q$) with epoch $m.e$, and let $e$ be the epoch $q$ holds for that friendship. If $m.e\le e$, $m$ is not consumed and does not change $q$'s state.
\end{restatable}

\begin{definition}[Active Friend Set]\label{definition:gsg-cva-active-friend-set}\label{lemma:gsg-mutuality-cva}
The \emph{active friend set} of $p$ is $\widetilde F_p:=\{q\in\mathrm{dom}(\mathrm{FMap}_p):\mathrm{epoch}_p(q)\text{ is odd}\}$.
\end{definition}

\subsubsection{Safety Properties}\label{section:gsg-cva-safety}

We state the intent-based safety guarantees of the social graph CVA implementation, which hold at every reachable state---a configuration on some run from the initial configuration---not only at quiescence. They are phrased in terms of an agent \emph{wanting to befriend} another, the CVA counterpart of the volitional act of offering or accepting friendship.

We denote by Person$_p$ the person behind agent $p$. 

\begin{definition}[Wanting to Befriend]\label{definition:gsg-cva-wants}
An agent $p$ \emph{wants to befriend } an agent $q$ at time $t$ if there is a time $t'\le t$ at which Person$_p$ wills either Offer friendship$(q')$ for an agent $q'$ s.t.\ Person$_q$ = Person$_{q'}$, or the acceptance of a friend request from such an agent $q'$ (the volition guarding Accept offer or Resolve simultaneous offer with $q'$), and  $p$ does not will End friendship$(q')$ for an agent $q'$ of Person$_q$ 
at any time in $(t',t]$.
\end{definition}

Note that a major fault changes an agent's identity but not the person behind it, so wanting is a relation between the two persons rather than between the identities they currently hold.

\begin{restatable}[Post-Accept State]{lemma}{ThmPostAccept}\label{lemma:gsg-cva-post-accept}
Suppose $p$ sent $\mathsf{accept}(x)$ to $q$ at time $t_a$ (via Accept offer or Resolve simultaneous offer), and neither $p$ nor $q$ willed End friendship on the other during $(t_a,t]$, for some $t\ge t_a$. Then $\mathrm{epoch}_q(p)\in\{x-1,x\}$ at $t$.
\end{restatable}

\begin{restatable}[Post-Unfriend State]{lemma}{ThmPostUnfriend}\label{lemma:gsg-cva-post-unfriend}
Suppose $\mathrm{epoch}_p(q)=\mathrm{epoch}_q(p)=x$ with $x$ odd at time $t_0$, at least one of $p,q$ willed End friendship on the other at some $t_1\ge t_0$, and neither willed Offer friendship on the other during $[t_0,t]$. Then $\mathrm{epoch}_p(q),\mathrm{epoch}_q(p)\in\{x,x+1\}$ for all $t\ge t_1$.
\end{restatable}

\begin{restatable}[Observed-Epoch Provenance]{lemma}{ThmFofProvenance}\label{lemma:gsg-cva-fof-provenance}
If $\mathrm{epoch}_p(q,f)=x\ge 1$ at time $t_1$, then there is a time $t_2\le t_1$ at which $\mathrm{epoch}_q(f)=x$ or $\mathrm{epoch}_f(q)=x$.
\end{restatable}

\begin{restatable}[Friend List Soundness]{theorem}{ThmFriendListSoundness}\label{theorem:gsg-cva-friend-list-soundness}
If $q\in\widetilde F_p$ at time $t_3$, then there are times $t_1,t_2\le t_3$ such that $p$ wants to befriend $q$ at $t_1$ and $q$ wants to befriend $p$ at $t_2$.
\end{restatable}

\begin{restatable}[Friend-of-Friend Soundness]{theorem}{ThmFofSoundness}\label{theorem:gsg-cva-fof-soundness}
If $\mathrm{epoch}_p(q,r)$ is odd at time $t_3$, then there are times $t_1,t_2\le t_3$ such that $q$ wants to befriend $r$ at $t_1$ and $r$ wants to befriend $q$ at $t_2$.
\end{restatable}

\begin{restatable}[Addressability]{theorem}{ThmChannelValidity}\label{theorem:gsg-cva-channel-validity}
If $q\in\widetilde F_p$ at time $t$, then $q\in\mathit{known}_p$ and $p\in\mathit{known}_q$ at $t$.
\end{restatable}

\subsubsection{Liveness Properties}\label{section:gsg-cva-liveness}

The liveness guarantees hold under the eventual-delivery property of correct runs: a message placed in an outbox is eventually delivered to its recipient and integrated, since the communicate transaction carrying it is unguarded and remains enabled until it fires (Section~\ref{section:dts}).

\begin{restatable}[Mutual Friendship Liveness]{lemma}{ThmMutualLiveness}\label{lemma:gsg-cva-mutual-liveness}
If $p$ wants to befriend $q$ at every time $t'\ge t$ and $q$ wants to befriend $p$ at every time $t'\ge t$, then there is a time $t^*>t$ such that, from $t^*$ onward, $\mathrm{epoch}_p(q)=\mathrm{epoch}_q(p)$ and this common value is odd.
\end{restatable}

\begin{restatable}[Friendship Establishment]{theorem}{ThmFriendshipEstablishment}\label{theorem:gsg-cva-establishment}
If $p$ wants to befriend $q$ at every time $t'\ge t$ and $q$ wants to befriend $p$ at every time $t'\ge t$, then there is a time $t^*>t$ such that, from $t^*$ onward, $q\in\widetilde F_p$ and $p\in\widetilde F_q$.
\end{restatable}

\begin{restatable}[Friend-of-Friend Visibility]{theorem}{ThmFofVisibility}\label{theorem:gsg-cva-fof-visibility}
If $p$ wants to befriend $q$ at every time $t'\ge t_1$, $q$ wants to befriend $p$ at every time $t'\ge t_1$, $q$ wants to befriend $r$ at every time $t'\ge t_2$, and $r$ wants to befriend $q$ at every time $t'\ge t_2$, then there is a time $t^*>\max(t_1,t_2)$ such that, from $t^*$ onward, $\mathrm{epoch}_p(q,r)$ is odd, so $(q,r)$ appears in the friend-of-friend view at $p$.
\end{restatable}

\begin{restatable}[Unfriend Propagation]{theorem}{ThmUnfriendPropagation}\label{theorem:gsg-cva-unfriend-propagation}
If from some time $t$ onward $p$ does not want to befriend $q$, or from $t$ onward $q$ does not want to befriend $p$, then there is a time $t^*\ge t$ such that, from $t^*$ onward:
\begin{enumerate}
    \item $q\notin\widetilde F_p$ and $p\notin\widetilde F_q$; and
    \item for every $r$ that is a sustained mutual friend of $p$ or of $q$, $\mathrm{epoch}_r(p,q)$ is even, so $(p,q)$ does not appear in the friend-of-friend view at $r$.
\end{enumerate}
\end{restatable}

%% file: sections/secure-gsg-cva.tex
\subsection{The Security Layer: Identity Records and Recovery}\label{section:secure-gsg-cva}

We present the secure social graph as a CVA platform extending the social graph CVA implementation with identity-record transport on Befriend, storage of friends' identity records in $\mathrm{FMap}$ entries, and the recovery protocols Replace and Restore that implement the abstract Replace and Recover transactions (Section~\ref{section:secure-gsg-abstract}).  Replace and Restore are given in Sections~\ref{section:secure-gsg-cva-replace} and~\ref{section:secure-gsg-cva-restore}.

\subsection{Platform State Extension}\label{section:secure-gsg-cva-state}

An \emph{identity record} of $p\in\Pi$ is the pair $\mathit{IR}_p=(K_p,\sigma_p)$ of $p$'s custodian set $K_p\subset\Pi$ and supermajority threshold $\sigma_p\in(1/2,1]$, intrinsic and immutable (Section~\ref{section:secure-gsg-abstract}).  Being intrinsic, $\mathit{IR}_p$ is accessible at $p$ as a read-only datum throughout; in particular, it survives a state-loss fault, in which only platform state is reset.

The platform state of Secure GSG-CVA extends GSG-CVA's by widening each $\mathrm{FMap}$ entry with the counterpart's identity record:
\[
    \mathrm{FMap}_p : \Pi \rightharpoonup (\mathbb{N} \times D \times 2^\Pi \times (1/2,1])
\]

with each entry $\mathrm{FMap}_p[q]=(\mathit{epoch},\mathit{data},K,\sigma)$ recording $p$'s epoch of the friendship with $q$ and $p$'s view of $q$'s identity record as received on befriend.  The data component is inherited from GSG-CVA (Section~\ref{section:gsg-cva-data}), opaque to the security layer, and preserved by all secure transactions.  For brevity, $\mathrm{FMap}_p[q].\mathit{epoch}$, $\mathrm{FMap}_p[q].K$, $\mathrm{FMap}_p[q].\sigma$ denote the components of the record; $$\mathit{IR}_p[q]:=(\mathrm{FMap}_p[q].K,\mathrm{FMap}_p[q].\sigma).$$

The shorthand $\mathrm{epoch}_p(q)$ refers to $\mathrm{FMap}_p[q].\mathit{epoch}$, with convention $\mathrm{epoch}_p(q):=0$ when $q\notin\mathrm{dom}(\mathrm{FMap}_p)$.  $\mathrm{FoFMap}_p$ is unchanged from GSG-CVA.  A write to an $\mathrm{FMap}$ entry shown as a triple $(\mathit{epoch},K,\sigma)$ sets those three components and leaves the data component unchanged---$\bot$ for a newly created entry; the full entry is the 4-tuple of the state definition above.

\subsection{Cargo Extension}\label{section:secure-gsg-cva-cargo}

The GSG-CVA cargo is extended so that friendship-offer messages transport the sender's identity record, for $R=(K,\sigma)$ an identity record:
\begin{itemize}
    \item $\mathsf{friend\_request}(x,R)$: friendship offer at epoch $x$ carrying the offerer's identity record.
    \item $\mathsf{accept}(x,R)$: acceptance of a pending $\mathsf{friend\_request}(x,\cdot)$, carrying the accepter's identity record.
\end{itemize}
The $\mathsf{unfriend}$ and $\mathsf{stream\_update}$ cargos are unchanged.  Additional cargo for Replace is introduced in Section~\ref{section:secure-gsg-cva-replace-cargo}.

\subsection{Befriend Extension}\label{section:secure-gsg-cva-befriend}

The Befriend transactions of Section~\ref{section:gsg-cva-befriend} are augmented to transport and store identity records.  The structure is unchanged; only the cargo and the effect on $\mathrm{FMap}$ entries are extended.

\begin{definition}[Befriend Platform Transactions, Secure Extension]\label{definition:secure-gsg-cva-befriend}
\leavevmode

\mypara{Offer friendship}  As in Section~\ref{section:gsg-cva-befriend}, sending $\mathsf{message}(p,q,\mathsf{friend\_request}(x,\mathit{IR}_p))$ in place of $\mathsf{message}(p,q,\mathsf{friend\_request}(x))$.

\mypara{Accept offer}  As in Section~\ref{section:gsg-cva-befriend}, on receipt of $\mathsf{message}(p,q,\mathsf{friend\_request}(x,R))$, setting $\mathrm{FMap}_q[p]:=(x,K,\sigma)$ where $R=(K,\sigma)$, and sending $\mathsf{accept}(x,\mathit{IR}_q)$.

\mypara{Resolve simultaneous offer}  As in Section~\ref{section:gsg-cva-befriend}, with cargo and effect extended analogously.

\mypara{Integrate accept}  As in Section~\ref{section:gsg-cva-befriend}, on receipt of $\mathsf{message}(q,p,\mathsf{accept}(x,R))$, setting $\mathrm{FMap}_p[q]:=(x,K,\sigma)$ where $R=(K,\sigma)$.
\end{definition}

\subsection{Unfriend, Inherited}\label{section:secure-gsg-cva-unfriend}

The Unfriend transactions are inherited from Section~\ref{section:gsg-cva-unfriend} unchanged, with the convention that End friendship and Integrate unfriend preserve the $K$ and $\sigma$ fields of the $\mathrm{FMap}$ entry while modifying only the epoch.

\subsection{Stream Dissemination, Inherited}\label{section:secure-gsg-cva-stream}

Inherited from Section~\ref{section:gsg-cva-stream} unchanged.  The $\mathsf{stream\_update}$ cargo does not carry identity records; $\mathrm{FoFMap}$ entries do not record them.

\subsection{Periodic Re-broadcast}\label{section:secure-gsg-cva-rebroadcast}

The fault model (Section~\ref{section:secure-gsg-abstract}) allows a state loss to erase an agent's machine state, including its outbox, so a friendship change in transit can be lost and a counterpart left stale.  To recover without a global resource, each agent periodically re-broadcasts a snapshot of its own friendships to its friends, and a recipient absorbs it monotonically, catching up on whatever it missed.  
This single mechanism serves both purposes of state recovery: it repairs an agent that lost buffered messages, and it rebuilds an agent that lost its whole machine state, which recovers passively by absorbing the checkpoints its friends continue to send (Section~\ref{section:secure-gsg-cva-restore}).

The cargo is extended with $\mathsf{checkpoint}(R,L)$: a re-broadcast  carrying the sender's identity record $R=\mathit{IR}$ and a snapshot $L=\{(f,e)\}$ of the sender's friendships, each a (friend, epoch) pair.  The identity record lets a recipient that has lost its record of the sender re-establish it, the mechanism on which state-loss Restore rests (Section~\ref{section:secure-gsg-cva-restore}).

\begin{definition}[Periodic Re-broadcast Transactions]\label{definition:secure-gsg-cva-rebroadcast}
\leavevmode

\mypara{Re-broadcast}  A unary transaction at $p$, unguarded, provided $\mathrm{Rec}_p\ne\emptyset$, %
with $L:=\{(f,\mathrm{epoch}_p(f)) : f\in\mathrm{Rec}_p\}$:
\begin{itemize}
    \item For each $r\in\mathrm{Rec}_p$, add $\mathsf{message}(p,r,\mathsf{checkpoint}(\mathit{IR}_p,L))$ to $o_p$.
\end{itemize}

\mypara{Integrate checkpoint}  A unary transaction at $r$, unguarded, provided\\ $\mathsf{message}(p,r,\mathsf{checkpoint}(R,L))\in i_r$ and the state change below is non-trivial:
\begin{itemize}
    \item If $p\notin\mathrm{dom}(\mathrm{FMap}_r)$, writing $(K,\sigma):=R$, set $\mathrm{FMap}_r[p]:=(0,\bot,K,\sigma)$, a skeleton entry at the inactive epoch $0$ carrying the identity record of the sender; if instead $p\in\mathrm{dom}(\mathrm{FMap}_r)$ and $\mathrm{FMap}_r[p]$ is a stub (identity-record fields $\bot$), writing $(K,\sigma):=R$, set those fields to $K,\sigma$, promoting the stub to a full entry and preserving its epoch.
    \item If $(r,e)\in L$ for some $e$ with $\mathrm{epoch}_r(p)<e$: set $\mathrm{FMap}_r[p].\mathit{epoch}:=e$, preserving the $K$ and $\sigma$ fields of the entry.
    \item For each $(f,e)\in L$ with $f\ne r$ and $\mathrm{epoch}_r(p,f)<e$: set $\mathrm{FoFMap}_r[p][f]:=e$ and $\mathrm{FoFMap}_r[f][p]:=e$.
\end{itemize}
\end{definition}

Re-broadcast is unguarded: it is mechanical, like the advance of local time, and requires no volition. As long as there are potential receivers, it is always enabled. 
Thus, every agent re-broadcasts infinitely often, the property on which passive state recovery rests (Section~\ref{section:secure-gsg-cva-restore}).  

Integrate checkpoint heals two kinds of staleness under the epoch order.  The pair $(r,e)\in L$ carries the epoch $p$ holds for the direct $p$--$r$ friendship; absorbing it when it exceeds $\mathrm{epoch}_r(p)$ repairs a direct edge on which $r$ fell behind, for instance after $r$ lost an $\mathsf{accept}$ or $\mathsf{unfriend}$ from $p$.  The direct heal advances only the epoch, preserving the identity-record fields, exactly as Unfriend does (Section~\ref{section:secure-gsg-cva-unfriend}).  The remaining pairs carry the other friendships of $p$ and are absorbed exactly as $\mathsf{stream\_update}$ messages, refreshing $\mathrm{FoFMap}_r[p]$ and preserving FoFMap Symmetry.  When the sender $p$ is one $r$ does not yet record, Integrate checkpoint first installs a skeleton entry for $p$ at the inactive epoch $0$ carrying the sender's identity record from $R$, the mechanism on which state-loss Restore rests (Section~\ref{section:secure-gsg-cva-restore}); it is inert in fault-free operation, where every checkpoint sender is already a recorded friend.  If instead $r$ holds only a stub for $p$---installed by a Replace rebind (Section~\ref{section:secure-gsg-cva-replace}) whose acknowledgement to $r$ was lost---the identity record carried by any later checkpoint from $p$ promotes the stub to a full entry, so periodic re-broadcast completes a friend recovered only as a stub.  The epoch order guards ensure no absorption drives an epoch backward.  A checkpoint never establishes an active friendship the recipient did not consent to: $\mathrm{epoch}_p(r)$ is odd only after $r$ itself offered or accepted at that epoch, so absorbing it restores a friendship $r$ has already willed.

%% file: sections/secure-gsg-cva-restore.tex
\subsection{Restore Protocol}\label{section:secure-gsg-cva-restore}

We present the CVA implementation of the abstract Recover transaction (Section~\ref{section:secure-gsg-abstract}): the recovery by which an agent $p$, after a state loss that reset its machine state, rebuilds $\mathrm{FMap}_p$ and $\mathrm{FoFMap}_p$ from its friends.  Recovery is entirely passive: it requires no custodian authorisation, no dedicated recovery message, and no action by $p$ beyond resuming under its retained key.  It rests on the periodic re-broadcast already specified (Section~\ref{section:secure-gsg-cva-rebroadcast}).

\mypara{Fault model}  A state loss at $p$ (the abstract state-loss fault, Definition~\ref{def:state-loss-fault}) resets its machine state---$\mathit{known}_p$, $o_p$, $i_p$, $\mathrm{FMap}_p$, $\mathrm{FoFMap}_p$---to the initial empty values.  Only the intrinsic identity record $\mathit{IR}_p=(K_p,\sigma_p)$ survives, as it is an immutable attribute of the agent rather than machine state.  The retained key lets $p$ resume under the same identity; no custodian action is required, in contrast with the key-replacement Replace of Section~\ref{section:secure-gsg-cva-replace}.

\mypara{Recovery}  After the state loss $p$ resumes with empty $\mathrm{FMap}_p$ and $\mathrm{FoFMap}_p$.  Each friend $q$ of $p$ still records $p$---state loss touches only $p$'s own state---so $q$ continues to re-broadcast its periodic $\mathsf{checkpoint}(\mathit{IR}_q,L_q)$ to $p$, with $p\in\mathrm{Rec}_q$ and $(p,\mathrm{epoch}_q(p))\in L_q$.  On the first such checkpoint $p$ receives from $q$, Integrate checkpoint (Section~\ref{section:secure-gsg-cva-rebroadcast}) finds $q\notin\mathrm{dom}(\mathrm{FMap}_p)$ and installs the skeleton $\mathrm{FMap}_p[q]:=(0,\bot,K_q,\sigma_q)$ from the carried identity record $\mathit{IR}_q=(K_q,\sigma_q)$; the direct-heal clause then absorbs $(p,\mathrm{epoch}_q(p))\in L_q$, advancing $\mathrm{epoch}_p(q)$ to $\mathrm{epoch}_q(p)$ and preserving the identity-record fields, and the observer-heal clause rebuilds $\mathrm{FoFMap}_p[q]$ from the remaining pairs of $L_q$.  Restore adds no platform transaction and no cargo of its own: recovery is carried entirely by Integrate checkpoint.

\mypara{Coverage}  Every friend $q$ that still records $p$ re-broadcasts to $p$ infinitely often (Section~\ref{section:secure-gsg-cva-rebroadcast}), so under reliable delivery $p$ receives a checkpoint from each; at quiescence $\mathrm{FMap}_p$ holds $(q,\mathrm{epoch}_q(p))$ with $q$'s identity record, and $\mathrm{FoFMap}_p$ holds $q$'s friendships, for every such $q$.  Recovery is therefore complete and needs no custodian: $p$ recovers every friendship still recorded at the friend's side, with no residue.  This is the asymmetry with Replace, whose new identity is recorded by no friend and so cannot be recovered passively, requiring the active cascade of Section~\ref{section:secure-gsg-cva-replace}.  A checkpoint carries an odd epoch for the $q$--$p$ edge only when $q$ records that edge active, which $q$ did only after $p$'s own offer or acceptance (Message Bounds); so every active friendship $p$ recovers reflects one $p$ itself willed before the fault, and Friend List Soundness is preserved.

%% file: sections/secure-gsg-cva-replace.tex
\subsection{Replace Protocol}\label{section:secure-gsg-cva-replace}

We present the CVA implementation of the abstract Replace transaction (Section~\ref{section:secure-gsg-abstract}): the protocol by which an agent $p$ whose person has suffered identity loss is replaced by a fresh agent $p'$ across the social graph, authorised by a supermajority of $p$'s custodians.

\mypara{Protocol outline}  Replace proceeds in three phases.  Phase~1 is off-chain: the person behind $p$ chooses a new keypair, obtains a new machine if needed, and convinces a supermajority of $p$'s custodians to authorise the replacement.  Phase~2 is on-protocol vouching: each willing custodian $c\in K_p$ sends a $\mathsf{vouch}$ to $p'$ carrying its own view of $p$'s friends (to bootstrap $p'$'s knowledge of whom to notify).  Phase~3 is the cascade: $p'$, once a supermajority has vouched, sends $\mathsf{new\_identity}$ to every friend of $p$ it has learned about; each recipient $q$ verifies the vouches against its stored $\mathit{IR}_q[p]$, renames $p\to p'$ in $\mathrm{FMap}_q$ and $\mathrm{FoFMap}_q$ (preserving the epoch), and returns a $\mathsf{rebind}$ to $p'$; $p'$ integrates the rebinds to rebuild $\mathrm{FMap}_{p'}$ and $\mathrm{FoFMap}_{p'}$.  Each rebind also carries the rebinding friend's own view of $p$'s friends, which may reveal friends of $p$ that no voucher named; $p'$ notifies these in turn, so the cascade reaches every friend of $p$ recorded by a custodian or an already-reached friend.

\mypara{Phase 1}  Outside the protocol.  The person behind $p$ selects $p'$ (a fresh agent, not yet befriended by anyone) together with $\mathit{IR}_{p'}=(K_{p'},\sigma_{p'})$, and convinces custodians off-chain.

\subsubsection{Messages}\label{section:secure-gsg-cva-replace-cargo}

The cargo space $C$ is extended with:
\begin{itemize}
    \item $\mathsf{vouch}(p,p',L)$: custodian authorisation for the replacement of $p$ by $p'$, carrying a list $L$ of $(f,e)$ pairs drawn from the custodian's $\mathrm{FoFMap}[p]$, for bootstrapping $p'$'s knowledge of $p$'s friends.
    \item $\mathsf{new\_identity}(p,p',V,R')$: announcement that $p$ is being replaced by $p'$; $V\subseteq\Pi$ is the set of custodians that have vouched, $R'=\mathit{IR}_{p'}$ is $p'$'s identity record.
    \item $\mathsf{rebind}(e,R,L,L_p)$: acknowledgement sent by a friend $q$ to $p'$, carrying $q$'s epoch $e$ with $p$ (now renamed as with $p'$), $q$'s identity record $R$, a list $L$ of $(f,e')$ pairs drawn from $q$'s $\mathrm{FMap}$ snapshot (excluding $p'$), and a list $L_p$ of $(w,e'')$ pairs drawn from $q$'s view of the replaced identity's friends, $\mathrm{FoFMap}_q[p']$ (after the rename), for propagating the recovery to friends of $p$ that no voucher named.
\end{itemize}

\subsubsection{Platform Transactions}\label{section:secure-gsg-cva-replace-transactions}

\begin{definition}[Replace Platform Transactions]\label{definition:secure-gsg-cva-replace}
\leavevmode

\mypara{Vouch}  A unary transaction at $c\in\Pi$, guarded by $\{c\}$, for a person-chosen $p'\in\Pi\setminus\{c\}$ and a $p\in\mathrm{dom}(\mathrm{FMap}_c)$:
\begin{itemize}
    \item Let $L_c:=\{(f,e) : f\in\mathrm{dom}(\mathrm{FoFMap}_c[p]),\,f\ne p,\,e=\mathrm{FoFMap}_c[p][f]\}$.
    \item Add $\mathsf{message}(c,p',\mathsf{vouch}(p,p',L_c))$ to $o_c$.
\end{itemize}

\mypara{Announce new identity}  A unary transaction at $p'$, guarded by $\{p'\}$, for $p\in\Pi$ with $p\ne p'$, provided there exists $V\subseteq K_p$ with $|V|\ge\lceil\sigma_p\cdot|K_p|\rceil$ and, for every $c\in V$, some $\mathsf{message}(c,p',\mathsf{vouch}(p,p',L_c))\in i_{p'}$ (the $L_c$'s fixed by these messages).  Let the \emph{known friends of $p$} be
\[
  N \;:=\; V \;\cup\; \bigcup_{c\in V}\{f : (f,e)\in L_c\} \;\cup\; \{w : \mathrm{FMap}_{p'}[w]\text{ is a stub entry}\},
\]
where a \emph{stub} entry is one whose identity-record fields are $\bot$ (installed by Integrate rebind below).  For each $w\in N\setminus\{p,p'\}$ with $w\in\mathit{known}_{p'}$, $\mathrm{FMap}_{p'}[w]$ not a full entry, and no $\mathsf{message}(p',w,\mathsf{new\_identity}(p,p',V,\mathit{IR}_{p'}))\in o_{p'}$:
\begin{itemize}
    \item Add $\mathsf{message}(p',w,\mathsf{new\_identity}(p,p',V,\mathit{IR}_{p'}))$ to $o_{p'}$.
\end{itemize}
The values $K_p$ and $\sigma_p$ are known to the person behind $p'$ out-of-band, the same person being behind $p$ and $p'$.  The vouch messages are not consumed, so $V$ and the $L_c$'s are recomputable at each firing; $N$ grows only as Integrate rebind installs stubs for friends of $p$ revealed by rebinds.  Because $N$ is monotone and the agent set is finite, the transaction notifies each reachable friend of $p$ and reaches a fixpoint after finitely many rounds; a duplicate $\mathsf{new\_identity}$ is harmless, as Integrate new identity drops it once the recipient has renamed $p$.

\mypara{Integrate new identity}  A unary transaction at $q$, unguarded, provided\\ $\mathsf{message}(\cdot,q,\mathsf{new\_identity}(p,p',V,R'))\in i_q$, $p\in\mathrm{dom}(\mathrm{FMap}_q)$, and, writing $(K,\sigma):=\mathit{IR}_q[p]$, $V\subseteq K$ and $|V|\ge\lceil\sigma\cdot|K|\rceil$:
\begin{itemize}
    \item Let $e_q:=\mathrm{epoch}_q(p)$, and $(K',\sigma'):=R'$.
    \item Remove the entry $\mathrm{FMap}_q[p]$; set $\mathrm{FMap}_q[p']:=(e_q,K',\sigma')$.
    \item For every $y\in\mathrm{dom}(\mathrm{FoFMap}_q)$ with $y\ne p$: if $\mathrm{FoFMap}_q[y][p]$ is defined, move its value to $\mathrm{FoFMap}_q[y][p']$ and delete $\mathrm{FoFMap}_q[y][p]$.
    \item If $\mathrm{FoFMap}_q[p]$ is defined, move the entire inner map to $\mathrm{FoFMap}_q[p']$ and delete $\mathrm{FoFMap}_q[p]$.
    \item For every $f\in\mathrm{dom}(\mathrm{FMap}_q)$ with $p\in\mathrm{FMap}_q[f].K$: set $\mathrm{FMap}_q[f].K:=(\mathrm{FMap}_q[f].K\setminus\{p\})\cup\{p'\}$.
    \item Let $L_q:=\{(f,e) : f\in\mathrm{dom}(\mathrm{FMap}_q),\,f\notin\{p,p',q\},\,e=\mathrm{epoch}_q(f)\}$ and $L_q^p:=\{(w,e'') : w\in\mathrm{dom}(\mathrm{FoFMap}_q[p']),\,w\notin\{p',q\},\,e''=\mathrm{FoFMap}_q[p'][w]\}$ (both computed after the rename).  Add $\mathsf{message}(q,p',\mathsf{rebind}(e_q,\mathit{IR}_q,L_q,L_q^p))$ to $o_q$.
\end{itemize}

\mypara{Integrate rebind}  A unary transaction at $p'$, unguarded, provided\\ $\mathsf{message}(q,p',\mathsf{rebind}(e,R,L,L_p))\in i_{p'}$ with $q\ne p'$, and the state change below is non-trivial:
\begin{itemize}
    \item Let $(K,\sigma):=R$.  If $\mathrm{epoch}_{p'}(q)<e$, or $\mathrm{FMap}_{p'}[q]$ is a stub entry, set $\mathrm{FMap}_{p'}[q]:=(\max(\mathrm{epoch}_{p'}(q),e),K,\sigma)$ (promoting a stub to a full record).
    \item For each $(f,e')\in L$: if $\mathrm{epoch}_{p'}(q,f)<e'$, set $\mathrm{FoFMap}_{p'}[q][f]:=e'$ and $\mathrm{FoFMap}_{p'}[f][q]:=e'$ (preserving FoFMap Symmetry).
    \item For each $(w,e'')\in L_p$ with $w\notin\{p,p',q\}$ and $w\notin\mathrm{dom}(\mathrm{FMap}_{p'})$: install the stub entry $\mathrm{FMap}_{p'}[w]:=(e'',\bot,\bot)$, marking $w$ as a friend of $p$ to be notified by Announce new identity.
\end{itemize}
\end{definition}

 \mypara{Verification locality}  Integrate new identity verifies the vouching set $V$ against $q$'s own stored $\mathit{IR}_q[p]=(K_p,\sigma_p)$, which was received from $p$ upon initial befriend and is immutable within $p$'s identity.  No global custodian registry is consulted; each $r\in F_p$ verifies independently and consistently against the same identity record.  On the same evidence the transaction also replaces $p$ by $p'$ in every custodian set $q$ stores, so that a later Replace vouched for by $p'$ passes this guard.

\mypara{Changing custodians}  Although the custodian set of an identity is fixed, a person may replace their custodians by invoking an identity change.   In particular, should a custodian die, defect, or become unavailable, the person replaces their identity, choosing a fresh custodian set and threshold for the new one.

\mypara{Post-replacement isolation}  After $q$ executes Integrate new identity, $p\notin\mathrm{dom}(\mathrm{FMap}_q)$.  Every reactive transaction at $q$ whose precondition references an incoming message from $p$ additionally requires $p\in\mathrm{dom}(\mathrm{FMap}_q)$ and therefore fails.  Messages from the compromised old key are silently ignored.

\mypara{Inactive friendships}  Integrate new identity makes no case split on the parity of $\mathrm{epoch}_q(p)$; both active and inactive friendships are renamed from $p$ to $p'$, carrying the pre-rename epoch over.  An inactive friendship between $q$ and $p$ becomes an inactive friendship between $q$ and $p'$ at the same epoch, without re-establishing the friendship.  Should the person behind $q$ later wish to re-friend, a fresh Befriend with $p'$ is initiated via the normal protocol (Section~\ref{section:gsg-cva-befriend}), starting a new active epoch.

\mypara{Chained identity}  A friend may miss one or more replacements of a given agent.  If $p$ is replaced in succession, $p\to p'\to p''$, and a friend $q$ processes $\mathsf{new\_identity}(p',p'',\cdot,\cdot)$ without having processed $\mathsf{new\_identity}(p,p',\cdot,\cdot)$, then $p'\notin\mathrm{dom}(\mathrm{FMap}_q)$, the precondition of Integrate new identity fails, and the message is dropped: $q$ has no record of $p'$ to rename.  Two recovery paths are available, the choice left to the person behind $q$.  (i)~$q$ verifies the chain of $\mathsf{new\_identity}$ announcements back to the most recent identity it recognises, checking each step against the identity record carried for the previous step --- the first step against the immutable $\mathit{IR}_q[p]$ it stored on befriend, each subsequent step against the identity record announced by the step before it --- and then applies the cumulative rename $p\to p''$, preserving the epoch.  (ii)~The person behind $q$ learns out of band, by the same means used to designate custodians, that $p''$ succeeds $p$, and wills the cumulative rename $p\to p''$: the retained $\mathsf{new\_identity}$ announcement in $i_q$ supplies $p''$'s identity record, the epoch is preserved, and $q$'s $\mathsf{rebind}$ promotes the stub $p''$ holds for $q$.  Both paths preserve the epoch.  If $q$ recognises no identity in the chain and cannot reach the person, it is a stranded friend in the sense of the residue below.
\udi{Resolution of the thread above, per Idit: recovery is only by $q$ learning that $p''$ succeeds $p$ --- by verifying the chain (path~i) or out of band from the person (path~ii) --- followed by the cumulative rename, epoch preserved.  The no-learning re-befriend is dropped: $q$'s offer can reach $p''$ before any vouch or rebind installs the stub for $q$, and $p''$ then completes a handshake at epoch $1$, breaking observer monotonicity --- Idit's race.}

\mypara{Coverage}  The cascade recovers every friendship between $p$ and a friend $w$ that was recorded, before the fault, at some friend of $p$ connected to the seed (Theorem~\ref{theorem:secure-gsg-cva-replace-convergence}).  The one edge it cannot reach is a $w$--$p$ friendship recorded nowhere but at $w$ and $p$ themselves, $p$'s $\mathsf{stream\_update}$ having reached no other agent before the fault.  This is the sole unrecoverable residue (Section~\ref{section:cva-impl-overview}, ``The unrecoverable fault''; Definition~\ref{def:friendship-preserving}).

%% file: sections/secure-gsg-cva-properties.tex
\subsection{Properties of the Security Layer}\label{section:secure-gsg-cva-properties}

We state the additional CVA-level invariants introduced by the security layer and the two main theorems: implementation of the abstract secure social graph spec, and post-quiescence convergence of the Replace cascade.  The base invariants (Knowledge Monotonicity, Friendship Monotonicity, Message Bounds, Observer Bound, FoFMap Symmetry; Section~\ref{section:gsg-cva-invariants}) and the safety properties proved from them carry over.  Their proofs enumerate the base writers; each writer the security layer adds either adopts an epoch its sender already held, re-files an entry under another identity of the same person without changing its epoch, or installs a skeleton at the inactive epoch $0$ that no odd-epoch claim sees.

\subsubsection{Invariants}\label{section:secure-gsg-cva-invariants}

\begin{restatable}[Half Friendship Mutuality at CVA]{lemma}{ThmHalfMutualityCva}\label{lemma:secure-gsg-cva-half-mutuality}
Consider a run of the secure social graph CVA implementation that uses only the Befriend, Unfriend, and
Replace transactions (no Periodic Re-broadcast and no Restore) in which, whenever an agent is replaced, its cascade reaches every
friend connected to its vouched seed through the disseminated graph (the
hypothesis of Theorem~\ref{theorem:secure-gsg-cva-replace-convergence}).  If
$\mathrm{epoch}_p(q)=x$ with $x$ odd at time $t$, then at $t$ at least one of:
\begin{enumerate}
  \item $\mathrm{epoch}_q(p)\in\{x,x+1\}$;
  \item a $\mathsf{friend\_request}$, $\mathsf{accept}$, $\mathsf{unfriend}$, or
        $\mathsf{rebind}$ message of epoch $\ge x$ between $p$ and $q$ is in transit;
  \item $q$ has been replaced---$q$ ran Replace to a new identity $q'$---and $p$
        has not yet integrated the rename, so a $\mathsf{new\_identity}(q,q',\cdot)$
        addressed to $p$ is in transit or in $i_p$.
\end{enumerate}
\end{restatable}

Base safety and the invariants carry over as above; base liveness carries over except at the unrecoverable residue.
\begin{corollary}[Secure Friendship Establishment]\label{cor:secure-gsg-cva-establishment}
On a friendship-preserving run (Definition~\ref{def:friendship-preserving}), if $p$ and $q$ each want to befriend the other from some time on, then from some time on $q\in\widetilde F_p$ and $p\in\widetilde F_q$, unless a Replace of $p$ or $q$ drops the $p$--$q$ edge as the unrecoverable residue.
\end{corollary}
\begin{proof}
By Friendship Establishment (Theorem~\ref{theorem:gsg-cva-establishment}) the base handshake completes and both record the edge active. Unfriend fires only on an End-friendship volition, excluded by persistent wanting; Restore reinstalls the edge from the friends' checkpoints; Replace renames it, reaching the counterpart on a friendship-preserving run (Theorem~\ref{theorem:secure-gsg-cva-replace-convergence}) unless the edge is the residue. Hence it persists or recovers, the residue excepted.
\end{proof}

\subsubsection{Implementation}\label{section:secure-gsg-cva-refinement}

\begin{definition}[Abstract Friend Set at CVA]\label{definition:secure-gsg-cva-abstract-friend-set}
At a secure social graph CVA state, relabel each replaced identity to the original it replaced.  The \emph{abstract friend set} of a live agent $p$ is
\[
  \widetilde F_p \;:=\; \{\,q\in\mathrm{dom}(\mathrm{FMap}_p) : \mathrm{epoch}_p(q)\text{ is odd}\,\}:
\]
the set of friendships of $p$ as recorded in $p$'s state.
\end{definition}

Every quiescent state of a friendship-preserving run is mapped correctly to a state of the secure social graph (Section~\ref{section:secure-gsg-abstract}), except for the unrecoverable residue (Remark below).

Two results are shown, paralleling the abstract level (Theorems~\ref{thm:ssg-implements-sg} and~\ref{thm:ssg-resilient}): the Correspondence maps each completed protocol to its abstract transaction at quiescence on friendship-preserving runs, the unrecoverable residue excluded; the Fault-Resilience then shows that a state-loss fault and its Restore net to a stutter at quiescence, extending the mapping to runs that contain faults.

\ThmSecureCorrespondence*

\ThmSecureImplementsSG*

\begin{restatable}[CVA Fault-Resilience]{theorem}{ThmCvaResilient}\label{thm:secure-gsg-cva-resilient}
Let $F''$ be the CVA state-loss faults: the resets of $\mathit{known}_p$, $o_p$, $i_p$, $\mathrm{FMap}_p$, and $\mathrm{FoFMap}_p$ to empty (Section~\ref{section:secure-gsg-cva-restore}), disjoint from the CVA transactions, and let $\sigma'':=(F_p\mapsto\widetilde F_p)$.  Restricted to friendship-preserving runs, the secure social graph CVA implementation is, at quiescence, an $F''$-resilient implementation of the abstract secure social graph (Definition~\ref{def:fault-resilient}): a state-loss fault and the Restore that follows it return $\widetilde F_p$ to its pre-fault value, so at quiescence they are a stutter and $\sigma''$ maps every live run over $T''\cup F''$ to a correct run of the abstract secure social graph.
\end{restatable}

\noindent As $\sigma$ maps correct runs of the abstract secure social graph to correct runs of the social graph (Theorem~\ref{thm:ssg-implements-sg}), the composite $\sigma\circ\sigma''$ is $F''$-resilient over the social graph at quiescence (Lemma~\ref{lemma:implementation-transitivity}).

The construction is modular: the base social graph CVA implementation (Appendix~\ref{section:gsg-cva}) is extended by the security layer, leaving the base transactions unchanged.

\subsubsection{Convergence}\label{section:secure-gsg-cva-convergence}

\begin{restatable}[Repair after Faults]{lemma}{ThmRepairAfterFaults}\label{lemma:secure-gsg-cva-repair}
Consider a friendship-preserving run with finitely many faults, the last at time $t_F$.  Then there is a time $t^*\ge t_F$ from which onward no staleness caused by the faults remains: $\mathrm{epoch}_p(q)$ agrees with $\mathrm{epoch}_q(p)$ except while a volitional transaction between $p$ and $q$ is in progress, and $\mathrm{epoch}_p(q,r)$ equals $\mathrm{epoch}_q(r)$ for every observer $p$.  From $t^*$ the run is indistinguishable from one of the same volitions in which no message was lost, so Friendship Establishment (Theorem~\ref{theorem:gsg-cva-establishment}), Friend-of-Friend Visibility (Theorem~\ref{theorem:gsg-cva-fof-visibility}) and Unfriend Propagation (Theorem~\ref{theorem:gsg-cva-unfriend-propagation}) hold with $t^*$ in place of the initial time.
\end{restatable}

\begin{definition}[Quiescence before Replace]\label{definition:secure-gsg-cva-quiescence}
A run is \emph{quiescent at $p$ before Replace} if, at the moment Announce new identity fires at $p'$, every $\mathsf{stream\_update}$ message $p$ has ever sent has been delivered and integrated by its recipient---in particular, every recipient $c\in K_p$ has integrated every $\mathsf{stream\_update}(w,\cdot)$ ever emitted by $p$.
\end{definition}

\begin{restatable}[Replace Convergence]{theorem}
{ThmReplaceConvergence}\label{theorem:secure-gsg-cva-replace-convergence}
Suppose the run is quiescent at $p$ before Replace (Definition~\ref{definition:secure-gsg-cva-quiescence}) and every $c\in K_p$ vouches (so $V=K_p$).  Then the cascade reaches every friend of $p$: each $w\in F_p$ receives $\mathsf{new\_identity}$, renames $p\to p'$, and rebinds.  After the cascade terminates, $\widetilde F_{p'}=F_p$ and for every $r\in F_p$ the friendship between $r$ and $p$ has been renamed to one between $r$ and $p'$, realising the abstract Replace exactly.
\end{restatable}

\mypara{Remark (the unrecoverable residue)}  Under the quiescence hypothesis the cascade reaches all of $F_p$ and recovery is exact; the identity-loss residue without it is the unrecoverable fault of Section~\ref{section:cva-impl-overview} (Definition~\ref{def:friendship-preserving}).

\begin{restatable}[Recovery Requires a Record]{lemma}{ThmIrrecoverable}\label{lemma:secure-gsg-cva-irrecoverable}
Let the friendship between $p$ and $w$ be unrecoverable at time $t$ (Definition~\ref{def:friendship-preserving}).  Then in every continuation in which neither person wills a fresh Offer friendship or acceptance toward the other, the two never again record each other.
\end{restatable}

%% file: sections/appendix-coins-bonds.tex
\input{sections/secure-bonds}

\input{sections/secure-coins-cva}

\section{Proofs for Grassroots Coins and Bonds}\label{appendix:coins-bonds-proofs}

\input{sections/proofs-secure-bonds}

\input{sections/proofs-secure-coins-cva}

%% file: sections/secure-bonds.tex
\section{Secure Grassroots Coins}\label{section:secure-bonds}

We extend the grassroots coins specification~\cite{shapiro2026bonds} with the sovereign's transaction log and state custodians, specified as guarded multiagent atomic transactions.

In grassroots coins~\cite{shapiro2024gc,shapiro2026bonds}, each agent issues their own currency: the issuer of a coin---its \emph{sovereign}---is the natural authority on transactions in that currency.  The Grassroots Flash payment system~\cite{lewis2023grassroots} realises this by having the sovereign approve each payment in their coins, with the sovereign's personal blockchain serving as the authoritative ledger.  Here, we abstract this as a transaction log $L_p$ maintained by each sovereign $p$, recording every transaction in $p$-coins.  Since $L_p$ is append-only, it is monotonic.  To enable recovery if $p$ suffers a fault, $p$ designates some of its friends as \emph{state custodians} at the formation of its currency, who hold copies of $L_p$.

\subsection{Specification with Guarded Transactions}

Each agent $p \in P$ has a machine state extending the secure social graph (Section~\ref{section:secure-gsg-abstract}) with:
\begin{itemize}
\item $S_p \subseteq F_p$: the \emph{state custodians} of $p$---a finite set of $p$'s friends, fixed at the formation of $p$'s currency and immutable thereafter.
\item $B_p$: the multiset of coins currently held by $p$.
\item $L_p$: the \emph{transaction log} of $p$-coins---an append-only sequence recording every transaction in $p$-coins.
\end{itemize}
A sovereign forms its currency once, before any transaction in its coins, fixing the custodian set $S_p$; until then $S_p$ is undefined.  Initially $B_p = \emptyset$ and $L_p = \varepsilon$ (the empty sequence) for all $p$, and on formation $L_p^r = \varepsilon$ for each $r \in S_p$.  Each state custodian $r \in S_p$ holds a copy $L_p^r$ of $L_p$ as part of $r$'s machine state.  The custodian set is intrinsic and immutable.  $\mathrm{Replace}(p,p')$ (Definition~\ref{definition:secure-gsg-transactions}) extends to the coins state: $S'_{p'} := S_p$, $B'_{p'} := B_p$, $L'_{p'} := L_p$, and each custodian copy $L_p^r$ is relabelled $L_{p'}^r$.

\begin{definition}[Valid Log, Derived Holdings]\label{definition:valid-log}
A sequence of transaction records is \emph{valid} if, executed in order from the initial state, each record's transaction is enabled in the state produced by the records before it.  The \emph{holdings derived from} a valid log are the coin holdings of the state it produces.  In particular, a coin is \emph{held by} an agent in a valid log if the derived holdings assign it to that agent, and \emph{unspent} by a prospective payer if the derived holdings still assign it to that payer.
\end{definition}

The transactions---Mint, Pay, Redeem, and Swap---correspond to the grassroots coins transactions~\cite{shapiro2026bonds}, extended with log updates at the sovereign and its state custodians.

\begin{definition}[Secure Coins Transactions]\label{definition:secure-bonds-transactions}
The \temph{secure coins transactions} are:
\begin{enumerate}
    \item \textbf{Form currency$(p, S)$}: where $S \subseteq F_p$ and $p$ has not yet formed its currency.  $S'_p := S$; for each $r \in S$: $L_p^{r} := \varepsilon$.  Participants $\{p\} \cup S$.  Guarded by $p$.

    \item \textbf{Mint$(p, k, t)$}: $B'_p := B_p \cup \text{\textcent}^k_{p,t}$, $L'_p := L_p \cdot [\text{mint}(k,t)]$.  For each $r \in S_p$: $L_p^{r\prime} := L'_p$.  Participants $\{p\}\cup S_p$.  Guarded by $p$.

    \item \textbf{Pay$(q, r, x)$}: where $x \subseteq B_q$ is a set of $s$-coins for some sovereign $s$.  $B'_q := B_q \setminus x$, $B'_r := B_r \cup x$, $L'_s := L_s \cdot [\text{pay}(q,r,x)]$.  For each $u \in S_s$: $L_s^{u\prime} := L'_s$.  Participants $\{q,r,s\}\cup S_s$.  Guarded by $q$.

    \item \textbf{Redeem$(q, s)$}: where $\text{\textcent}_s \in B_q$ and $\text{\textcent}_{r,t} \in B_s$.  $B'_q := (B_q \setminus \{\text{\textcent}_s\}) \cup \{\text{\textcent}_{r,t}\}$, $B'_s := (B_s \setminus \{\text{\textcent}_{r,t}\}) \cup \{\text{\textcent}_s\}$, $L'_s := L_s \cdot [\text{redeem}(q,\text{\textcent}_s,\text{\textcent}_{r,t})]$.  For each $u \in S_s$: $L_s^{u\prime} := L'_s$.  Participants $\{q,s\}\cup S_s$.  Guarded by $q$.

    \item \textbf{Swap$(p, q, x, y)$}: where $x \subseteq B_p$, $y \subseteq B_q$.  $B'_p := (B_p \setminus x) \cup y$, $B'_q := (B_q \setminus y) \cup x$.  For each sovereign $s$ whose coins appear in $x \cup y$: $L'_s := L_s \cdot [\text{swap}(p,q,x_s,y_s)]$, and for each $u \in S_s$: $L_s^{u\prime} := L'_s$.  Participants $\{p,q\}$ together with each such sovereign $s$ and its custodians $S_s$.  Guarded by $\{p, q\}$.
\end{enumerate}
\end{definition}

In each transaction in $p$-coins, the sovereign $p$ is a participant: $p$'s log grows, and all state custodians' copies are updated atomically.  Since every transaction in $p$-coins is enabled only when the coins it moves are held by the agent it debits, the log $L_p$ is valid by construction, and the holdings of $p$-coins by any agent are the holdings derived from $L_p$ (Definition~\ref{definition:valid-log}).  Pay and Redeem are guarded by the payer/redeemer alone---the sovereign's person need not consent---but the sovereign's machine must be live for the transaction to execute.  If the sovereign is unresponsive, transactions in their currency are blocked; holders who lose trust may treat the sovereign's coins as bad debt~\cite{shapiro2026bonds}.

\subsection{Properties}

\mypara{Safety}

\begin{restatable}[Log Consistency]{lemma}{ThmLogConsistency}\label{lemma:log-consistency}
In any run of the secure coins protocol, $L_p^r = L_p$ for every state custodian $r \in S_p$.
\end{restatable}

\begin{restatable}[Conservation of Money]{lemma}{ThmConservation}\label{lemma:conservation}
In any run of the secure coins protocol, the $p$-coins across all agents' holdings are exactly the $p$-coins recorded as minted in $L_p$.
\end{restatable}

\mypara{Recovery}

Recovery of the social graph proceeds via Replace (Definition~\ref{definition:secure-gsg-transactions}).  For the sovereign's transaction log, the person behind the recovered agent contacts any available state custodian to obtain a copy.  Since log copies are maintained in sync (Lemma~\ref{lemma:log-consistency}), a single available state custodian suffices.

\begin{restatable}[Coin Recovery]{theorem}{ThmBondRecovery}\label{theorem:bond-recovery}
Let $p$ be an agent with $K_p \neq \emptyset$ and $S_p \neq \emptyset$.  If $\mathrm{Replace}(p,p')$ is taken and some $r \in S_p$ is correct, then $F_{p'} = F_p$, $S_{p'} = S_p$, and $L_p^r = L_p$.
\end{restatable}

By Lemma~\ref{lemma:conservation}, the recovered log $L_p^r$ determines the correct holdings of all $p$-coins, enabling the recovered agent $p'$ to resume as sovereign.

\mypara{Grassroots}

\begin{restatable}{theorem}{ThmBondsGrassroots}\label{theorem:bonds-grassroots}
The secure grassroots coins protocol is grassroots.
\end{restatable}

\mypara{Recovery without key loss}

If a sovereign $p$ suffers a state loss---losing machine state but retaining their key---coin recovery is a special case of Recover (Definition~\ref{definition:secure-gsg-transactions}).  The person authenticates to any state custodian $r \in S_p$ using their retained key and obtains $L_p^r = L_p$ (Lemma~\ref{lemma:log-consistency}).  The recovered log determines the correct holdings of all $p$-coins (Lemma~\ref{lemma:conservation}), and the sovereign resumes operation.  For a trader $q$ who suffers a state loss and loses their local coin cache, the person behind $q$ contacts each friend sovereign $s$ after Recover and asks whether $q$ holds any $s$-coins; each sovereign can answer from their log $L_s$.  No custodian authorization is required in either case.

%% file: sections/secure-coins-cva.tex
\section{Secure Grassroots Coins: CVA Implementation}\label{section:secure-coins-cva}

We implement the secure coins specification (Section~\ref{section:secure-bonds}) as a CVA platform (Section~\ref{sec:cva}), extending the secure social graph CVA implementation (Section~\ref{section:secure-gsg-cva}).  At the specification the state custodians are updated in the same atomic step as the sovereign, so each holds the complete log and a single available custodian suffices for exact recovery.  An asynchronous implementation has no such atomicity: a custodian's copy of the log may lag the sovereign's.  Recovering from a single lagging custodian could reinstate a log that omits an already-approved payment, and the recovered sovereign might then approve the same coin a second time---an inadvertent double-spend.  We instead restore exactness by coupling finality to recoverability through a supermajority of custodians: a payment is final, and a payee may rely on it, only once a supermajority of the sovereign's state custodians hold it, and recovery collects copies from a supermajority.  Any two supermajorities of the same set intersect, so every final payment is recovered. 

This is the supermajority-intersection structure of the All-to-All Flash payment system~\cite{lewispye2023flash}, here in two simpler forms.  Our fault model is crash-only (Section~\ref{section:secure-gsg-abstract}), so the quorums are plain intersecting read and write quorums---a write quorum and a read quorum of a set of size $n$ intersect once each exceeds $n/2$---with none of the slack a Byzantine setting requires.  And the sovereign is the sole writer of its log, so the task is the durability of one writer's log across its custodians, not agreement among competing writers; no consensus is invoked.

Every CVA protocol is grassroots by construction (Theorem~\ref{thm:cva-grassroots}); the platform below is therefore grassroots without a separate argument.

\subsection{Platform State}\label{section:secure-coins-cva-state}

The platform state of $p$ extends the secure social graph CVA state (Section~\ref{section:secure-gsg-cva-state}) with the components of a sovereign, a holder, and a custodian:
\begin{itemize}
    \item $e_p \in \mathbb{N}$: the current epoch number at agent $p$,  initially $0$.
    \item $B_p$: the multiset of \emph{coins} held by $p$, initially empty.  This is a conservative local account, debited when $p$ requests a payment and credited when a payment to $p$ becomes final; the authoritative record of holdings is the sovereigns' logs, from which holdings are derived (Definition~\ref{definition:valid-log}), and $B_p$ is $p$'s local view of its own entry in them.
    \item $L_p$: the \emph{log} of $p$'s currency, a sequence of blocks; blocks are associated with monotonically increasing epoch numbers; 
    the block at position $n\in\mathbb{N}$ in epoch $e$ is $\mathsf{block}(e,n,\tau)$ with $\tau$ a \emph{transaction record}, one of $\mathsf{mint}(k,t)$, $\mathsf{pay}(q,r,x)$, $\mathsf{redeem}(q,c,c')$ with $c$ the bond redeemed and $c'$ the coin returned, or $\mathsf{swap}(q,r,x,y)$.  Initially empty.
    \item $\ell_p:S_p\to\mathbb{N} \times \mathbb{N}$, the \emph{ack frontier}: $\ell_p(c)$ is the maximal pair $e,n$
    (pairs ordered lexicographically) associated with a block
    that custodian $c$ has acknowledged holding, initially $0,0$.
     \item for each sovereign $s$ with $p\in S_s$, a copy $\widehat L_p[s]$ of $L_s$ together with the epoch $\widehat e_p[s]\in\mathbb{N}$ of $s$'s currency that $p$ has registered and the epoch $\widehat\pi_p[s]\ge\widehat e_p[s]$ that $p$ has promised to a recovering $s$ (Recovery below), created empty and at epochs $0$ when $p$ is designated a custodian of $s$ (Become custodian below).  In the absence of failures, $\widehat L_p[s]$ is a prefix of $L_s$ and $\widehat e_p[s]=\widehat\pi_p[s]=e_s$.
\end{itemize}
The state-custodian record of $p$'s currency is the set $S_p\subseteq F_p$, fixed at currency formation and immutable thereafter (Section~\ref{section:secure-bonds}), together with a \emph{state-custodian threshold} $\sigma^S\in(1/2,1]$, taken global to the platform; making $\sigma^S$ specific to each sovereign changes nothing below.  A \emph{supermajority} of $S_p$ is a subset $G\subseteq S_p$ with $|G|\ge\lceil\sigma^S\cdot|S_p|\rceil$.  Since $\sigma^S>1/2$, any two supermajorities of $S_p$ intersect.

The custodians and the holders that transact in a currency are friends of its sovereign---custodians by $S_p\subseteq F_p$, holders by befriending the sovereign before transacting---so every message below travels between friends, whose channels the secure social graph layer maintains (Section~\ref{section:secure-gsg-cva}); the CVA known-peer precondition on emission (Section~\ref{sec:cva}) is thereby met.  A payer communicates only with the sovereign and the sovereign only with its custodians; a payer need not be a friend of the custodians.

\subsection{Messages}\label{section:secure-coins-cva-cargo}

The cargo space $C$ is extended with, for a transaction record $\tau$, epoch $e\in\mathbb{N}$, position $n\in\mathbb{N}$, and sovereign $s$:
\begin{itemize}
    \item $\mathsf{become\_custodian}$: sovereign $s$ designating the recipient as a state custodian of $s$-coins.
    \item $\mathsf{request}(\tau)$: a holder's request that sovereign $s$ approve the transaction $\tau$ in $s$-coins.
    \item $\mathsf{block}(e,n,\tau)$: the sovereign disseminating its $n$-th log block in epoch $e$ to a custodian.
    \item $\mathsf{ack}(e,n)$: a custodian attesting that it holds the prefix of the sovereign's log through position $n$ in epoch $e$.
    \item $\mathsf{final}(e,n,\tau)$: the sovereign notifying a party to $\tau$ that block $n$ in epoch $e$, which records $\tau$, is final.
\end{itemize}

\subsection{Transactions}\label{section:secure-coins-cva-transactions}

The protocol realises each spec transaction (Definition~\ref{definition:secure-bonds-transactions}) as a sequence of unary CVA platform transactions: a holder requests, the sovereign approves and appends a block and disseminates it, each custodian appends and acknowledges, and the sovereign binds finality once a supermajority have acknowledged.  A sovereign first forms its currency, fixing its custodian set and notifying each custodian, which registers an empty log copy.  We give the protocol for Pay; Mint, Redeem, and Swap differ only in the transaction record carried and in which holdings are updated on finality.

\begin{definition}[Secure Coins CVA Transactions]\label{definition:secure-coins-cva-transactions}
\leavevmode

\mypara{Form currency}  A unary transaction at sovereign $p$, guarded by $\{p\}$, for a chosen $S\subseteq F_p$, provided $S_p$ is undefined (the currency is not yet formed):
\begin{itemize}
    \item Set $S_p:=S$ and $\ell_p(c):=(0,0)$ for each $c\in S$.
    \item For each $c\in S$, add $\mathsf{message}(p,c,\mathsf{become\_custodian})$ to $o_p$.
\end{itemize}

\mypara{Become custodian}  A unary transaction at $c$, unguarded, provided $\mathsf{message}(s,c,\mathsf{become\_custodian})\in i_c$ and $\widehat L_c[s]$ is undefined:
\begin{itemize}
   \item Set $\widehat L_c[s]:=\varepsilon$ and $\widehat e_c[s]:=\widehat\pi_c[s]:=0$, registering $c$ as a state custodian of $s$.
\end{itemize}

\mypara{Request pay}  A unary transaction at a holder $q$, guarded by $\{q\}$, for a recipient $r$, a set $x\subseteq B_q$ of $s$-coins of some sovereign $s\in\mathit{known}_q$:
\begin{itemize}
    \item Remove $x$ from $B_q$.
    \item Add $\mathsf{message}(q,s,\mathsf{request}(\mathsf{pay}(q,r,x)))$ to $o_q$.
\end{itemize}

\mypara{Approve}  A unary transaction at sovereign $s$, unguarded, provided $\mathsf{message}(q,s,\mathsf{request}(\mathsf{pay}(q,r,x)))\in i_s$ and $x$ is unspent in $L_s$ (Definition~\ref{definition:valid-log})---each coin in $x$ is, in the holdings derived from $L_s$, currently held by $q$:
\begin{itemize}
     \item Let $n:=|L_s|+1$; append $\mathsf{block}(e_s,n,\mathsf{pay}(q,r,x))$ to $L_s$.
    \item For each $c\in S_s$, add $\mathsf{message}(s,c,\mathsf{block}(e_s,n,\mathsf{pay}(q,r,x)))$ to $o_s$.
\end{itemize}
A request that is not unspent in $L_s$ is dropped without appending.

\mypara{Append and acknowledge}  A unary transaction at custodian $c\in S_s$, unguarded, provided $\mathsf{message}(s,c,\mathsf{block}(e,n,\tau))\in i_c$, $e=\widehat e_c[s]$, and $n=|\widehat L_c[s]|+1$ (the next block in order):
\begin{itemize}
    \item Append $\mathsf{block}(e,n,\tau)$ to $\widehat L_c[s]$.
    \item Add $\mathsf{message}(c,s,\mathsf{ack}(e,n))$ to $o_c$.
\end{itemize}
A block whose epoch exceeds $\widehat e_c[s]$, or whose position is not next in order, is retained in $i_c$ until $c$ registers that epoch or the predecessors arrive; a block of an earlier epoch, or one already held, is stale and dropped.

\mypara{Bind finality}  A unary transaction at sovereign $s$, unguarded, provided $\mathsf{message}(c,s,\mathsf{ack}(e_s,n))\in i_s$ for a custodian $c\in S_s$ with $\ell_s(c)<(e_s,n)$:
\begin{itemize}
    \item Set $\ell_s(c):=(e_s,n)$.
    \item For each position $m$ in epoch $e_s$ that has become final at this step---$|\{c'\in S_s:\ell_s(c')\ge(e_s,m)\}|\ge\lceil\sigma^S|S_s|\rceil$ for the first time---with $\mathsf{block}(e_s,m,\tau)$ in $L_s$, add $\mathsf{message}(s,a,\mathsf{final}(e_s,m,\tau))$ to $o_s$ for each party $a$ of $\tau$: the payer $q$ and payee $r$ of a $\mathsf{pay}(q,r,x)$; both parties $q,r$ of a $\mathsf{swap}(q,r,x,y)$; the redeemer $q$ and $s$ for a $\mathsf{redeem}(q,\cdot,\cdot)$; and $s$ for a $\mathsf{mint}(k,t)$.
\end{itemize}

\mypara{Accept payment}  A unary transaction at recipient $r$, unguarded, provided $\mathsf{message}(s,r,\mathsf{final}(e, n,\mathsf{pay}(q,r,x)))\in i_r$ with $r$ the payee:
\begin{itemize}
    \item Add $x$ to $B_r$.
\end{itemize}
\end{definition}

An $\mathsf{ack}(e,n)$ attests a prefix through position $n$ in epoch $e$, so it counts toward the finality of every position $m\le n$ of that epoch; a custodian that has acknowledged any $n\ge m$ in epoch $e$ is among the holders of block $(e,m)$. Transactions in old epochs are already finalised in the recovery process before the epoch begins, and therefore only blocks in the current epoch are considered here. Bind finality therefore finalises, in one step, every position up to the new frontier that the new acknowledgment carries past the threshold.  Mint, Redeem, and Swap follow the same four-transaction protocol: Mint is requested by the sovereign's own person and adds the minted coins to $B_s$ on approval, Redeem and Swap carry their own records and update the relevant holdings on finality; in each the sovereign appends a block, disseminates it, and binds finality at a supermajority of acknowledgments exactly as for Pay.

\subsection{Finality}\label{section:secure-coins-cva-finality}

\begin{definition}[Finality]\label{definition:secure-coins-cva-finality}
A block at position $(e,n)$ of $L_s$ is \temph{final} when a supermajority of $s$'s state custodians hold the prefix of $L_s$ through $(e,n)$.
A payment is \temph{final} when its block is final.
\end{definition}

A payee relies on a payment only once it is final---only upon receiving $\mathsf{final}$, which the sovereign sends only after a supermajority have acknowledged.  Finality is the line between safe-to-rely and not: a payment the sovereign has approved but a supermajority has not yet acknowledged is not final, may be dropped on recovery, and is safe to drop precisely because no payee was told it was final.

A payment commits when it is requested: Request pay debits the payer's account at once, and a payment that is never approved (the request is dropped) or never reaches a supermajority is not refunded.  This loses the payer nothing of value.  A payment in $s$-coins can fail to finalise only if the sovereign $s$ or a supermajority of its custodians is permanently unavailable, or if $s$ loses its state before a supermajority hold the payment's block, in which case recovery drops it; in the first case no $s$-coin can ever finalise again, so every $s$-coin is worthless, and the coins stranded in the failed payment are worth no less than the rest, while in the second the payer loses the coins of a payment that was never final and that no payee relied on.  An unresponsive sovereign renders its currency bad debt regardless~\cite{shapiro2026bonds}; the conservative account merely reflects this.

\subsection{Recovery}\label{section:secure-coins-cva-recovery}

The state-custodian record of a currency---the set $S_p$ together with the threshold $\sigma^S$---is fixed at currency formation and is intrinsic in the same sense as the identity record (Section~\ref{section:secure-gsg-cva-state}): a read-only datum that survives a fault in which platform state is reset.  A recovering sovereign therefore knows from whom to collect and how many copies it needs before it has recovered any log---knowledge that cannot be made to depend on the log itself, which is what is being recovered.

The exchanges in the procedures below are conducted out-of-band by the recovering person and are not modelled as platform cargo; only their outcome matters to the properties---the promises, registered epochs and log copies of a supermajority for a sovereign, or the answer the sovereign reads from its log for a holder.

\mypara{State loss}  A sovereign $p$ that suffers a state loss and retains its key resumes with empty $L_p$ and recovers in three exchanges with its custodians.  It first queries them and waits for replies from a supermajority $G\subseteq S_p$, each custodian $c$ answering with its promised epoch $\widehat\pi_c[p]$, its registered epoch $\widehat e_c[p]$, and its copy $\widehat L_c[p]$.  It then reserves the epoch $e':=1+\max_{c\in G}\widehat\pi_c[p]$: a custodian with $\widehat\pi_c[p]<e'$ sets $\widehat\pi_c[p]:=e'$ and acknowledges, and $p$ waits for acknowledgements from a supermajority.  Finally $p$ adopts as $L_p$ the longest copy among those in $G$ registered in the maximal epoch $\max_{c\in G}\widehat e_c[p]$, sends it with $e'$ to every custodian, sets $e_p:=e'$, and resumes approving, its next block being $\mathsf{block}(e',|L_p|+1,\cdot)$; a custodian with $\widehat\pi_c[p]=e'$ and $\widehat e_c[p]<e'$ sets $\widehat L_c[p]:=L_p$ and $\widehat e_c[p]:=e'$, which may discard blocks of earlier epochs that never reached a supermajority and so were never final.  Since a supermajority promised $e'$ and any two supermajorities intersect, every later recovery hears a promise of at least $e'$ and reserves above it: an epoch is reserved, and a log installed under it, by one recovery only, so all copies registered in an epoch are prefixes of one log.  A custodian whose promise exceeds $e'$, left by an interrupted recovery, sits out until a later recovery reserves above it. 

\mypara{Identity loss}  A sovereign that loses its key is replaced by a fresh agent $p'$ through the secure social graph (Section~\ref{section:secure-gsg-cva-replace}).  The state-custodian record $(S_p,\sigma^S)$ is known out-of-band to the person behind $p'$---the same person behind $p$ and $p'$, exactly as the new identity record is in Replace Phase~1 (Section~\ref{section:secure-gsg-cva-replace})---so $p'$ adopts $S_{p'}=S_p$; the Replace cascade itself carries no log or custodian record. Once the custodians have processed the rename and address $p'$, $p'$ proceeds exactly as under state loss.

\mypara{Holder recovery}  A holder $q$ that suffers a state loss and loses $B_q$ recovers it from the sovereigns: for each friend sovereign $s$, the person behind $q$ asks $s$ which $s$-coins $q$ holds, which $s$ answers from $L_s$.  No custodian is involved.

\mypara{Changing custodians}  The custodian set is fixed for the life of a currency; to change custodians, a sovereign issues a new currency with a fresh custodian set and migrates holdings to it, which the existing Mint and Swap transactions already express.

\mypara{Custodian availability and delivery}  The properties below assume that, for each currency, a supermajority of its state custodians remains correct and retains its log copy, and that delivery between a sovereign and its custodians is reliable. Under this assumption a custodian copy changes only by the protocol---appending, or aligning at a recovery, which discards no final block, and no disseminated block is lost, so Prefix Consistency and Exact Recovery hold as stated. This is the standard quorum assumption: were a supermajority to fail, the currency would be permanently unable to finalise or recover and its coins would be worthless regardless.  It is also where the coins layer differs from the secure social graph, which tolerates message loss and the loss of a counterpart copy through periodic re-broadcast (Section~\ref{section:secure-gsg-cva-rebroadcast}); a custodian re-synchronisation that repairs a regressed or lagging copy, and so lifts the reliable-delivery assumption, is left to future work.

\subsection{Properties}\label{section:secure-coins-cva-properties}

\begin{restatable}[Prefix Consistency]{lemma}{ThmPrefixConsistency}\label{lemma:secure-coins-cva-prefix}
In any run, for every sovereign $s$ and custodian $c\in S_s$ registered in the current epoch, $\widehat L_c[s]$ is a prefix of $L_s$, and within an epoch both $L_s$ and $\widehat L_c[s]$ only grow.
\end{restatable}

\begin{restatable}[Log Validity]{lemma}{ThmLogValidity}\label{lemma:secure-coins-cva-log-validity}
In any run, every sovereign's log $L_s$ is valid (Definition~\ref{definition:valid-log}).
\end{restatable}

\ThmCoinsConservation*

\begin{restatable}[Finality Soundness]{lemma}{ThmFinalitySoundness}\label{lemma:secure-coins-cva-finality-soundness}
If a payee receives $\mathsf{final}(e,n,\mathsf{pay}(q,r,x))$ from $s$, then $\mathsf{block}(e,n,\mathsf{pay}(q,r,x))$ is in $L_s$ and was approved by $s$ against an unspent $x$.
\end{restatable}

\mypara{Correspondence with the specification}  The CVA protocol implements the secure coins specification (Section~\ref{section:secure-bonds}) in the sense that its quiescent states are exactly the specification's states.  We read an abstract state off a CVA configuration by the map $\alpha$ that takes the holdings $B_p$ of each agent, the log $L_p$ of each sovereign, and, for each custodian, its copy $\widehat L_p[s]$.

\ThmCorrespondenceQ*

The correspondence is at quiescence, not at every state, and necessarily so: between Request pay and the finalisation of that payment a coin is debited from the payer and not yet credited to any payee nor recorded in any final block, a configuration the atomic specification---in which a payment moves a coin and records it in one step---never exhibits.  The in-flight coin is below the abstraction; quiescence is exactly the condition under which no payment is in flight.

\ThmExactRecovery*

\begin{restatable}[Coins Fault-Resilience]{theorem}{ThmCoinsResilient}\label{theorem:secure-coins-cva-resilient}
Let $F$ be the coins state-loss faults: the resets of a sovereign's $L_p,\ell_p$ and custodian copies, a custodian's $\widehat L_p[\cdot]$, or a holder's $B_p$ to empty (Section~\ref{section:secure-coins-cva-recovery}), disjoint from the coins transactions.  Under the custodian-availability and reliable-delivery assumption (Section~\ref{section:secure-coins-cva-recovery}), the secure coins CVA implementation is, at quiescence, an $F$-resilient implementation of the secure coins specification (Definition~\ref{def:fault-resilient}) under the map $\alpha$: a state-loss fault and the recovery that follows it leave the derived holdings and every final block unchanged, so at quiescence they are a stutter and $\alpha$ maps every live run over the coins transactions and $F$ to a correct run of the specification.
\end{restatable}

\ThmNoDoubleSpend*

\ThmPaymentLiveness*

\begin{restatable}[Grassroots]{theorem}{ThmCoinsGrassroots}\label{theorem:secure-coins-cva-grassroots}
The secure coins and bonds CVA platform is grassroots.
\end{restatable}

%% file: sections/proofs-secure-bonds.tex
\ThmLogConsistency*
\begin{proof}
Initially $L_p^r = L_p = \varepsilon$ for every $r \in S_p$.  Every transaction involving $p$-coins atomically appends the same entry to $L_p$ and to $L_p^r$ for all $r \in S_p$.  No transaction modifies $L_p$ or $L_p^r$ independently.
\end{proof}

\ThmConservation*
\begin{proof}
Mint adds $p$-coins to $B_p$ and records the mint in $L_p$.  Pay, Redeem, and Swap transfer coins between agents without creating or destroying them, and record the transfer in the relevant sovereign's log.  No transaction creates coins without a corresponding log entry or destroys coins without recording the transfer.
\end{proof}

\ThmBondRecovery*
\begin{proof}
$F_{p'} = F_p$ follows from Replace (Definition~\ref{definition:secure-gsg-transactions}).  $S_{p'} = S_p$ follows from the inheritance of the custodian set under Replace.  $L_p^r = L_p$ follows from Lemma~\ref{lemma:log-consistency}.
\end{proof}

\ThmBondsGrassroots*
\begin{proof}
\emph{Obliviousness:}  Let $P, P' \subset \Pi$ be disjoint and nonempty.  By the obliviousness argument of Theorem~\ref{theorem:secure-gsg-grassroots}, no social graph transaction with participants spanning both $P$ and $P'$ is ever enabled in any interleaving.  For the coin transactions: Form currency$(p,S)$ is guarded by $p$, with participants $\{p\}\cup S\subseteq\{p\}\cup F_p$, which lie in a single group in any interleaving of disjoint runs; Mint is unary.  Pay$(q,r,x)$ in $s$-coins requires $x \subseteq B_q$, where $x$ consists of $s$-coins; in an interleaving, an agent in $P$ never acquires coins issued by agents in $P'$ (as this would require a prior Swap spanning both groups), so no Pay or Redeem with participants spanning both groups is enabled.  Swap$(p,q,x,y)$ with $p \in P$, $q \in P'$ requires $x \subseteq B_p$ and $y \subseteq B_q$; by the same argument, $p$ holds no $P'$-coins and $q$ holds no $P$-coins, so no swap with nonempty exchange spanning both groups is enabled.

\emph{Interactivity:}  Let $p \in P$, $q \in P'$.  A run in which $p$ and $q$ befriend and then execute a voluntary swap---$p$ mints $p$-coins, $q$ mints $q$-coins, and they exchange---changes the local states of both $p$ and $q$, producing a configuration not reachable by either group independently.  Hence the protocol is interactive, and by the interactivity criterion of~\cite{lewis2026volitional}, grassroots.
\end{proof}

%% file: sections/proofs-secure-coins-cva.tex
\ThmPrefixConsistency*
\begin{proof}
By induction on the run.  Both are initially empty, and a custodian registering the current epoch aligns its copy to a prefix of $L_s$.  Approve appends one block, in epoch $e_s$, to the end of $L_s$ and changes no copy.  Append and acknowledge appends to $\widehat L_c[s]$ only the block at position $n=|\widehat L_c[s]|+1$ in the registered epoch $\widehat e_c[s]=e_s$, and that block is the one Approve placed at position $n$ of $L_s$ and disseminated; so the appended block matches $L_s$ at that position, extending the prefix by one.  Within an epoch no transaction shortens or rewrites either sequence.  Hence each remains a prefix of $L_s$, and both are non-decreasing.
\end{proof}

\ThmLogValidity*
\begin{proof}
By induction on the length of $L_s$.  The empty log is valid.  A block is appended to $L_s$ only by Approve, whose precondition requires the recorded transaction's coins to be unspent in $L_s$---held, in the holdings derived from the current $L_s$, by the agent the transaction debits.  By the induction hypothesis $L_s$ is valid, so its derived holdings are those of the specification state its records produce; the recorded transaction is therefore enabled in that state, and the extended log is valid.  Mint, Redeem, and Swap append only via the same Approve discipline, each with its own holding check, so the same argument applies.
\end{proof}

\ThmCoinsConservation*
\begin{proof}
An $s$-coin enters a holding for the first time only through Mint, which adds the minted coins to $B_s$ upon Approve and records $\mathsf{mint}(k,t)$ in $L_s$; no other transaction creates a coin, and none destroys one.  Thereafter a coin leaves a holding only at the requesting step of Pay, Redeem, or Swap, which debits it from the requester, and enters one only at the finality step of the same transaction, Accept payment or its counterpart, which credits the same coin to the counterparty.  Hence at every configuration the $s$-coins across all holdings are those recorded as minted in $L_s$, less the coins of requests that have been debited but not yet credited.

Let $c$ be a quiescent configuration of a run with no state-loss fault.  We show no such request exists at $c$.  By the argument of Correspondence at Quiescence (Theorem~\ref{theorem:secure-coins-cva-correspondence}), which relies only on Prefix Consistency and Log Validity, every $\mathsf{request}$ has been delivered and processed by Approve, every approved block has been appended by every custodian and is final, and every $\mathsf{final}$ has been accepted.  It remains to rule out a request that Approve dropped.  In a run with no state loss, a holder's $B_q$ is credited only by the finality of a block of $L_s$ and debited only by $q$'s own requests, so $B_q$ agrees with the holdings derived from $L_s$ except for $q$'s pending requests, whose coins are absent from $B_q$; the coins of a request are therefore unspent in $L_s$ when Approve processes it, and it is approved.  Every debit at $c$ thus has a matching credit, and the holdings of $s$-coins are exactly those minted in $L_s$.
\end{proof}

\ThmFinalitySoundness*
\begin{proof}
$s$ emits $\mathsf{final}(e,n,\mathsf{pay}(q,r,x))$ only from Bind finality, with $e=e_s$, and only when $\mathsf{block}(e,n,\mathsf{pay}(q,r,x))$ is already in $L_s$; that block was placed by Approve, whose precondition requires $x$ unspent in $L_s$ at the time of approval.
\end{proof}

\ThmCorrespondenceQ*
\begin{proof}
At quiescence no reactive transaction is enabled, so every emitted message has been delivered and processed; otherwise the corresponding communicate, or the reactive transaction consuming it, would be enabled.  Hence: no $\mathsf{request}$ is in flight or pending, so no payment is mid-approval; every $\mathsf{block}(e,n,\tau)$ approved has reached every custodian and been appended (every custodian having registered the current epoch, Append and acknowledge is enabled while an undelivered or unappended block remains), so $\widehat L_c[s]=L_s$ for every $c\in S_s$, whence by Prefix Consistency (Lemma~\ref{lemma:secure-coins-cva-prefix}) every block of $L_s$ is held by all of $S_s$ and is therefore final; and every $\mathsf{final}(e,n,\cdot)$ has been accepted (Accept payment is enabled while an unprocessed $\mathsf{final}$ remains), so every finalised credit has been applied to the payee's $B_r$ and every committed debit was applied at request.  Thus $B_p$ equals the holdings derived from the logs.  By Log Validity (Lemma~\ref{lemma:secure-coins-cva-log-validity}) each log is a valid sequence of transactions, hence a run of the specification that produces exactly these holdings; so $\alpha$ is the state reached by that run and is a reachable specification state.
\end{proof}

\ThmExactRecovery*
\begin{proof}
By induction on the recoveries of $p$'s currency.  Let $\mathsf{block}(e_0,n,\cdot)$ be final before the recovery, witnessed by a supermajority $G'\subseteq S_p$ registered in $e_0$ whose copies held the prefix through $n$ at the finalising step (Definition~\ref{definition:secure-coins-cva-finality}).  A copy changes thereafter only by appending, or by aligning at a later registration to the log adopted by that recovery, which contains the block by the induction hypothesis; so every $c\in G'$ still holds the block at recovery time.  As $|G|\ge\lceil\sigma^S|S_p|\rceil$ and $|G'|\ge\lceil\sigma^S|S_p|\rceil$ with $\sigma^S>1/2$, the two supermajorities intersect; pick $c\in G\cap G'$, whose collected copy contains the block.  Reserved epochs are distinct and increase along the run: a recovery reserves $e'$ only once a supermajority have promised it, promises never decrease, and every later recovery queries a supermajority that intersects the promising one, hears a promise of at least $e'$, and reserves above it.  Hence every copy registered in an epoch was installed by the one recovery that reserved it and extended only by blocks that recovery's sovereign approved in that epoch.  Let $e$ be the maximal epoch among the collected copies.  Every copy registered in $e$ is thus a prefix of one log, the log adopted at the reservation of $e$ extended by the blocks of $e$; this log contains the block, either because $c$ is registered in $e$, or because $e>\widehat e_c[p]\ge e_0$, so $e$ was reserved by a recovery later than the one that reserved $e_0$, hence after the fault that ended the sovereign's stint in $e_0$ and so after the block became final, and that recovery's adopted log contains the block by the induction hypothesis.  The longest copy in $e$ therefore contains the block.
\end{proof}

\ThmCoinsResilient*
\begin{proof}
The coins transactions realise the spec transactions at quiescence under $\alpha$, where each quiescent log is valid (Lemma~\ref{lemma:secure-coins-cva-log-validity}) and the derived holdings are the specification's.  A state-loss fault at a sovereign $p$ empties $L_p$; under the custodian-availability assumption (Section~\ref{section:secure-coins-cva-recovery}) a supermajority retains its copies, so by Exact Recovery (Theorem~\ref{theorem:secure-coins-cva-exact-recovery}) the recovered log contains every final block, and the holdings derived from it agree with the pre-fault holdings on every committed coin.  A custodian or holder loss is repaired likewise from the sovereign's log, which is authoritative for holdings (Definition~\ref{definition:valid-log}).  The fault and its recovery thus leave $\alpha$ unchanged at quiescence---a stutter---and a non-final payment that recovery may drop was relied upon by no payee (Section~\ref{section:secure-coins-cva-finality}).  Hence the quiescent image of a live run over the coins transactions and $F$ uses only spec transactions and is a correct run of the specification; that is, $\alpha$ is $F$-resilient.
\end{proof}

\ThmNoDoubleSpend*
\begin{proof}
By Exact Recovery (Theorem~\ref{theorem:secure-coins-cva-exact-recovery}) the recovered log contains every final block, so the holdings derived from the recovered $L_p$ (Definition~\ref{definition:valid-log}) record every coin spent in a final payment as spent.  Approve appends a payment of $x$ only when $x$ is unspent in $L_p$ (Definition~\ref{definition:secure-coins-cva-transactions}); a coin spent in a final payment is thus never re-approved.  A payment that was not final may be absent from the recovered log, but no payee relied on it (Section~\ref{section:secure-coins-cva-finality}), so no relied-upon payment is reversed.
\end{proof}

\ThmPaymentLiveness*
\begin{proof}
By eventual delivery in a correct run (Section~\ref{section:dts}) the $\mathsf{request}$ reaches $s$, which executes Approve, appends $\mathsf{block}(e_s,n,\mathsf{pay}(q,r,x))$, and disseminates it to every $c\in S_s$.  Fix a custodian $c$ in the live supermajority, registered in $e_s$.  By eventual delivery $c$ receives every block at positions $1,\dots,n$; Append and acknowledge fires in order as each predecessor arrives, so the frontier $|\widehat L_c[s]|$ reaches $n$ and $c$ sends $\mathsf{ack}(e_s,m)$ for $m$ up to $n$.  Each such ack reaches $s$, which sets $\ell_s(c)$ accordingly.  Once every $c$ in the live supermajority has done so, $|\{c:\ell_s(c)\ge(e_s,n)\}|$ meets the threshold, block $(e_s,n)$ is final, and Bind finality emits $\mathsf{final}(e_s,n,\mathsf{pay}(q,r,x))$ to $q$ and $r$, delivered by eventual delivery.
\end{proof}

\ThmCoinsGrassroots*
\begin{proof}
The platform transactions of Definition~\ref{definition:secure-coins-cva-transactions} are unary, each at a single agent, with guard $\{q\}$ or $\{p\}$ (Request pay, Form currency, and the requesting steps of Mint, Redeem, Swap) or empty (the reactive Approve, Append and acknowledge, Bind finality, Accept payment, Become custodian), and each modifies only its own agent's platform state and emits messages to known peers.  They are therefore CVA platform transactions in the sense of Definition~\ref{def:cva-transactions}, and by CVA Grassroots (Theorem~\ref{thm:cva-grassroots}) the platform is grassroots.
\end{proof}

%% file: sections/appendix-proofs.tex
\section{Proofs}\label{appendix:proofs}

\input{sections/proofs-social-graph}

\input{sections/proofs-cva}

\input{sections/proofs-gsg-cva-properties}

\input{sections/proofs-secure-gsg-cva-properties}

%% file: sections/proofs-social-graph.tex
\ThmGsgMutuality*
\begin{proof}
By induction on run length.  Initially $F_p=\emptyset$ for all $p$, so the biconditional holds vacuously.  Befriend$(p,q)$ adds $q$ to $F_p$ and $p$ to $F_q$ symmetrically; Unfriend$(p,q)$ removes them symmetrically; unrelated pairs are unchanged.
\end{proof}

\ThmGsgGrassroots*
\begin{proof}
\emph{Obliviousness.}  Both transactions have nonempty guards: Befriend by $\{p,q\}$ and Unfriend by $\{p\}$ or $\{q\}$.  By Guarded Obliviousness (Section~\ref{section:dts}), the social graph is oblivious.
\emph{Interactivity.}  For disjoint nonempty $P,P'\subset\Pi$, a correct run of $\calF(P\cup P')$ in which some $p\in P$ and $q\in P'$ will Befriend and the transaction is taken is no interleaving of correct runs of $\calF(P)$ and $\calF(P')$, since the Befriend step changes the local states of both $p$ and $q$.
\end{proof}

\ThmSsgFofExact*
\begin{proof}
By induction on run length.  Initially every friend set is empty, so the claim holds vacuously.  We check each transaction that writes a friend set or a friend-of-friend map.

\emph{Befriend$(p,q)$.}  It sets $N'_p(q)=F'_q$ and $N'_q(p)=F'_p$ for the new edge, and for every prior friend $r\in F_p$ sets $N'_r(p)=F'_p$ and likewise for $F_q$, so every record of the two changed sets $F'_p,F'_q$ is updated to match.  Records of unchanged friend sets are untouched and remain exact by hypothesis.

\emph{Unfriend$(p,q)$.}  Symmetric: $p$ and $q$ are removed from each other's maps, and every remaining friend's record of the two changed sets is updated to $F'_p$, $F'_q$.

\emph{Replace$(p,p')$.}  $p'$ inherits $N_p$, exact by hypothesis; each friend $r$ replaces its record of $p$ by a record of $p'$ at the same value $N_r(p)=F_p=F_{p'}$; and every $s$ recording a changed $F_r$ has $N_s(r)$ refreshed to $F'_r$.  All changed sets are thus matched.

No state loss occurs by hypothesis, and Recover is enabled only after a state loss; so these are the only transactions to consider, and exactness is preserved.
\end{proof}

\ThmSecureGsgGrassroots*
\begin{proof}
\emph{Obliviousness.}  Befriend, Unfriend, and Replace have nonempty guards.  StateLoss and Recover are unguarded but never cross-group-enabled in an interleaving of disjoint runs: in a run of $\calF(P)$ every friend set lies within $P$, so StateLoss's participant $\{p\}\subseteq P$ and Recover's participants $\{p,q\}$ with $q$ a friend (custodian) of $p$ lie in $P$.  Hence every cross-group transaction has a nonempty guard, and by Guarded Obliviousness (Section~\ref{section:dts}) the secure social graph is oblivious.
\emph{Interactivity.}  Inherited from the social graph (Theorem~\ref{theorem:gsg-grassroots}): Befriend changes both participants' states.
\end{proof}

\ThmSsgImplements*
\begin{proof}
\emph{$\sigma$ is an implementation.}  At the initial secure configuration every $F_p=\emptyset$, so $\sigma(c)_p=\emptyset$ for every $p$; this is the initial social-graph configuration, as required.

\emph{Each secure transaction maps to a social-graph transition or a stutter.}  Befriend$(p,q)$ adds the mutual pair to $F_p,F_q$; under $\sigma$ this adds $q$ to $\sigma(c)_p$ and $p$ to $\sigma(c)_q$ and changes no other projected set, the social-graph Befriend$(p,q)$.  Unfriend$(p,q)$ removes the pair from both $F_p$ and $F_q$, so $q$ leaves $\sigma(c)_p$ and $p$ leaves $\sigma(c)_q$, the social-graph Unfriend$(p,q)$.  Recover$(p,q)$, on a run with no state-loss fault, sets $F'_p=N_q(p)=F_p$ by Friend-of-Friend Exactness (Lemma~\ref{lem:ssg-fof-exact}), so $\sigma(c)_p$ is unchanged---a stutter.  Replace$(p,p')$, after relabelling $p'$ to $p$, leaves every projected friend set fixed: $F'_{p'}=F_p$ and each friend $r$ has $p'$ in place of $p$, which under the relabelling is $p$ again---a stutter.

\emph{Correct.}  Let $r'$ be a correct run of the secure social graph using only its transactions (Definition~\ref{definition:secure-gsg-transactions}).  Its non-stutter image $\sigma(r')$ consists of social-graph Befriend and Unfriend transitions in order, as Recover and Replace project to stutters.  Consecutive distinct configurations of $\sigma(r')$ are therefore social-graph transitions, so $\sigma(r')$ is a social-graph run.  For liveness: a social-graph class enabled throughout a suffix of $\sigma(r')$ is a Befriend or Unfriend class whose precondition is a condition on friend sets alone; the corresponding secure class is enabled throughout the corresponding suffix of $r'$, since the projected precondition holds exactly when the secure one does and the guards coincide.  As $r'$ is live, that class is taken in $r'$, hence in $\sigma(r')$.  So $\sigma(r')$ is live, hence correct.

\emph{Complete.}  Let $\rho=c_0\to c_1\to\cdots$ be a correct run of the social graph.  We construct a secure run $\rho'=c_0'\to c_1'\to\cdots$ with $\sigma(c_i')=c_i$ and $F_p(c_i')=F_p(c_i)$ for every $i$ and $p$, by induction on $i$.  For $i=0$, $c_0'$ is the initial secure configuration, and $\sigma(c_0')=c_0$ as shown above.  For the step $c_i\to c_{i+1}$: if it is a change-volition transaction of some $p$, let $c_{i+1}'$ apply the same change to $p$'s volitional state, which $\sigma$ preserves.  If it is Befriend$(p,q)$, then $q\notin F_p$ and $p,q$ will it at $c_i$, hence at $c_i'$, where friend sets and volitional states coincide; so the secure Befriend$(p,q)$ is enabled at $c_i'$, and we take it.  Its friend-set effect is that of the social-graph Befriend, so $F_r(c_{i+1}')=F_r(c_{i+1})$ for every $r$, and its friend-of-friend updates are determined.  Unfriend is the same.  As Befriend and Unfriend update both endpoints, friend sets along $\rho'$ are mutual, so $\sigma(c_{i+1}')_r=F_r(c_{i+1}')=F_r(c_{i+1})$ for every $r$, that is, $\sigma(c_{i+1}')=c_{i+1}$.  No step of $\rho'$ is a stutter under $\sigma$, each projecting to the corresponding step of $\rho$; hence $\sigma(\rho')=\rho$.

It remains to show that $\rho'$ is correct.  It is safe by construction.  For liveness, a Befriend or Unfriend class enabled throughout a suffix of $\rho'$ is enabled throughout the corresponding suffix of $\rho$, as preconditions, volitional states, and guards coincide; $\rho$ being live, it is taken in $\rho$, and by construction in $\rho'$.  Replace is never enabled, since the volitional states of $\rho'$ are those of $\rho$, which has no Replace class to will.  Recover$(p,q)$ in a configuration with no state loss has $N_q(p)=F_p$ and $N_p(r)=F_r$ for every $r\in N_q(p)$ by Friend-of-Friend Exactness (Lemma~\ref{lem:ssg-fof-exact}), so its effect is the identity; a volitional machine transaction requires $c\ne c'$ (Definition~\ref{definition:vmat}), so it induces no transition of $\rho'$ and imposes no liveness obligation.  Hence $\rho'$ is live, so correct, and $\rho=\sigma(\rho')$ as Definition~\ref{def:implementation-properties} requires.
\end{proof}

\ThmSsgResilient*
\begin{proof}
Let $r'\subseteq T'\cup F$ be a live friendship-preserving run, $F$ the state-loss faults (Definition~\ref{def:state-loss-fault}).  The secure social graph transactions map to social-graph transitions or stutters as in the proof of Theorem~\ref{thm:ssg-implements-sg}.  A state-loss fault at $p$ sets $F_p:=\emptyset$ but changes no other agent's friend set, so for each friend $r$ of $p$ the membership $p\in F_r$ is untouched and $p\in\sigma(c)_r$ still holds; hence $\sigma(c)_p=\{r:p\in F_r\}$ is unchanged and every $\sigma(c)_r$ is unchanged---a stutter.  Recover$(p,q)$ on a friendship-preserving run sets $F'_p=N_q(p)=F_p^-$, the friend set held just before the loss; since a state-lost agent takes no step until recovery and Befriend and Unfriend update both endpoints symmetrically, $F_p^-=\{r:p\in F_r\}$ and $\sigma(c)_p$ is unchanged---a stutter.  Thus the non-stutter image $\sigma(r')$ consists of social-graph Befriend and Unfriend transitions in order, so it is a social-graph run, and the liveness argument of Theorem~\ref{thm:ssg-implements-sg} applies verbatim, the faults and the Recover transitions being stutters.  Hence $\sigma(r')$ is a correct run of the social graph.
\end{proof}

%% file: sections/proofs-cva.tex
\ThmKnownContainment*
\begin{proof}
By induction on run length.  Initially $\mathit{known}_p = \emptyset \subseteq P$.  Only discover modifies $\mathit{known}$, and a discover transaction with participants $\{p,q\}$ lies in $R(P)$ only if $\{p,q\}\subseteq P$, so it adds only members of $P$.
\end{proof}

\ThmOutboxContainment*
\begin{proof}
Only platform transactions write to $o_p$, and by outbox well-formedness any message added to $o_p$ has recipient in $\mathit{known}_p$.  By Lemma~\ref{lem:known-containment}, $\mathit{known}_p\subseteq P$.
\end{proof}

\ThmCvaGrassroots*
\begin{proof}
\emph{Obliviousness.}  Discover is guarded by $\{p\}$, so by Guarded Obliviousness~\cite{lewis2026volitional} no cross-group discover is enabled in any interleaving of independent runs.  Platform transactions and Advance-date are unary, so they pose no cross-group case.  The remaining case is a cross-group communicate with $p\in P$, $q\in P'$: its precondition requires $\mathsf{message}(p,q,\cdot)\in o_p$, but by Lemma~\ref{lem:outbox-containment} applied to a $P$-run every message in $o_p$ has recipient in $P$, so the precondition fails.

\emph{Interactivity.}  For any disjoint nonempty $P, P'\subset \Pi$, pick $p\in P$ and $q\in P'$.  The transaction $\mathit{discover}(p,q)$ is in $R(P\cup P')$ and, when $p$ wills it, is enabled at $c0(P\cup P')$.  A correct run of $\calF(P\cup P')$ whose first step carries out $\mathit{discover}(p,q)$ exists, and is no interleaving of correct runs of $\calF(P)$ and $\calF(P')$ since its first step's participants span both groups.

Obliviousness and interactivity together yield grassroots~\cite{lewis2026volitional}.
\end{proof}

%% file: sections/proofs-gsg-cva-properties.tex
\ThmKnowledgeMono*
\begin{proof}
By induction on the length of the run.  The claim concerns the domains $\mathrm{dom}(\mathrm{FMap}_p)$ and $\mathrm{dom}(\mathrm{FoFMap}_p[q])$, so it suffices to show that no transaction removes a key from either.  Every transaction that writes platform state --- Accept offer, Resolve simultaneous offer, Integrate accept, End friendship, Integrate unfriend (Sections~\ref{section:gsg-cva-befriend},~\ref{section:gsg-cva-unfriend}), and Integrate stream update (Section~\ref{section:gsg-cva-stream}) --- either leaves $\mathrm{FMap}_p$ and $\mathrm{FoFMap}_p$ unchanged, or assigns $\mathrm{FMap}_p[q]:=x$ for some $q$, or assigns $\mathrm{FoFMap}_p[q][f]:=x$ for some $q,f$.  An assignment adds the key when absent and retains it when present; none deletes a key.  Hence both domains are non-decreasing along the run.
\end{proof}

\ThmFriendshipMono*
\begin{proof}
By induction on the run, inspecting every transaction that writes an epoch.  For the direct epoch $\mathrm{epoch}_p(q)$: Accept offer and Resolve simultaneous offer set $\mathrm{FMap}_q[p]:=x$ under the precondition $\mathrm{epoch}_q(p)<x$; Integrate accept sets $\mathrm{FMap}_p[q]:=x$ under $\mathrm{epoch}_p(q)<x$; End friendship sets $\mathrm{FMap}_p[q]:=\mathrm{epoch}_p(q)+1$; and Integrate unfriend sets $\mathrm{FMap}_q[p]:=x$ under $\mathrm{epoch}_q(p)<x$.  For the observed epoch $\mathrm{epoch}_p(q,f)$: Integrate stream update sets it under $\mathrm{epoch}_r(q,f)<x$.  In every case the value written strictly exceeds the value stored, and no other transaction writes an epoch; hence both epochs are non-decreasing.
\end{proof}

\ThmMessageBounds*
\begin{proof}
Each bound holds at the instant the message is emitted and persists thereafter by Friendship Monotonicity (Lemma~\ref{lemma:gsg-friendship-mono}).
(1) Offer friendship at $B$ sends $\mathsf{friend\_request}(x)$ with $x=\mathrm{epoch}_B(A)+1$ and does not modify $\mathrm{FMap}_B$, so $\mathrm{epoch}_B(A)=x-1$ when sent.
(2) Accept offer and Resolve simultaneous offer at $B$ set $\mathrm{FMap}_B[A]:=x$ before sending $\mathsf{accept}(x)$, so $\mathrm{epoch}_B(A)=x$ when sent.
(3) End friendship at $B$ sets $\mathrm{epoch}_B(A):=x$ before sending $\mathsf{unfriend}(x)$, so $\mathrm{epoch}_B(A)=x$ when sent.
(4) A $\mathsf{stream\_update}(f,x)$ is emitted either as a \emph{per-edge update} or as a \emph{snapshot entry}. A per-edge update is emitted by a transaction (Accept offer, Resolve simultaneous offer, Integrate accept, End friendship, or Integrate unfriend) that has, in the same step, just set $\mathrm{epoch}_B(f):=x$. A snapshot entry is emitted by an activating transaction with $x:=\mathrm{epoch}_B(f)$ for an existing friend $f$. In both cases $\mathrm{epoch}_B(f)=x$ when sent.
\end{proof}

\ThmObserverBound*
\begin{proof}
By induction on the run.  Initially $\mathrm{epoch}_A(p,q)=0$ and the bound holds.  The only transaction that writes $\mathrm{FoFMap}_A$ is Integrate stream update, which sets $\mathrm{epoch}_A(p,q):=x$ either on processing $\mathsf{stream\_update}(q,x)$ from $p$, in which case $x\le\mathrm{epoch}_p(q)$ by Message Bounds~(4), or, through its symmetric assignment, on processing $\mathsf{stream\_update}(p,x)$ from $q$, in which case $x\le\mathrm{epoch}_q(p)$ by Message Bounds~(4).  Either way the new value is at most $\max(\mathrm{epoch}_p(q),\mathrm{epoch}_q(p))$.  All other transactions leave $\mathrm{epoch}_A(p,q)$ unchanged, and by Friendship Monotonicity the right-hand side never decreases, so the bound is preserved.
\end{proof}

\ThmFofMapSymmetry*
\begin{proof}
By induction on the run.  In the initial state $\mathrm{FoFMap}_p=\emptyset$, so $\mathrm{epoch}_p(q,f)=\mathrm{epoch}_p(f,q)=0$.  Only Integrate stream update (Section~\ref{section:gsg-cva-stream}) writes $\mathrm{FoFMap}$, setting $\mathrm{FoFMap}_r[q][f]:=x$ and $\mathrm{FoFMap}_r[f][q]:=x$ in one step.  It assigns the two mirror entries the same value simultaneously; entries it does not touch are unchanged.  By the induction hypothesis the mirror entries were equal before the step, so they remain equal after it.  Hence $\mathrm{epoch}_p(q,f)=\mathrm{epoch}_p(f,q)$ at every reachable state.
\end{proof}

\ThmStaleInert*
\begin{proof}
Here $e$ is $\mathrm{epoch}_q(p)$ when $r=q$ and $\mathrm{epoch}_q(p,r)$ otherwise.  The transactions at $q$ that consume $m$ are, for $r=q$, Accept offer and Resolve simultaneous offer on a $\mathsf{friend\_request}$, Integrate accept on an $\mathsf{accept}$, and Integrate unfriend on an $\mathsf{unfriend}$ (Sections~\ref{section:gsg-cva-befriend},~\ref{section:gsg-cva-unfriend}); and, for $r\ne q$, Integrate stream update on a $\mathsf{stream\_update}$ (Section~\ref{section:gsg-cva-stream}).  Each requires $e<m.e$.  With $m.e\le e$ none is enabled, and by Friendship Monotonicity (Lemma~\ref{lemma:gsg-friendship-mono}) $e$ never decreases, so none becomes enabled later and $m$ is never consumed.  The only preconditions at $q$ that mention a message other than the one consumed concern $q$'s own outbox, so a message sitting in $i_q$ enables and disables nothing: $q$ takes the same transactions, with the same effects, as in the run without $m$.
\end{proof}

\ThmPostAccept*
\begin{proof}
$p$ sends $\mathsf{accept}(x)$ only from Accept offer or Resolve simultaneous offer, each fired on receiving $\mathsf{friend\_request}(x)$ from $q$, which $q$ sent from Offer friendship toward $p$. By Message Bounds~(1) (Lemma~\ref{lemma:gsg-message-bounds}), $\mathrm{epoch}_q(p)\ge x-1$ from that send onward. The only transaction that raises $\mathrm{epoch}_q(p)$ to $x$ is Integrate accept on the reply from $p$; any value above $x$ requires End friendship or a fresh Offer, both excluded on $(t_a,t]$ by hypothesis. By Friendship Monotonicity (Lemma~\ref{lemma:gsg-friendship-mono}), $\mathrm{epoch}_q(p)\in\{x-1,x\}$ at $t$.
\end{proof}

\ThmPostUnfriend*
\begin{proof}
Without loss of generality $p$ willed End friendship$(q)$ at $t_1$, setting $\mathrm{epoch}_p(q):=x+1$ and sending $\mathsf{unfriend}(x+1)$. By Friendship Monotonicity $\mathrm{epoch}_p(q)\ge x+1$ thereafter, and advancing past $x+1$ needs a fresh Offer (excluded), so $\mathrm{epoch}_p(q)=x+1$. For $q$: starting from $x$ at $t_0$, the only writes available without a fresh Offer are $q$ willing End friendship toward $p$ or $q$ executing Integrate unfriend on $\mathsf{unfriend}(x+1)$, each setting $\mathrm{epoch}_q(p):=x+1$. Hence $\mathrm{epoch}_q(p)\in\{x,x+1\}$.
\end{proof}

\ThmFofProvenance*
\begin{proof}
By induction on the run. Initially every observed epoch is $0$, so the claim is vacuous. The only transaction writing $\mathrm{FoFMap}_p$ is Integrate stream update, which sets $\mathrm{epoch}_p(q,f):=x$ (with its mirror) on processing $\mathsf{stream\_update}(f,x)$ from $q$, or on processing $\mathsf{stream\_update}(q,x)$ from $f$ through the symmetric assignment. In the first case the message was emitted by $q$ at a time $t_2\le t_1$ when $\mathrm{epoch}_q(f)=x$, in the second by $f$ when $\mathrm{epoch}_f(q)=x$, in both cases by the emission rule underlying Message Bounds~(4) (Lemma~\ref{lemma:gsg-message-bounds}).
\end{proof}

\ThmFriendListSoundness*
\begin{proof}
That $q\in\widetilde F_p$ means $\mathrm{epoch}_p(q)$ is odd at $t_3$; let $s\le t_3$ be the last time $p$ wrote this odd value. The writing transaction is Accept offer$(q)$, Resolve simultaneous offer$(q)$, or Integrate accept on an $\mathsf{accept}$ from $q$.

\emph{Intent of $p$.} In the first two cases $p$ willed the transaction, an intent act toward $q$; set $t_1:=s$. In the third, Integrate accept completes an Offer friendship$(q)$ that $p$ willed at some $t_1<s$; take that $t_1$. In all cases $p$ willed an offer or acceptance toward $q$ at $t_1\le t_3$. Since $s$ is the last write of $\mathrm{epoch}_p(q)$ and it is odd, $p$ willed no End friendship$(q)$ in $(t_1,t_3]$ (such a transaction would advance the epoch to an even value). Hence $p$ wants to befriend $q$ at $t_1$.

\emph{Intent of $q$.} The transaction at $s$ was triggered by a message from $q$: a $\mathsf{friend\_request}$ (triggering Accept offer or Resolve simultaneous offer), sent by the willed Offer friendship$(p)$ of $q$; or an $\mathsf{accept}$ (triggering Integrate accept), sent by the willed Accept offer or Resolve simultaneous offer of $q$ toward $p$. Either is a willed offer or acceptance by $q$ toward $p$ at some $t_2\le s\le t_3$; at that instant the interval $(t_2,t_2]$ is empty, so $q$ wants to befriend $p$ at $t_2$.
\end{proof}

\ThmFofSoundness*
\begin{proof}
Let $x=\mathrm{epoch}_p(q,r)$ at $t_3$, odd. By Observed-Epoch Provenance (Lemma~\ref{lemma:gsg-cva-fof-provenance}) there is $t'\le t_3$ with $\mathrm{epoch}_q(r)=x$ or $\mathrm{epoch}_r(q)=x$; without loss of generality $\mathrm{epoch}_q(r)=x$, odd, so $r\in\widetilde F_q$ at $t'$. Friend List Soundness (Theorem~\ref{theorem:gsg-cva-friend-list-soundness}) applied to $q$ at $t'$ yields $t_1,t_2\le t'\le t_3$ with $q$ wanting to befriend $r$ at $t_1$ and $r$ wanting to befriend $q$ at $t_2$.
\end{proof}

\ThmChannelValidity*
\begin{proof}
Addressability in CVA is mutual known-set membership, which permits each side to send to the other (Section~\ref{sec:cva}). Known sets never shrink: only Discover writes them, and only by union (Definition~\ref{def:cva-transactions}). It therefore suffices to exhibit, for each direction, a handshake step that placed the peer in the set.

The value $\mathrm{epoch}_p(q)$ became odd via Accept offer, Resolve simultaneous offer, or Integrate accept. In the first two, $p$ sends $\mathsf{accept}(x)$ to $q$ in the same step; in the third, Integrate accept completes an Offer friendship$(q)$ in which $p$ had sent $\mathsf{friend\_request}$ to $q$. A platform transaction emits a message to $q$ only when $q\in\mathit{known}_p$, hence $q\in\mathit{known}_p$. Symmetrically, $q$ sent $p$ a $\mathsf{friend\_request}$ (if $q$ offered) or an $\mathsf{accept}$ (if $q$ accepted) during the handshake, hence $p\in\mathit{known}_q$. Both memberships persist by monotonicity.
\end{proof}

\ThmMutualLiveness*
\begin{proof}
By hypothesis there are witnessing offer-or-accept acts by $p$ toward $q$ and by $q$ toward $p$ at times $\le t$; and, since wanting holds at every $t'\ge t$, neither agent wills End friendship on the other at any time $\ge t$ --- an End friendship would set its own epoch even and, being its latest act toward the other, falsify wanting at the next state until a later offer, contradicting the hypothesis. Hence after $t$ no $\mathsf{unfriend}$ passes between $p$ and $q$, and no epoch between them is advanced by End friendship or Integrate unfriend.

\emph{Finitely many direct messages after $t$.} The direct messages between $p$ and $q$ are $\mathsf{friend\_request}$, $\mathsf{accept}$, and $\mathsf{unfriend}$; the last is excluded after $t$. Offer friendship emits a $\mathsf{friend\_request}$ only while the offerer is inactive, and Accept offer or Resolve simultaneous offer emits an $\mathsf{accept}$ only in response to a $\mathsf{friend\_request}$. In a correct run every such message is eventually delivered and integrated; a delivered $\mathsf{friend\_request}$ from a wanting recipient yields an $\mathsf{accept}$ whose integration makes both sides active, after which Offer friendship emits nothing further. Thus only finitely many direct messages pass between $p$ and $q$ after $t$, and there is a time $t_0>t$ by which all have been integrated and no further direct message between $p$ and $q$ is ever sent.

\emph{Both active at $t_0$.} A handshake completes, so $\mathrm{epoch}_p(q)$ and $\mathrm{epoch}_q(p)$ are odd at $t_0$; with no even-writing transition available after $t$, Friendship Monotonicity (Lemma~\ref{lemma:gsg-friendship-mono}) keeps them odd thereafter.

\emph{Agreement via Message Bounds.} Let $x=\max(\mathrm{epoch}_p(q),\mathrm{epoch}_q(p))$ at $t_0$; without loss of generality $\mathrm{epoch}_p(q)=x$. The last transaction setting $\mathrm{epoch}_p(q):=x$ is Integrate accept on an $\mathsf{accept}(x)$ from $q$, giving $\mathrm{epoch}_q(p)\ge x$ by Message Bounds~(2) (Lemma~\ref{lemma:gsg-message-bounds}); or Accept offer / Resolve simultaneous offer, which sends $\mathsf{accept}(x)$ to $q$ in the same step, a message integrated by $t_0$, so again $\mathrm{epoch}_q(p)\ge x$. As $x$ is the maximum, $\mathrm{epoch}_q(p)=x=\mathrm{epoch}_p(q)$.

Set $t^*:=t_0$. After $t^*$ no direct message between $p$ and $q$ is sent or pending and no volition between them remains, so neither epoch changes; hence $\mathrm{epoch}_p(q)=\mathrm{epoch}_q(p)=x$, odd, from $t^*$ onward.
\end{proof}

\ThmFriendshipEstablishment*
\begin{proof}
By Mutual Friendship Liveness (Lemma~\ref{lemma:gsg-cva-mutual-liveness}) there is $t^*>t$ from which onward $\mathrm{epoch}_p(q)=\mathrm{epoch}_q(p)$ is a common odd value. By the definition of $\widetilde F$ as the odd-epoch entries, $q\in\widetilde F_p$ and $p\in\widetilde F_q$ from $t^*$ onward.
\end{proof}

\ThmFofVisibility*
\begin{proof}
By Mutual Friendship Liveness (Lemma~\ref{lemma:gsg-cva-mutual-liveness}) applied to $(q,r)$, there is a time from which $\mathrm{epoch}_q(r)=\mathrm{epoch}_r(q)=e$, a common odd value, stably. By Friendship Establishment (Theorem~\ref{theorem:gsg-cva-establishment}) applied to $(p,q)$, there is a time from which $q\in\widetilde F_p$ and $p\in\widetilde F_q$ stably; in particular $p\in\mathrm{dom}(\mathrm{FMap}_q)=\mathrm{Rec}_q$.

We show $p$ eventually processes a $\mathsf{stream\_update}(r,e)$ from $q$. Consider the transaction at which $q$ set $\mathrm{epoch}_q(r):=e$. If $p\in\mathrm{Rec}_q$ at that step, $q$ emits the per-edge $\mathsf{stream\_update}(r,e)$ to $p$. Otherwise $p$ joins $\mathrm{Rec}_q$ only afterwards, and the activating transaction of the $p$--$q$ edge sends $p$ the on-join snapshot, which includes $\mathsf{stream\_update}(r,\mathrm{epoch}_q(r))$ for the existing friend $r$; since $\mathrm{epoch}_q(r)=e$ is stable by then, the value conveyed is $e$. Either way, by eventual delivery $p$ processes a $\mathsf{stream\_update}(r,e)$ from $q$ with $q\in\mathrm{dom}(\mathrm{FMap}_p)$, and Integrate stream update sets $\mathrm{epoch}_p(q,r):=e$.

Let $t^*>\max(t_1,t_2)$ be a time past all the above. From $t^*$ onward, Friendship Monotonicity (Lemma~\ref{lemma:gsg-friendship-mono}) gives $\mathrm{epoch}_p(q,r)\ge e$, and Observer Bound (Lemma~\ref{lemma:gsg-observer-bound}) gives $\mathrm{epoch}_p(q,r)\le\max(\mathrm{epoch}_q(r),\mathrm{epoch}_r(q))=e$. Hence $\mathrm{epoch}_p(q,r)=e$, odd, from $t^*$ onward.
\end{proof}

\ThmUnfriendPropagation*
\begin{proof}
Without loss of generality $p$ does not want to befriend $q$ at any time $\ge t$.

\emph{(1).} Suppose $\mathrm{epoch}_p(q)$ were odd at some $t''\ge t$, that is $q\in\widetilde F_p$. By Friend List Soundness (Theorem~\ref{theorem:gsg-cva-friend-list-soundness}), $p$ wants to befriend $q$ at some $t_1\le t''$, with no End friendship$(q)$ willed by $p$ in $(t_1,t'']$. If $t_1\ge t$ this contradicts the hypothesis directly; if $t_1<t$ then, since $t\in(t_1,t'']$, $p$ wants to befriend $q$ at $t$ (witness $t_1$), again a contradiction. Hence $\mathrm{epoch}_p(q)$ is even for all $t''\ge t$. Applied to $\mathrm{epoch}_q(p)$, the same argument (an odd value yields, by Friend List Soundness at $q$, a time $\le t''$ at which $p$ wants to befriend $q$, with the same contradiction) shows $\mathrm{epoch}_q(p)$ is even for all $t''\ge t$. Thus $q\notin\widetilde F_p$ and $p\notin\widetilde F_q$ from $t$ onward.

\emph{(2).} If $\mathrm{epoch}_r(p,q)$ is never odd from $t$ onward, the claim holds. Otherwise, by Observed-Epoch Provenance (Lemma~\ref{lemma:gsg-cva-fof-provenance}) some direct side of the $p$--$q$ edge was once odd, so $p$ and $q$ were mutual friends at a common odd epoch which, by part~(1) and the hypothesis, was then ended and never re-established; by Post-Unfriend State (Lemma~\ref{lemma:gsg-cva-post-unfriend}) and eventual delivery of the $\mathsf{unfriend}$, $\mathrm{epoch}_p(q)=\mathrm{epoch}_q(p)=e'$ for a common even value $e'$ from some time on. Let $r$ be a sustained mutual friend of $p$ (the case of $q$ is symmetric, using the stream of $q$). By Friendship Establishment (Theorem~\ref{theorem:gsg-cva-establishment}) applied to $(p,r)$, eventually $r\in\mathrm{Rec}_p$ stably. The transaction that set $\mathrm{epoch}_p(q):=e'$ emits $\mathsf{stream\_update}(q,e')$ to $\mathrm{Rec}_p$; if $r\in\mathrm{Rec}_p$ then, $r$ receives it; otherwise $r$ joins afterwards and the on-join snapshot from $p$ carries $\mathsf{stream\_update}(q,\mathrm{epoch}_p(q))=\mathsf{stream\_update}(q,e')$. Either way, by eventual delivery $r$ processes $\mathsf{stream\_update}(q,e')$ from $p$ with $p\in\mathrm{dom}(\mathrm{FMap}_r)$, and Integrate stream update sets $\mathrm{epoch}_r(p,q):=e'$. By Friendship Monotonicity (Lemma~\ref{lemma:gsg-friendship-mono}), $\mathrm{epoch}_r(p,q)\ge e'$ thereafter; by Observer Bound (Lemma~\ref{lemma:gsg-observer-bound}), $\mathrm{epoch}_r(p,q)\le\max(\mathrm{epoch}_p(q),\mathrm{epoch}_q(p))=e'$. Hence $\mathrm{epoch}_r(p,q)=e'$, even, so $(p,q)$ does not appear in the friend-of-friend view at $r$.

Take $t^*$ to be a time past all the events above.
\end{proof}

%% file: sections/proofs-secure-gsg-cva-properties.tex
\ThmHalfMutualityCva*
\begin{proof}
By Friendship Monotonicity (Lemma~\ref{lemma:gsg-friendship-mono}),
$\mathrm{epoch}_p(q)$ is non-decreasing; let $T$ be the last transition up to $t$
that set $\mathrm{epoch}_p(q):=x$, so $\mathrm{epoch}_p(q)=x$ throughout $[T,t]$.

\emph{Replaced $q$.}  Suppose $q$ ran Replace, to $q'$, at some time $\le t$.  By
this hypothesis $q$'s cascade reaches $p$, so a
$\mathsf{new\_identity}(q,q',\cdot)$ is addressed to $p$.  Had $p$ executed
Integrate new identity on it, $\mathrm{FMap}_p[q]$ would have been removed
(renamed to $q'$), contradicting $\mathrm{epoch}_p(q)=x$ at $t$; so $p$ has not,
and the message is in transit to $p$ or lies in $i_p$---case~(3).  Henceforth
assume $q$ is not replaced by $t$; then $\mathrm{epoch}_q(p)$ is defined and
non-decreasing.

The transitions writing an odd value to $\mathrm{epoch}_p(q)$ in this fragment
are Accept offer / Resolve simultaneous offer, Integrate accept, Integrate new
identity (renaming some $o$ to $q$ at $p$), and Integrate rebind.  We case on $T$.

\emph{Accept offer / Resolve simultaneous offer.}  $T$ fires on
$\mathsf{friend\_request}(x,\cdot)$ from $q$ and adds $\mathsf{accept}(x,\cdot)$ to
$o_p$ toward $q$ in the same step.  If that $\mathsf{accept}$ is unprocessed it is
in transit with epoch $x$---case~(2).  Otherwise $q$ processed it by Integrate
accept, setting $\mathrm{epoch}_q(p):=x$.  As $\mathrm{epoch}_p(q)=x$ is odd at
$t$, $p$ integrated no $\mathsf{unfriend}$ from $q$ after $T$; so either no End
friendship between them occurred and $\mathrm{epoch}_q(p)=x$ (case~1), or one did
and Post-Unfriend State (Lemma~\ref{lemma:gsg-cva-post-unfriend}) gives
$\mathrm{epoch}_q(p)\in\{x,x+1\}$ (case~1) unless the $\mathsf{unfriend}(x+1)$ is
still in transit to $p$ (case~2).

\emph{Integrate accept.}  $T$ fires on $\mathsf{accept}(x,\cdot)$ from $q$, which
$q$ sent only after setting $\mathrm{epoch}_q(p):=x$ (Message
Bounds~(2), Lemma~\ref{lemma:gsg-message-bounds}); so $\mathrm{epoch}_q(p)\ge x$,
and the End-friendship analysis above gives case~(1) or case~(2).

\emph{Integrate new identity (rename $o\to q$ at $p$).}  $p$ renamed a replaced
friend $o$ to its new identity $q$, setting
$\mathrm{epoch}_p(q):=\mathrm{epoch}_p(o)=x$ and adding
$\mathsf{rebind}(x,\cdot,\cdot,\cdot)$ to $o_p$ toward $q$ in the same step.  If
unprocessed, that $\mathsf{rebind}$ is in transit with epoch $x$---case~(2);
otherwise $q$ (the recovered identity) processed it by Integrate rebind, setting
$\mathrm{epoch}_q(p):=x$---case~(1).

\emph{Integrate rebind.}  Here $p$ is a recovered identity; $T$ fires on
$\mathsf{rebind}(x,\cdot,\cdot,\cdot)$ from $q$, which $q$ sent (in its Integrate
new identity) only after setting $\mathrm{epoch}_q(p):=x$.  Hence
$\mathrm{epoch}_q(p)\ge x$, giving case~(1) absent further activity or case~(2)
with a later message in transit.

In every case one of (1)--(3) holds.
\end{proof}

\ThmSecureCorrespondence*
\begin{proof}
The configuration is read against the abstract friend set $\widetilde F$ (Definition~\ref{definition:secure-gsg-cva-abstract-friend-set}), each live agent's friendships as recorded in its own state, with each replaced identity relabelled to the original it replaced.

\emph{(1), (2).}  A completed Befriend handshake (Offer friendship, Accept offer or Resolve simultaneous offer, Integrate accept, with the emitted stream updates delivered) leaves $\mathrm{epoch}_p(q)=\mathrm{epoch}_q(p)$ at a common odd value: at quiescence, with neither party replaced, Half Friendship Mutuality (Lemma~\ref{lemma:secure-gsg-cva-half-mutuality}, case~1 with no message in transit) gives $\mathrm{epoch}_q(p)\in\{\mathrm{epoch}_p(q),\mathrm{epoch}_p(q)+1\}$, and the completed handshake with no intervening unfriend forces equality; each records the other active---$q\in\widetilde F_p$ and $p\in\widetilde F_q$, the effect of abstract Befriend$(p,q)$ on each side's friend set.  A completed Unfriend (End friendship, Integrate unfriend) advances both epochs to the next even value, removing $q$ from $\widetilde F_p$ and $p$ from $\widetilde F_q$---abstract Unfriend$(p,q)$.  The identity-record fields, set once and preserved, are inert to the mapping.

\emph{(3).}  By Replace Convergence (Theorem~\ref{theorem:secure-gsg-cva-replace-convergence}), on a run quiescent at $p$ before Replace the cascade renames every $r\in F_p$ from $p$ to $p'$, yields $\widetilde F_{p'}=F_p$, and leaves no $\mathrm{FMap}_r$ retaining $p$.  This is exactly abstract Replace$(p,p')$: $F'_{p'}:=F_p$, $F'_p:=\emptyset$, and each $r\in F_p$ substitutes $p'$ for $p$.  More generally, on any friendship-preserving run (Definition~\ref{def:friendship-preserving}) the cascade reaches every recoverable friend, and abstract Replace holds on each such edge.  The one unrecoverable friendship---recorded, at the fault, only by the two friends, so that no $\mathsf{stream\_update}$ carried it to another agent---is not recovered: the fresh identity $p'$ holds no record of $w$, so $w\notin\widetilde F_{p'}$, and $w$ retains a stale record of the retired key.  The specification does not require its recovery (Definition~\ref{def:friendship-preserving}; the Remark in Section~\ref{section:secure-gsg-cva-convergence}).

\emph{(4).}  After a state loss at $p$, Restore (Section~\ref{section:secure-gsg-cva-restore}) is carried entirely by Integrate checkpoint: each friend $r$ that still records $p$ keeps $p\in\mathrm{Rec}_r$ and re-broadcasts $\mathsf{checkpoint}(\mathit{IR}_r,L_r)$ with the pair $(p,\mathrm{epoch}_r(p))\in L_r$; on the first such checkpoint, $p$ installs the skeleton $\mathrm{FMap}_p[r]$ from the carried $\mathit{IR}_r$ and the direct-heal clause sets $\mathrm{epoch}_p(r):=\mathrm{epoch}_r(p)$.  A carried epoch is odd only when $r$ records the $r$--$p$ edge active, which by Message Bounds (Lemma~\ref{lemma:gsg-message-bounds}) holds only after an offer or acceptance by $p$ at that epoch; so each recovered active $r$ has $p\in F_r$ at the recovering state, and the edge is mutual once $p$ re-records it.  At quiescence every such $r$ has re-broadcast to $p$ (Section~\ref{section:secure-gsg-cva-rebroadcast}), so $\widetilde F_p$ becomes exactly the set of friendships recorded at $p$ before the state loss, the value the abstract Recover$(p,q)$ assigns from any one custodian $q$ that holds the friend-of-friend record for $p$ (Section~\ref{section:secure-gsg-abstract}).  Passive recovery is complete and needs no custodian authorisation, the asymmetry with Replace noted in Section~\ref{section:secure-gsg-cva-restore}.
\end{proof}

\ThmReplaceConvergence*
\begin{proof}
Fix $w\in F_p$, so $\mathrm{epoch}_p(w)$ is odd.

\emph{Coverage.}  When $p$ last set $\mathrm{epoch}_p(w)$ it emitted $\mathsf{stream\_update}(w,\mathrm{epoch}_p(w))$ to every $r\in\mathrm{Rec}_p=\mathrm{dom}(\mathrm{FMap}_p)$.  By quiescence at $p$ before Replace, every custodian $c\in K_p$ has integrated it, so $w\in\mathrm{dom}(\mathrm{FoFMap}_c[p])$ and $(w,\cdot)\in L_c$.  Hence $w$ lies in the known-friend set $N$ at the first firing of Announce new identity.  As $p'$ discovers each named friend (Discover is always enabled at $p'$; Section~\ref{sec:cva}), eventually $w\in\mathit{known}_{p'}$, and Announce sends $\mathsf{message}(p',w,\mathsf{new\_identity}(p,p',K_p,\mathit{IR}_{p'}))$.

\emph{Rename and rebind.}  By Half Friendship Mutuality at quiescence (Lemma~\ref{lemma:secure-gsg-cva-half-mutuality}, case~1, with no message in transit) $\mathrm{epoch}_w(p)=\mathrm{epoch}_p(w)$, so $p\in\mathrm{dom}(\mathrm{FMap}_w)$ with $\mathit{IR}_w[p]=(K_p,\sigma_p)$.  On delivery, $w$ executes Integrate new identity: the vouching set $V=K_p$ satisfies $V\subseteq K_p$ and $|V|=|K_p|\ge\lceil\sigma_p|K_p|\rceil$, so the guard passes; $w$ renames $p\to p'$ in $\mathrm{FMap}_w$ and $\mathrm{FoFMap}_w$, carrying the epoch $e_w:=\mathrm{epoch}_w(p)$, and sends $\mathsf{rebind}(e_w,\mathit{IR}_w,L_w,L_w^p)$ to $p'$.  By reliable delivery $p'$ processes it through Integrate rebind, setting $\mathrm{FMap}_{p'}[w]:=(e_w,\mathit{IR}_w)$, a full entry; since $e_w=\mathrm{epoch}_p(w)$ is odd, $w\in\widetilde F_{p'}$.

\emph{Outcome.}  As $w\in F_p$ was arbitrary, every friend of $p$ is renamed and rebinds, so $\widetilde F_{p'}=F_p$; and each such $r$ executed Integrate new identity, after which $p\notin\mathrm{dom}(\mathrm{FMap}_r)$, so no $\mathrm{FMap}_r$ retains $p$.  This is abstract Replace$(p,p')$ exactly.
\end{proof}

\ThmCvaResilient*
\begin{proof}
Befriend, Unfriend, and Replace protocols realise the corresponding abstract transitions at quiescence by the Correspondence (Theorem~\ref{theorem:secure-gsg-cva-refinement}).  A state-loss fault empties $\mathrm{FMap}_p$, so momentarily $\widetilde F_p=\emptyset$; on a friendship-preserving run Restore is complete (Section~\ref{section:secure-gsg-cva-restore}, Lemma~\ref{lemma:secure-gsg-cva-repair}), so at the next quiescence $\widetilde F_p$ is again the set of friendships $p$ recorded before the fault.  The fault and its Restore thus net to no change in $\widetilde F_p$, with the intermediate states non-quiescent, so observed at quiescence the round-trip is a stutter.  Hence the quiescent image of a live friendship-preserving run over $T''\cup F''$ uses only abstract Befriend, Unfriend, and Replace, and is a correct run of the abstract secure social graph; that is, $\sigma''$ is $F''$-resilient.
\end{proof}

\ThmIrrecoverable*
\begin{proof}
Every transaction that installs a name in $\mathrm{FMap}_x$ or $\mathrm{FoFMap}_x$ other than Accept offer, Resolve simultaneous offer and Integrate accept---which the excluded volition rules out---reads that name from an incoming message, and every message carries only names drawn from its sender's own maps.  By unrecoverability no agent other than $p$ and $w$ records the edge at $t$, and the fault erases it at the endpoint that lost its state or its identity.  Hence no message from $t$ onward carries $w$ to $p$'s side, or $p$'s successor to $w$: under state loss neither endpoint holds the other, and under identity loss $w$ is named in no $L_c$ and in no $L_p$, so Announce new identity never addresses $w$ and no rebind reveals $w$ to $p'$.  By induction on the run neither ever records the other.
\end{proof}

\ThmRepairAfterFaults*
\begin{proof}
A fault erases only the faulting agent's own state and the messages in its boxes, so the only staleness it introduces is at an agent that missed a message the faulting agent had sent or was to receive.  Faults are finite in number and none occurs after $t_F$, so finitely many such divergences arise, all by $t_F$.

\emph{Direct.}  Suppose $\mathrm{epoch}_p(q)>\mathrm{epoch}_q(p)$, the effect of a lost $\mathsf{accept}$ or $\mathsf{unfriend}$; a lost $\mathsf{friend\_request}$ writes no entry and leaves no divergence.  Re-broadcast is unguarded and enabled whenever $\mathrm{Rec}_p\ne\emptyset$, so $p$ re-broadcasts infinitely often, and by eventual delivery $q$ integrates a checkpoint carrying $(q,\mathrm{epoch}_p(q))$ in $L$; the direct-heal clause finds it above $\mathrm{epoch}_q(p)$ and adopts it (Definition~\ref{definition:secure-gsg-cva-rebroadcast}).

\emph{Observed.}  Suppose $\mathrm{epoch}_r(p,q)<\mathrm{epoch}_p(q)$, the effect of a lost $\mathsf{stream\_update}$.  Every periodic checkpoint of $p$ carries $(q,\mathrm{epoch}_p(q))$ in its snapshot, which the observer-heal clause absorbs; by the Observer Bound (Lemma~\ref{lemma:gsg-observer-bound}) it rises no higher.

A state loss is the case in which the faulting agent lost its whole state; on a friendship-preserving run its friends still record it, and their checkpoints rebuild it by the same two clauses (Section~\ref{section:secure-gsg-cva-restore}).  Each divergence is therefore repaired within finitely many re-broadcast rounds; let $t^*$ be the time by which all have been.  No divergence arises after $t_F$, so from $t^*$ every direct and observed epoch holds the value a loss-free run of the same volitions would hold, and the proofs of the three cited theorems apply verbatim with $t^*$ in place of the initial time.
\end{proof}